\documentclass[draftcls,onecolumn,11pt]{IEEEtran}    
\usepackage[T1]{fontenc}

\usepackage[utf8]{inputenc}
\usepackage{amsmath}
\usepackage{amssymb}
\usepackage{stmaryrd}
\usepackage{epic,eepic}
\usepackage{theorem}
\usepackage{pifont}  
\usepackage{euscript}
\usepackage{calc}
\usepackage{paralist}
\usepackage[titles]{tocloft}
\usepackage{acronym}

\usepackage[usenames]{color}          
\usepackage{xcolor}
\definecolor{Brown}{rgb}{0.55,0.0,0.10}
\definecolor{dgreen}{rgb}{0.00,0.56,0.00}
\definecolor{vertmoinsfonce}{rgb}{0.00,0.50,0.00}
\definecolor{vert}{rgb}{0.00,0.60,0.00}
\definecolor{llightggray}{rgb}{0.97,0.97,0.97}
\definecolor{lightggray}{rgb}{0.9,0.9,0.9}
\definecolor{ggray}{rgb}{0.5,0.5,0.5}
\definecolor{darkggray}{rgb}{0.25,0.25,0.25}
\definecolor{ddarkggray}{rgb}{0.1,0.1,0.1}
\definecolor{bleu}{rgb}{0.00,0.00,1.00}
\definecolor{darkblue}{rgb}{0,0,0.7}

\usepackage{epsf}           
\usepackage{graphicx,psfrag}       
\usepackage{epsfig}         
\usepackage{pstricks}
\usepackage{pstricks-add}
\usepackage{pst-plot}
\usepackage{dsfont}
\theoremheaderfont{\color{black}\normalfont\bfseries}

\newtheorem{lemma}{Lemma}
\newtheorem{theorem}{Theorem}[section]
\newtheorem{definition}[theorem]{Definition}
\newtheorem{corollary}[theorem]{Corollary}

\newtheorem{remark}[theorem]{Remark}

\newcommand{\mrb}[1]{\textcolor{black}{#1}}
\newcommand{\mlt}[1]{\textcolor{black}{#1}}

\newcommand{\R}{\mathbb{R}}

\newcommand{\N}{\mathbb{N}}

\newcommand{\prob}{\mathbb{P}}
\newcommand{\E}{\mathbb{E}}
\newcommand{\Q}{\mathbb{Q}}

\font\dsrom=dsrom10 scaled 1200 \def \indic{\textrm{\dsrom{1}}}
\newcommand{\UN}{\indic}

\newcommand{\C}{\mathcal{C}}

\newcommand{\QQ}{\mathcal{Q}}

\newcommand{\W}{\mathcal{W}}

\newcommand{\PP}{\mathcal{P}}

\newcommand{\mc}{\mathcal}

{\Large\normalsize}
{\Large\normalsize}

\acrodef{DMC}[DMC]{Discrete Memoryless Channel}
\acrodef{KL}[KL]{Kullback-Leibler}

\begin{document}

\title{State Leakage and Coordination with Causal State Knowledge at the Encoder}


\author{Ma\"{e}l~Le~Treust,~\IEEEmembership{Member,~IEEE} and
Matthieu~R.~Bloch,~\IEEEmembership{Senior Member,~IEEE}
\thanks{
Manuscript received December 17, 2018; revised October 21, 2020; accepted October 29, 2020. Ma\"{e}l~Le~Treust acknowledges financial support of INS2I CNRS for projects JCJC CoReDe 2015, PEPS StrategicCoo 2016, BLANC CoS 2019; DIM-RFSI under grant EX032965; CY Advanced studies and The Paris Seine Initiative 2018 and 2020. This research has been conducted as part of the Labex MME-DII (ANR11-LBX-0023-01). Matthieu~R.~Bloch acknowledges financial support of National Science Foundation under award CCF 1320304. The authors gratefully acknowledge the financial support of SRV ENSEA for visits at Georgia Tech Atlanta in 2014 and at ETIS in Cergy in 2017. This work was presented in part at the IEEE International Symposium on Information Theory (ISIT), Barcelona, Spain, in July 2016 \cite{LeTreustBloch(ISIT)16}. \textit{(Corresponding author: Ma\"{e}l~Le~Treust.)}

M.~Le~Treust is with ETIS UMR 8051, CY Cergy Paris Universit\'e, ENSEA, CNRS, 6 avenue du Ponceau, 95014 Cergy-Pontoise CEDEX, France
(e-mail: mael.le-treust@ensea.fr).

M.~R.~Bloch is with School of Electrical and Computer Engineering, Georgia Institute of Technology, Atlanta, Georgia 30332
(e-mail: matthieu.bloch@ece.gatech.edu).

Communicated by N. Merhav, Associate Editor for Shannon Theory.

  }}

%

\maketitle

\begin{abstract}


We revisit the problems of state masking and state amplification through the lens of empirical coordination. 
Specifically, we characterize the rate-equivocation-coordination trade-offs regions of a state-dependent channel in which the encoder has causal and strictly causal state knowledge. We also extend this characterization to the cases of two-sided state information and noisy channel feedback. Our approach is based on the notion of core of the receiver's knowledge, \mrb{which we introduce} to capture what the decoder can infer about all the signals involved in the model. \mrb{Finally, we exploit the aforementioned results} to solve a channel state estimation zero-sum game in which the encoder prevents the decoder to estimate the channel state accurately.

\end{abstract}


\IEEEpeerreviewmaketitle


\section{Introduction}\label{sec:Introduction}

The study of state-dependent channels can be traced back to the early works of Shannon~\cite{Shannon58} \mrb{and} Gelf'and and Pinsker~\cite{gelfand-it-1980}, which identified optimal coding strategies to transmit reliably in the presence of a state known at the encoder causally or non-causally, respectively. The insights derived from the models have since proved central to the study of diverse topics including wireless communications~\cite{Costa1983,Keshet2008}, information-hiding and watermarking~\cite{Moulin2003}, and information transmission in repeated games~\cite{GossnerHernandezNeyman06}. The present work relates to the latter application and studies state-dependent channels with causal state knowledge from the perspective of \emph{empirical coordination}~\cite{CuffPermuterCover10}.

Previous studies that have explored the problem of not only decoding messages at the receiver but also estimating the channel state\mlt{,} are particularly relevant to the present work. 
\begin{figure}[!ht]
\begin{center}
 \psset{xunit=0.9cm,yunit=0.9cm}
 \begin{pspicture}(0,0)(9,2.2)
 \pscircle(0,1.5){0.445}
 \pscircle(0,0.5){0.445}
 \psframe(2,0)(3,1)
 \pscircle(5,0.5){0.45}
 \psframe(7,0)(8,1)
 \psline[linewidth=1pt]{->}(0.5,0.5)(2,0.5)
 \psline[linewidth=1pt]{->}(3,0.5)(4.5,0.5)
 \psline[linewidth=1pt]{->}(5.5,0.5)(7,0.5)
 \psline[linewidth=1pt]{->}(8,0.5)(9.5,0.5)
 \psline[linewidth=1pt]{->}(0.5,1.5)(5,1.5)(5,1)
 \psline[linewidth=1pt]{->}(2.5,1.5)(2.5,1)
 \rput[u](1,0.8){$M$}
 \rput[u](1,1.8){$S^{i}$}
 \rput[u](3.75,0.8){$X_i $}
 \rput[u](6.25,0.8){$Y^n$}
 \rput[u](8.85,0.8){$(\hat{M}, V^n)$}
 \rput(0,0.5){$\PP_{M}$}
 \rput(0,1.5){$\PP_S$}
 \rput(2.5,0.5){$f_i$}
 \rput(5,0.5){$\mc{T}$}
 \rput(7.5,0.5){$g,h$}
 \end{pspicture}
\caption{\mrb{The memoryless channel $\mc{T}_{Y|XS}$ depends on the state drawn i.i.d. according to $\PP_S$.} The encoding function is causal  $f_i: \mc{M} \times \mc{S}^{i} \rightarrow \mc{X}$, for all $i\in\{1,\ldots,n\}$ and the decoding functions $g: \mc{Y}^n  \rightarrow \mc{M} $ and  $h: \mc{Y}^n  \rightarrow \Delta(\mc{V}^n)$ are non-causal.}
\label{fig:CE}
\end{center}
\end{figure}
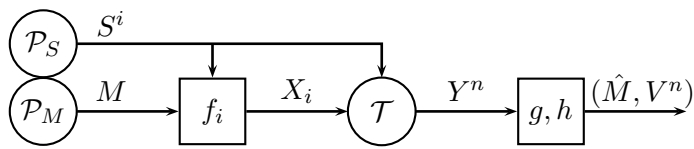
The \emph{state masking} formulation of the problem~\cite{MerhavShamai(StateMasking)07} aims at characterizing the trade-off between the rate of reliable communication and the minimal leakage about the channel state. The rate-leakage capacity region of state masking has been successfully characterized for both causal and non-causal state knowledge. The \emph{state amplification} formulation~\cite{KimSutivongCover(StateAmplification)08}, in which the state is conveyed to the receiver instead of being masked, aims at characterizing the trade-off between the rate of reliable communication and the reduction of uncertainty about the state. The rate-uncertainty reduction capacity region of state amplification has also been successfully characterized for causal and non-causal state knowledge. The state amplification formulation was subsequently extended in the causal case by replacing the reduction of uncertainty about the state by an average distortion function~\cite{ChoudhuriKimMitra13} (this model was dubbed causal \emph{state communication}). \mrb{Note that, in such a scenario, the channel output feedback at the encoder increases the region of achievable rate-distortion pairs~\cite{BrossLapidoth18}.} The rate-distortion capacity region of state communication has been successfully characterized for causal and strictly causal state knowledge, and has been characterized for noiseless and noisy non-causal state knowledge in the case of Gaussian channels with a quadratic distortion~\cite{SutivongChiangCoverKim05,Tian15}. Both formulations have been combined in~\cite{Koyluoglu2016} to study the trade-off between amplification and leakage rates in a channel with two receivers having opposite objectives. The amplification-leakage capacity region has been investigated for non-causal state knowledge, via generally non-matching inner and outer bounds. As a perhaps more concrete example,~\cite{TutuncuogluITW14} has studied the trade-off between amplification and leakage in the context of an energy harvesting scenario. An extreme situation of state masking, called state obfuscation, in which the objective is to make the channel output \mlt{sequence nearly} independent of the channel state\mlt{s}, has recently been investigated in~\cite{Wang2020}.

We revisit \mrb{here}, the problems of state masking and state amplification with causal and strictly causal state knowledge through the lens of \emph{empirical coordination}~\cite{CuffPermuterCover10}, \cite{KramerSavari07}. Empirical coordination refers to the control of the joint histograms of the various sequences such as states, codewords, that appear in channel models, and is related to the coordination of autonomous decision makers in game theory~\cite{GossnerHernandezNeyman06}. Specifically, the study of empirical coordination over state-dependent channels is a proxy for characterizing the utility of autonomous decision makers playing a repeated game in the presence of \mrb{an} environment variable (the state), random~\cite{GossnerHernandezNeyman06,GossnerVieille02} or adversarial~\cite{GossnerTomala06,GossnerTomala07,GossnerLarakiTomala09}, and of an observation structure (the channel) describing how agents observe each other's actions. The characterization of the empirical coordination capacity requires the design of coding schemes in which the actions of the decision makers are sequences that embed coordination information. The empirical coordination capacity has been studied for state-dependent channels under different constraints including strictly causal and causal encoding~\cite{CuffSchieler11}, for perfect channel~\cite{Cuff(ImplicitCoordination)11}, for strictly causal and causal decoding~\cite{LeTreust(EmpiricalCoordination)17}, with source feedforward~\cite{LarrousseLasaulceBloch(IT)18}, for lossless decoding~\cite{LeTreust(CorrelationITW)14}, with secrecy constraint~\cite{SchielerCuff(RateDistortion14)}, with two-sided state information~\cite{LeTreust(ISIT-TwoSided)15} and with channel feedback~\cite{LeTreust(ISITfeedbacks)15}. Empirical coordination is \mrb{also} a powerful tool for controlling the \mrb{Bayesian posterior beliefs} of the decoder, e.g. in the problems of Bayesian persuasion~\cite{LeTreustTomala19} and strategic communication~\cite{LeTreustTomala_IT19}. 

The main contribution of the present work is to show that empirical coordination provides a natural framework in which to jointly study the problems of reliable communication, state masking, and state amplification.  \mrb{This connection highlights some of the benefits of empirical coordination beyond those already highlighted in earlier works~\cite{CuffSchieler11}--\cite{LeTreust(ISITfeedbacks)15}.} In particular, we obtain the following.
\begin{itemize}
\item We introduce and characterize the notion of \emph{core of the receiver's knowledge}, which captures what the decoder can exactly know about the other variables involved in the system. For instance, this allows us to characterize the rate-leakage-coordination region for the causal state-dependent channel (Theorem~\ref{theo:LeakageCE}). Our definition of leakage refines previous work by exactly characterizing the leakage rate instead of only providing a single-sided bound. When specialized, our result (Theorem~\ref{theo:LeakageCEd}) simultaneously recovers the constraints already established both in~\cite[Section V]{MerhavShamai(StateMasking)07} and~\cite[Theorem 2]{KimSutivongCover(StateAmplification)08}.
\item We revisit the problem of causal state communication and characterize the normalized \ac{KL}-divergence between the decoder's posterior beliefs and a target belief induced by coordination (Theorem~\ref{theo:Estimation}). This allows us to characterize the rate-distortion trade-off for a zero-sum game, in which the decoder attempts to estimate the state while the encoder tries to mask it (Theorem~\ref{theo:DecoderEstimation}). 
\item We extend the results to other models, including two-sided state information (Theorem~\ref{theo:LeakageCE2SI}), noisy feedback (Theorem~\ref{theo:LeakageCENoisyFeedback}), and strictly causal encoding (Theorem~\ref{theo:LeakageSCE}).
\end{itemize}

The rest of the paper is organized as follows. In Section~\ref{sec:ProblemFormulation}, we formally introduce the model, along with necessary definitions and notation, and we state our main results. In Section~\ref{sec:ChannelStateEstimation}, we investigate the channel state estimation problem by introducing the \ac{KL}-divergence and the decoder's posterior beliefs. In Section~\ref{sec:Extensions} and Section~\ref{sec:StrictlyCausal}, we present some extensions of our results to different scenarios. The proofs \mrb{of most results} are \mrb{provided} in Appendices~\ref{sec:AchievabilityLeakageCE}-\ref{sec:ConverseTheoDecoderEstimation}, \mrb{with some details relegated to Supplementary Materials.}

%


%
%
%
%

\section{Problem formulation and main result}\label{sec:ProblemFormulation}

\subsection{Notation}
\label{sec:notation}

Throughout the paper, capital letters, e.g., $S$, denote random variables while lowercase letters, e.g., $s$. denote their realizations and calligraphic fonts, e.g., $\mc{S}$, denote the alphabets in which the realizations take values. All alphabets considered in the paper are assumed finite, i.e., $\vert\mc{S}\vert<\infty$. Sequences of random variables and realizations are denoted by $S^n=(S_1,\dots,S_n)$ and $s^n=(s_1,\ldots,s_n)$, respectively. We denote the set of probability distributions over $\mc{S}$ by $\Delta(\mc{S})$. For a distribution $\QQ_S\in\Delta(\mc{S})$, we drop the subscript and simply write $\QQ(s)$ in place of $\QQ_S(s)$ for the probability mass assigned to realization $s\in\mc{S}$. \mrb{The notation $\QQ_{X}(\cdot|y)\in \Delta(\mc{X})$ denotes the conditional distribution of $X\in\mc{X}$, given the realization $y\in \mc{Y}$.} For two distributions $\QQ_X,\PP_X\in\Delta(\mc{X})$, $||\QQ_X - \PP_X||_{1}= \sum_{x\in \mc{X}} | \QQ(x) - \PP(x)|$ stands for the $\ell_1$-distance between the vectors \mrb{representing} the distributions,
see also \cite[pp. 370]{cover-book-2006} and \cite[pp. 19]{CsiszarKorner(Book)11}. \mrb{We write} $Y  -\!\!\!\!\minuso\!\!\!\!-X    -\!\!\!\!\minuso\!\!\!\!-  W$ \mrb{when $Y$, $X$, and $W$ form a Markov chain in that order.} The notation $\UN(v=s)$ stands for the indicator function, which is equal to 1 if $v=s$ and 0 otherwise.

For a sequence $s^n\in\mc{S}^n$, $\textsf{N}(\mrb{ \tilde{s}}|s^n)$ denotes the occurrence number of symbol $\mrb{ \tilde{s}} \in \mc{S}$ in the sequence $s^n$. The empirical distribution ${Q}_S^n\in\Delta(\mc{S})$ of sequence $s^n\in\mc{S}^n$
is then defined as
\begin{align}
 \forall  \mrb{ \tilde{s}}\in \mc{S}, \quad{Q}^n( \mrb{ \tilde{s}}) = \; \frac{\textsf{N}(\mrb{ \tilde{s}}|s^n)}{n}.
\label{eq:EmpiricalDistribution}
\end{align}
Given $\delta>0$ and a distribution $\QQ_{SX}\in\Delta(\mc{S}\times\mc{X})$, $T_{\delta}(\QQ_{SX})$ stands for the set of sequences $(s^n,x^n)$ that are jointly typical with tolerance $\delta>0$ with respect to the distribution $\QQ_{SX}$, i.e., such that 
\begin{align}
\Big|\Big|{Q}^n_{SX}  - \QQ_{SX} \Big|\Big|_1 = \sum_{s,x}\Big|{Q}^n( s, x) - \QQ(s,x)\Big| \leq \delta.
\end{align}
\mrb{We denote by $\prob\big(S^n\in T_{\delta}(\PP_{S})\big)$ the probability value assigned to the event $\big\{S^n\in T_{\delta}(\PP_{S})\big\}$, according to the distribution of $S^n$.}

\subsection{System model}\label{sec:ModelDefinition}

The problem under investigation is illustrated in Figure~\ref{fig:CE}. A uniformly distributed message represented by the random variable $M \in \mc{M}$ is to be transmitted over a state dependent memoryless channel characterized by the conditional distribution $\mc{T}_{Y|XS}$ and a channel state \mrb{$S\in\mc{S}$} drawn according to the i.i.d. distribution $\PP_S$. For $n\in\mathbb{N}^{\star}=\mathbb{N}\setminus\{0\}$, the message $M$ and the state sequence $S^n$ are encoded into a codeword $X^n\in\mc{X}^n$ using an encoder, subject to \mrb{causal} constraints to be precised later. Upon observing the output $Y^n\in\mc{Y}^n$ of the noisy channel, the receiver uses a decoder to form an estimate $\hat{M} \in \mc{M}$ of $M$ and to generate actions $V^n\in\mc{V}^n$, whose exact role will be precised shortly. For now, $V^n$ can be thought of as an estimate of the state sequence $S^n$ but more generally captures the ability of the receiver to coordinate with the transmitter \mrb{and the channel state}.  Both $\mc{T}_{Y|XS}$ and $\PP_S$ are assumed known to all parties.
We are specifically interested in causal encoders formally defined as follows.
\begin{definition}
  \label{def:CodeLeakageCE}
  \mrb{A code with causal encoding consists of stochastic encoding functions $f_i : \mc{M} \times \mc{S}^{i}  \longrightarrow \Delta(\mc{X})$ $\forall i \in \{1, \ldots,n\}$, a deterministic decoding function $g : \mc{Y}^n  \longrightarrow  \mc{M}$, and a stochastic receiver action function $h : \mc{Y}^n  \longrightarrow  \Delta(\mc{V}^n)$. The set of codes with causal encoding with length $n$ and message set $\mc{M}$ is denoted $\mc{C}_{\sf{c}}(n,\mc{M})$.}
\end{definition}

A code $c \in \mc{C}_{\sf{c}}(n,\mc{M})$, the uniform distribution of the messages $\PP_{M}$, the source $\PP_S$ and the channel $\mc{T}_{Y|XS}$, induce a distribution on $(M,S^n,X^n,Y^n,V^n,\hat{M})$ \mrb{given by}
\begin{align}
\PP_{M}   \prod_{i=1}^n\bigg[ \PP_{S_i}  f_{X_i|S^iM}  \mc{T}_{Y_i|X_iS_i}\bigg] h_{V^n|Y^n}   \UN \big( \hat{M} = g(Y^n)\big) .\label{eq:GeneralDistribution}
\end{align}
Since the sequences $(S^n,X^n,Y^n,V^n)$ are random, the empirical distribution $Q^n_{SXYV}$ is also a random variable. The performance of codes is measured as follows.
\begin{definition}\label{def:achievability}
Fix a target rate $\textsf{R}\geq 0$, a target state leakage $\textsf{E}\geq 0$ and a target distribution $\QQ_{SXYV}$. The triple $(\textsf{R},\textsf{E},\mrb{\QQ_{SXYV}} )$ is achievable if for all $\varepsilon>0$, there exists  $\bar{n}\in \N^{\star}$, for all $n \geq \bar{n}$, there exists a code \mrb{$c\in\mc{C}_{\sf{c}}(n,\mc{M})$} that satisfies
\begin{align*}
\frac{\log_2 |\mc{M}|}{n}  \geq& \textsf{R} - \varepsilon,\\
 \bigg| \mc{L}_{\textsf{e}}(c)  - \textsf{E} \; \bigg|  \leq&\varepsilon, \qquad \text{\normalsize with } \qquad \mc{L}_{\textsf{e}}(c) = \frac{1}{n}  I(S^n;Y^n)  ,\\
 \PP_{\textsf{e}}(c)  =& \prob\bigg( M \neq \hat{M} \bigg)\nonumber\\
&+\prob\bigg(\Big|\Big|Q_{SXYV}^n - \QQ_{SXYV} \Big|\Big|_{1}> \varepsilon\bigg) \leq \varepsilon.  \end{align*}
We denote by \mrb{$\mc{A}_{\sf{c}}$} the set of achievable triples $(\textsf{R},\textsf{E},\mrb{\QQ_{SXYV}} )$.
\end{definition}

In layman's term, performance is captured along three metrics:
\begin{inparaenum}[i)]
  \item the rate at which the message $M$ can be reliably transmitted;
  \item the information leakage rate about the state sequence $S^n$ at the receiver; and 
  \item the ability of the encoder to coordinate with the receiver, captured by the empirical coordination with respect to $Q_{SXYV}$. 
\end{inparaenum}
\mrb{The need to coordinate with receiver action $V$ is motivated by problems in which the terminals represent decision makers that choose actions $(X,V)$ as a function of the system state $S$, as in~\cite{GossnerHernandezNeyman06}. The state can also be used to represent a system to control, in which case coordination also ties to the Witsenhausen's counterexample~\cite{Grover2015,Oechtering2019}.}


\begin{figure}[!ht]
\begin{center}
 \psset{xunit=0.5cm,yunit=0.4cm}
 \begin{pspicture}(-1,-0.5)(7.5,9.5)
 \psline[linewidth=1pt]{->}(0,-1)(0,10)
 \psline[linewidth=1pt]{->}(-1,0)(11,0)
 \pspolygon[fillstyle=vlines*,linecolor=blue,hatchcolor = blue](0,2)(0,7)(1,7)(8,2)
 \psline[linewidth=0.5pt,linestyle = dotted]{-}(0,7)(10,7)
 \psline[linewidth=0.5pt,linestyle = dotted]{-}(0,8)(10,8)
 \psline[linewidth=0.5pt,linestyle = dotted]{-}(0,2)(10,2)
 \psline[linewidth=0.5pt,linestyle = dotted]{-}(8,0)(8,9.5)
 \psdots(0,2)(0,8)(8,0)(0,0)(0,7)
 \rput[r](-0.2,8){$I(S,W_1 ;Y)$}
 \rput[r](-0.2,7){$H(S)$}
 \rput[r](-0.2,2){$I(S;Y,W_1,W_2)$}
 \rput[b](5.5,-1.2){$I(W_1,W_2;Y) - I(W_2;S|W_1)$}
 \rput[b](11,-0.9){$\textsf{R}$}
 \rput[r](-0.2,9.5){$\textsf{E}$}
 \rput[b](-0.5,-0.75){$0$}
 \rput[r](7,8.7){$ H(S)< I(S,W_1;Y)$}
 \end{pspicture}
\label{fig:KimSutivongCover08}
\end{center}
\caption{\mrb{The region of achievable $(\textsf{R},\textsf{E}  ) \in\mc{A}_{\sf{c}}$ for a given distribution $\QQ_{SW_1W_2XYV}$ for which $H(S)<I(S,W_1;Y)$.}}\label{fig:RegionRE}
\end{figure}
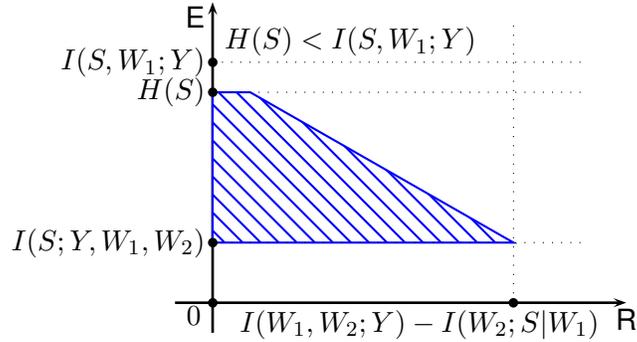

\subsection{Main result}\label{sec:MainResult}

\begin{theorem}\label{theo:LeakageCE}
Consider a target distribution $\QQ_{SXYV}$ that decomposes as $\QQ_{SXYV} = \PP_S     \QQ_{X|S} \mc{T}_{Y|XS}      \QQ_{V|SXY}$. \mrb{Then, $(\textsf{R},\textsf{E},\QQ_{SXYV} )\in\mc{A}_{_{\sf{c}}}$} if and only if there exist two auxiliary random variables $(W_1,W_2)$ with distribution $ {\QQ}_{SW_1W_2XYV} \in \Q_{\sf{c}}$ \mrb{satisfying}
\begin{align}
I(S;W_1,W_2,Y) \leq \textsf{E}&\leq H(S), \label{eq:theoremCEachie2} \\
\textsf{R} + \textsf{E} &\leq  I(W_1,S; Y),\label{eq:theoremCEachie3}
\end{align}
where $\Q_{\sf{c}}$ is the set of distributions ${\QQ}_{SW_1W_2XYV}$ with marginal $\QQ_{SXYV}$ that decompose as 
\begin{align}
 \PP_S   \QQ_{W_1}    \QQ_{W_2|SW_1}  \QQ_{X|SW_1}
     \mc{T}_{Y|XS} \QQ_{V|YW_1W_2} , \label{eq:distribution}
\end{align}
and \mrb{such that} $\max(|\mc{W}_1|,|\mc{W}_2|) \leq  | \mc{S}\times  \mc{X}\times \mc{Y} \times \mc{V} | +1 $.
\end{theorem}

The achievability and converse proofs are provided in Appendices~\ref{sec:AchievabilityLeakageCE} and~\ref{sec:ConverseLeakageCE}, respectively, with the cardinality bounds established \mrb{in the Supplementary Materials}. \mrb{The key idea behind the achievability proof is the following. The encoder operates in a Block-Markov fashion to ensure that the transmitted signals, the state, the received sequence, and the receiver actions are coordinated subject to the causal constraint at the encoder. This requires the use of two auxiliary codebooks, captured by the auxiliary random variables $W_1$ and $W_2$, where the first codebook is used for reliable communication while the second one is used to coordinate with the state. Simultaneously, the encoder quantizes the channel state and transmits carefully chosen bin indices on top of its messages to finely control how much the receiver can infer about the channel state.} \mrb{The region of achievable pairs $(R,E)$ is depicted in Fig.~\ref{fig:RegionRE} for a given distribution $\QQ_{SW_1W_2XYV}$, assuming $H(S)<I(S,W_1;Y)$.}


\begin{remark}\label{remark:MarkovCausal}
Equation \eqref{eq:theoremCEachie3} and the first inequality of \eqref{eq:theoremCEachie2} imply the information constraints of \cite[Theorem 3]{ChoudhuriKimMitra13} for causal state communication and of \cite[Theorem 2]{CuffSchieler11} for empirical coordination. 
\begin{align}
\textsf{R} &\leq I(W_1,W_2;Y) - I(W_2;S|W_1). \label{eq:theoremCEachie1} 
\end{align}
Indeed, both Markov chains $X -\!\!\!\!\minuso\!\!\!\!- (S ,W_1) -\!\!\!\!\minuso\!\!\!\!-  W_2$ and $Y -\!\!\!\!\minuso\!\!\!\!- (X , S ) -\!\!\!\!\minuso\!\!\!\!-  ( W_1,W_2 )$ imply $Y -\!\!\!\!\minuso\!\!\!\!- (W_1 , S ) -\!\!\!\!\minuso\!\!\!\!-  W_2$. \end{remark}

Theorem \ref{theo:LeakageCE} has several important consequences. First, the coordination of both encoder and decoder actions according to  $\PP_S      \QQ_{X|S} \mc{T}_{Y|XS}      \QQ_{V|SXY}$ is compatible with the reliable transmission of additional information at rate $\textsf{R}\geq0$. Second, the case of equality in the right-hand-side inequality of \eqref{eq:theoremCEachie2} corresponds to the full \mrb{disclosure} of the channel state $S$ to the decoder. Third, \mrb{for any $(\textsf{R},\mrb{\QQ_{SXYV}})$, the minimal state leakage $\textsf{E}^{\star}(\textsf{R}, \mrb{\QQ_{SXYV}})$ such that $(\textsf{R},\textsf{E}^{\star}(\textsf{R}, \mrb{\QQ_{SXYV}}),\mrb{\QQ_{SXYV}})\in\mc{A}_{\sf{c}}$, if it exists,} is given by
\begin{align}
\textsf{E}^{\star}(\textsf{R}, \mrb{\QQ_{SXYV}}) = \min_{\QQ_{SW_1W_2XYV}  \in \Q_{\sf{c}},\atop \text{s.t.}\; \textsf{R} \leq I(W_1,W_2;Y) - I(W_2;S|W_1) }  I(S;W_1,W_2,Y) .\label{eq:MinimumLeakage}
\end{align}

The reliable transmission of information requires the decoder to know the encoding function, from which it can \mrb{estimate} the channel state $S$. 
In Section \ref{sec:ChannelStateEstimation}, we investigate the relationship between the state leakage $\mc{L}_{\textsf{e}}(c)$ and the decoder's posterior belief  $\PP_{S^n|Y^n}$ induced by the \mlt{encoding} process.


\subsection{Special case without \mrb{receiver} actions} \label{sec:RemovingReceiverAction}

We now assume that the decoder does not return an action $V$ coordinated with the other symbols $(S,X,Y)$\mlt{, in order} to compare our setting with the problems of ``state masking'' \cite[Section V]{MerhavShamai(StateMasking)07} and ``state amplification'' \cite[Section IV]{KimSutivongCover(StateAmplification)08}. Note that these earlier works involve slightly different notions of achievable \mrb{state} leakage. In \cite{MerhavShamai(StateMasking)07}, the state leakage is upper bounded by $\mc{L}_{\textsf{e}}(c) = \frac{1}{n} I(S^n;Y^n)  \leq \textsf{E} + \varepsilon$. In \cite{KimSutivongCover(StateAmplification)08}, the decoder forms a list $L_n(Y^n)\subseteq \mc{S}^n$ with cardinality $\log_2 |L_n(Y^n)| = H(S) - \textsf{E} $ such that the list decoding error \mrb{probability} $\prob(S^n\notin L_n(Y^n))\leq \varepsilon$ is small, \mrb{hence reducing} the uncertainty about the state. Here, we require the leakage $\mc{L}_{\textsf{e}}(c) =\frac{1}{n} I(S^n;Y^n)$ induced by the code to be controlled \mrb{by} $ \big| \mc{L}_{\textsf{e}}(c)  - \textsf{E} \; \big|  \leq\varepsilon$. Nevertheless, we shall see that our definition allows us to obtain the \mrb{rate constraints} of~\cite{MerhavShamai(StateMasking)07,KimSutivongCover(StateAmplification)08} as extreme cases.

\begin{definition}\label{def:CodeLeakageCEd}
\mrb{A code without receiver actions consists of stochastic encoding functions $f_i : \mc{M} \times \mc{S}^{i}  \longrightarrow \Delta(\mc{X})$ $\forall i \in \{1, \ldots,n\}$ and a deterministic decoding function $g : \mc{Y}^n  \longrightarrow  \mc{M}$. The set of such codes with length $n$ and message set $\mc{M}$ is denoted $\mc{C}_{\sf{d}}(n,\mc{M})$. The corresponding set of achievable triples $(\textsf{R},\textsf{E},\QQ_{SXY} )$ is defined as in Definition~\ref{def:achievability} and is denoted $\mc{A}_{\sf{d}}$.}
\end{definition}

Note that the target distribution is here restricted to \mrb{$\QQ_{SXY} \in \Delta(  \mc{S}   \times \mc{X}\times \mc{Y})$} \mrb{since the receiver does not take an action.}

\begin{theorem}\label{theo:LeakageCEd}
Consider a target distribution $\QQ_{SXY}$ that decomposes as $\QQ_{SXY} = \PP_S      \QQ_{X|S} \mc{T}_{Y|XS}$. Then, \mrb{$(\textsf{R},\textsf{E}, \QQ_{SXY} )\in \mc{A}_{\sf{d}}$} if and only if there exists an auxiliary random variable $W_1$ with distribution $ \QQ_{SW_1XY} \in \Q_{\sf{d}}$ that satisfies
\begin{align}
I(S;W_1,Y) \leq \textsf{E}&\leq H(S), \label{eq:theoremCEachie2noCO} \\
\textsf{R} + \textsf{E} &\leq  I(W_1,S; Y),\label{eq:theoremCEachie3noCO}
\end{align}
where $\Q_{\sf{d}}$ is the set of distributions $\QQ_{SW_1XY}$ with marginal $\QQ_{SXY}$ that decompose as
\begin{align}
  \PP_S    \QQ_{W_1}   \QQ_{X|SW_1}      \mc{T}_{Y|XS}, \label{eq:distributiond}
\end{align}
and \mrb{such that} $|\mc{W}_1| \leq  | \mc{S}\times  \mc{Y}  | +1 $.
\end{theorem}

The achievability \mrb{proof is obtained} from Theorem \ref{theo:LeakageCE} by \mrb{setting} $W_2 = \emptyset$ and \mrb{by} considering a single block coding instead of block-Markov coding. 
The converse proof \mrb{is similar to the converse of Theorem~\ref{theo:LeakageCE} and is provided in the Supplementary Materials}.

\begin{remark}\label{remark:MarkovCEd}
  \mrb{When setting $W_2 = \emptyset$, \eqref{eq:theoremCEachie1} in Remark~\ref{remark:MarkovCausal} simplifies to
\begin{align}
\textsf{R}  &\leq  I(W_1; Y), \label{eq:theoremCEachie1noCO} 
\end{align}
which, together with the first inequality in~\eqref{eq:theoremCEachie2noCO}, coincides with the information constraints of~\cite[pp. 2260]{MerhavShamai(StateMasking)07}.} Furthermore, \eqref{eq:theoremCEachie1noCO}, \eqref{eq:theoremCEachie3noCO} and \mrb{the} second inequality of \eqref{eq:theoremCEachie2noCO} \mrb{correspond} to the region $\mc{R}_0$ stated in \cite[Lemma 3]{KimSutivongCover(StateAmplification)08}. Formally, the region characterized by Theorem~\ref{theo:LeakageCEd} is the intersection of the regions identified in~\cite[pp. 2260]{MerhavShamai(StateMasking)07} and~\cite[Lemma 3]{KimSutivongCover(StateAmplification)08}. 
\end{remark}





\section{Channel state estimation via distortion function}\label{sec:ChannelStateEstimation}

\subsection{Decoder posterior belief}\label{sec:Estimation}

\mrb{In this section, we provide an upper bound on the \ac{KL}-divergence between the decoder posterior belief $\PP_{S^n|Y^n}$ induced by \mlt{an} encoding, and the target conditional distribution $\QQ_{S|YW_1W_2}$.}


\begin{theorem}[Channel state estimation]\label{theo:Estimation}
Assume that the distribution $\mrb{\QQ=}\QQ_{SW_1W_2XY}$ has full support. For any conditional distribution $\PP_{W_1^nW_2^nX^n|S^n}$, we have
 \begin{align}
   \frac{1}{n} D\bigg(\PP_{S^n|Y^n} \bigg|\bigg|& \prod_{i=1}^n \QQ_{S_i|Y_iW_{1,i}W_{2,i}} \bigg) \\
   \leq& \mrb{ \mc{L}_{\textsf{e}}(c) - I(S;W_1,W_2,Y)}   + \alpha_1\delta\nonumber\\&
   +  \alpha_2\prob\Big((S^n,W_1^n,W_2^n,Y^n) \notin T_{\delta}(\QQ) \Big), \label{eq:LeakageKL0}
 \end{align}
where $\delta>0$ denotes the tolerance of the set of typical sequences $T_{\delta}(\QQ)$ and the constants $\alpha_1 = \sum_{s,w_1,\atop w_2,y} \log_2  \frac{1}{\QQ(s|w_1,w_2,y)}$ and $\alpha_2 =\log_2 \frac{1}{\min_{s,y,w_1,w_2} \QQ(s|y,w_{1}, w_{2} )}$ are strictly positive.
\end{theorem}

The proof of Theorem \ref{theo:Estimation} is given in Appendix \ref{sec:ProofTheoremKL}. Consider a target leakage $\textsf{E} =I(S;W_1,W_2,Y)$ and a pair $( \textsf{R},\QQ_{SXYV})$, and assume there exists a distribution $\QQ_{SW_1W_2XYV}\in \Q_{\sf{c}}$ with full support, satisfying \eqref{eq:theoremCEachie2} and \eqref{eq:theoremCEachie3}. By Theorem \ref{theo:LeakageCE}, for all $\varepsilon>0$ and \mrb{for} all $\delta>0$, there exists $\mlt{\bar{n}\in\N^{\star}}$ such that for all $n \geq \bar{n}$ there exists a code $c\in\mc{C}(n,\mc{M})$ with two auxiliary sequences $(W_1^n,W_2^n)$, such that
\begin{align}
&\bigg| \mc{L}_{\textsf{e}}(c)  -  I(S;W_1,W_2,Y) \; \bigg|  \leq\varepsilon
\quad\text{and}\quad\nonumber\\&
\mrb{ \prob\bigg(\Big|\Big|Q_{SW_1W_2Y}^n - \QQ_{SW_1W_2Y} \Big|\Big|_{1}> \delta\bigg) \leq \varepsilon.}\label{eq:LeakageProperty}
\end{align}
Hence, by Theorem \ref{theo:Estimation} we have
 \begin{align}
 \frac{1}{n} D\bigg(\PP_{S^n|Y^n} \bigg|\bigg| \prod_{i=1}^n \QQ_{S_i|Y_iW_{1,i}W_{2,i}} \bigg) 
 \leq& \varepsilon +  \alpha_1\delta   +  \alpha_2\varepsilon,\label{eq:EstimationEpsilon}
\end{align}
where $\epsilon$ and $\delta$ may go to zero when $n$ goes to infinity. The control of the leakage $ \mc{L}_{\textsf{e}}(c) $ and the joint typicality of the sequences $(S^n,W_1^n,W_2^n,Y^n)\in T_{\delta}(\QQ)$ implies that the decoder posterior belief  $\PP_{S^n|Y^n}$ \mrb{approaches} the single-letter distribution $\QQ_{S|YW_1W_2}$. Based on the triple of symbols $(Y,W_1,W_2)$, the decoder generates action $V$ using  the conditional  distribution $\QQ_{V|YW_1W_2}$ and infers the channel state $S$ according to the conditional distribution $\QQ_{S|YW_1W_2}$. \mrb{In that regard,} the random variables $(Y , W_1,W_2)$ capture the "\emph{core of the receiver's knowledge}," regarding other random variables $S$ and $V$. \mrb{The bound on the \ac{KL}-divergence in~\eqref{eq:LeakageKL0} relates to the notion of strategic distance \cite[Section 5.2]{GossnerVieille02}, later used in several articles on repeated game \cite{GossnerTomala06}, \cite{GossnerTomala07},  \cite{GossnerLarakiTomala09}, on Bayesian persuasion  \cite{LeTreustTomala19} and on strategic communication \cite{LeTreustTomala_IT19}.}

%

\subsection{Channel state estimation zero-sum game} \label{sec:ZeroSumGame}

We now introduce a channel state estimation zero-sum game, in which the encoder and decoder are \mrb{opponents} choosing \mrb{an} encoding and \mrb{a} decoding strategy, \mrb{respectively}. Although \mrb{the encoder and the decoder} cooperate in transmitting reliably at rate $\textsf{R}$, \mrb{the} encoder seeks to prevent the decoder \mrb{from returning} a good estimate $v\in\mc{V}$ of the channel state $s\in\mc{S}$ \mrb{by maximizing the expected long-run distortion, while the decoder attempts to minimize it.}

\begin{definition}\label{def:AchievabilityEstimation}
A target rate $\textsf{R}\geq 0$ and a target distortion  $\textsf{D}\geq 0$ are achievable if for all $\varepsilon>0$, there exists $\bar{n}\in \N^{\star}$ such that for all $n \geq \bar{n}$, there exists a code in $\mc{C}_{\sf{d}}(n,\mc{M})$ such that
\begin{align}
\frac{\log_2 |\mc{M}|}{n}  \geq& \textsf{R} - \varepsilon,\label{eq:DefAchievabilityEstimate1}\\
\prob\bigg( M \neq \hat{M} \bigg) \leq& \varepsilon,\label{eq:DefAchievabilityEstimate2}\\
\bigg|\min_{h_{V^n|Y^n}}  \frac{1}{n} \sum_{i=1}^n\E\Big[d(S_i,V_i)\Big] - \textsf{D}\bigg| \leq &\varepsilon.\label{eq:DefAchievabilityEstimate3}
 \end{align}
We denote by $\mc{A}_{\sf{g}}$ the set of achievable pairs $(\textsf{R},\textsf{D}) \in \mc{A}_{\sf{g}}$.
\end{definition}

\begin{theorem}[Zero-sum game]\label{theo:DecoderEstimation}
A pair of rate and distortion $(\textsf{R},\textsf{D}) \in \mc{A}_{\sf{g}}$ is achievable if and only if there exists an auxiliary random variable $W_1$ with distribution $ \QQ_{SW_1XY} \in \Q_{\sf{d}}$ that satisfies
\begin{align}
  \textsf{R} &\leq I(W_1;Y), \label{eq:DecoderEstimation1} \\
  \textsf{D} & = \min_{\PP_{V|W_1Y}} \E \Big[ d(S,V)\Big]  ,\label{eq:DecoderEstimation2} 
\end{align}
\mrb{where} the set $\Q_{\sf{d}}$ \mrb{is} defined in Theorem \ref{theo:LeakageCEd}.
\end{theorem}

The achievability proof of Theorem \ref{theo:DecoderEstimation} is \mrb{provided} in Appendix \ref{sec:AchievabilityTheoDecoderEstimation} and is a consequence of Theorem \ref{theo:LeakageCEd} and Theorem \ref{theo:Estimation}, and of \cite[Lemma A.8, Lemma A.21]{LeTreustTomala19}.
The converse proof of Theorem is \mrb{provided} in Appendix \ref{sec:ConverseTheoDecoderEstimation}. 

\begin{remark}[Maximin-minimax result]
The optimal distortion-rate function $\textsf{D}^{\star}( \textsf{R})$ reformulates as a maximin problem
\begin{align}
\textsf{D}^{\star}( \textsf{R}) =& \max_{\QQ_{W_1}, \QQ_{X|SW_1} \atop \textsf{R} \leq I(W_1;Y)} \min_{\PP_{V|W_1Y}} \E \Big[ d(S,V)\Big]\nonumber\\
=& \min_{\PP_{V|W_1Y}}  \max_{\QQ_{W_1}, \QQ_{X|SW_1} \atop \textsf{R} \leq I(W_1;Y)} \E \Big[ d(S,V)\Big].\label{eq:SolutionZeroSum2}
 \end{align}
The maximum and the minimum are taken over compact and convex sets and the distortion function is linear. Hence by Sion's Theorem \cite{sion1958general} the maximin is equal to the minimax and the value of this channel state estimation zero-sum game is $\textsf{D}^{\star}( \textsf{R})$.
 \end{remark}

\begin{remark}[One auxiliary random variable]
The formulation of Theorem \ref{theo:DecoderEstimation} is based on the set of distributions $\Q_{\sf{d}}$ with only one auxiliary random variable $W_1$, instead of the two random variables $(W_1,W_2)$ of the set $\Q_{\sf{c}}$. When the encoder tries to mask the channel state, it does not require the auxiliary random variable $W_2$ anymore, \mrb{since}
\begin{align}
\textsf{D}^{\circ} =& \max_{\QQ_{W_1}, \QQ_{X|SW_1}, \QQ_{W_2|SW_1} \atop \textsf{R} \leq I(W_1,W_2;Y) - I(W_2;S|W_1)} \min_{\PP_{V|W_1W_2Y}} \E \Big[ d(S,V)\Big]\\
\leq & \max_{\QQ_{W_1}, \QQ_{X|SW_1}, \QQ_{W_2|SW_1} \atop \textsf{R} \leq I(W_1,W_2;Y) - I(W_2;S|W_1)} \min_{\PP_{V|W_1Y}} \E \Big[ d(S,V)\Big]\label{eq:zerosumCirc1}\\
\leq & \max_{\QQ_{W_1}, \QQ_{X|SW_1} \atop \textsf{R} \leq I(W_1;Y) } \min_{\PP_{V|W_1Y}} \E \Big[ d(S,V)\Big] = \textsf{D}^{\star},\label{eq:zerosumCirc2}
 \end{align}
where \eqref{eq:zerosumCirc1} comes from taking the minimum over $\PP_{V|W_1Y}$ instead of $\PP_{V|W_1W_2Y}$;  \eqref{eq:zerosumCirc2} comes from the Markov chain $Y -\!\!\!\!\minuso\!\!\!\!- (S ,W_1) -\!\!\!\!\minuso\!\!\!\!-  W_2$ stated in \eqref{eq:distribution}, that ensures $ I(W_1,W_2;Y) - I(W_2;S|W_1)\leq I(W_1;Y)$. Hence, the information constraint $\textsf{R} \leq I(W_1,W_2;Y) - I(W_2;S|W_1)$ is more restrictive than $\textsf{R} \leq I(W_1;Y)$.
\end{remark}
 
\begin{remark}[Zero rate case]

\mrb{In the special case $\textsf{R}=0$, which corresponds to a channel estimation game without communication,} the encoding functions \mrb{reduce to} $f_{X_i|S^i}$ instead of $f_{X_i|S^iM}$. The channel state estimation zero-sum game \mrb{becomes the} maximin \mrb{problem}
\begin{align}
\max_{\{f_{X_i|S^i}\}_{i\in\{1,\ldots,n\}} }\min_{h_{V^n|Y^n}} \E\bigg[ \frac{1}{n} \sum_{i=1}^n d(S_i,V_i)\bigg], \label{eq:NstagesMaxMin}
 \end{align}
in which the encoder chooses $\{f_{X_i|S^i}\}_{i\in\{1,\ldots,n\}}$ and the decoder \mrb{chooses} $h_{V^n|Y^n}$. \mrb{Theorem~\ref{theo:DecoderEstimation} shows} that the single-letter solution is $\max_{\QQ_{W_1}, \QQ_{X|SW_1}} \min_{\PP_{V|W_1Y}} \E \Big[ d(S,V)\Big]$.

\mrb{If the objectives of \mlt{both} encoder and decoder were aligned, i.e., they would both try to minimize the long term average distortion
  \begin{align}
\min_{\{f_{X_i|S^i}\}_{i\in\{1,\ldots,n\}}, \atop h_{V^n|Y^n}}  \E\bigg[ \frac{1}{n} \sum_{i=1}^n d(S_i,V_i)\bigg], \label{eq:NstagesMinMitra}
  \end{align}
  the problem would become the causal channel state communication studied in}~\cite{ChoudhuriKimMitra13}.

\end{remark}

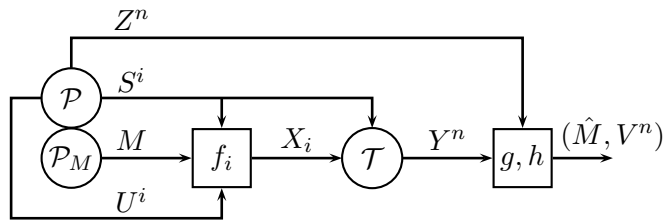
\begin{figure}[!ht]
\begin{center}
 \psset{xunit=0.8cm,yunit=0.8cm}
 \begin{pspicture}(-0.5,-0.5)(9,2.5)
 \pscircle(0,1.5){0.395}
 \pscircle(0,0.5){0.395}
 \psframe(2,0)(3,1)
 \pscircle(5,0.5){0.40}
 \psframe(7,0)(8,1)
 \psline[linewidth=1pt]{->}(0.5,0.5)(2,0.5)
 \psline[linewidth=1pt]{->}(3,0.5)(4.5,0.5)
 \psline[linewidth=1pt]{->}(5.5,0.5)(7,0.5)
 \psline[linewidth=1pt]{->}(8,0.5)(9,0.5)
 \psline[linewidth=1pt]{->}(0.5,1.5)(5,1.5)(5,1)
 \psline[linewidth=1pt]{->}(2.5,1.5)(2.5,1)
 \psline[linewidth=1pt]{->}(0,2)(0,2.5)(7.5,2.5)(7.5,1)
 \rput[u](1,2.8){$Z^{n}$}
 \psline[linewidth=1pt]{->}(-0.5,1.5)(-1,1.5)(-1,-0.5)(2.5,-0.5)(2.5,0)
 \rput[u](1,-0.2){$U^i$}
 \rput[u](1,0.8){$M$}
 \rput[u](1,1.8){$S^{i}$}
 \rput[u](3.75,0.8){$X_i $}
 \rput[u](6.25,0.8){$Y^n$}
 \rput[u](9,0.9){$(\hat{M}, V^n)$}
 \rput(0,0.5){$\PP_{M}$}
 \rput(0,1.5){$\PP$}
 \rput(2.5,0.5){$f_i$}
 \rput(5,0.5){$\mc{T}$}
 \rput(7.5,0.5){$g,h$}
 \end{pspicture}
\caption{\mrb{The causal encoding function is} $f_i: \mc{M} \times \mc{U}^i \times \mc{S}^{i} \rightarrow \mc{X}$, for all $i\in\{1,\ldots,n\}$ and the non-causal decoding functions \mrb{are} $g: \mc{Y}^n \times \mc{Z}^n \rightarrow \mc{M} $ and $h: \mc{Y}^n \times \mc{Z}^n \rightarrow \Delta(\mc{V}^n)$.}
\label{fig:CE2SI}
\end{center}
\end{figure}

\section{Extensions to more general scenarios}\label{sec:Extensions}

\subsection{Two-sided state information}\label{sec:CE2SI}

The case of two-sided state information is \mrb{illustrated in} Fig. \ref{fig:CE2SI}. \mrb{The channel state $\mrb{S^n}$, information source $\mrb{U^n}$ and \mrb{decoder} state information $\mrb{Z^n}$ are jointly distributed according to the i.i.d. distribution $\PP_{USZ} \in \Delta(\mc{U} \times\mc{S} \times \mc{Z}  )$.}

\begin{definition}\label{def:CodeLeakageCE2SI}
\mrb{A code with two-sided state information consists of stochastic functions $f_i : \mc{M} \times \mc{U}^i \times \mc{S}^{i}  \longrightarrow \Delta(\mc{X})$  $\forall i \in \{1, \ldots,n\}$, a deterministic decoding function $g : \mc{Y}^n  \times \mc{Z}^n \longrightarrow  \mc{M}$, and a stochastic receiver action function $h : \mc{Y}^n  \times \mc{Z}^n \longrightarrow \Delta(\mc{V}^n)$. The set of codes with causal encoding with length $n$ and message set $\mc{M}$ is denoted $\mc{C}_{\sf{s}}(n,\mc{M})$.} 
\end{definition}

A code $c\in\mc{C}_{\sf{s}}(n,\mc{M})$, the uniform distribution of the messages $\PP_{M}$, the source $\PP_{USZ}$ and the channel $\mc{T}_{Y|XS}$ induce a distribution on $(M,U^n,S^n,Z^n,X^n,Y^n,V^n,\hat{M})$ \mrb{given by}
\begin{align}
\PP_{M}  \prod_{i=1}^n\bigg[ \PP_{U_iS_iZ_i} & f_{X_i|U^iS^iM}  \mc{T}_{Y_i|X_iS_i}\bigg] \nonumber\\
&h_{V^n|Y^nZ^n}   \UN \big( \hat{M} = g(Y^n,Z^n)\big) .
\end{align}
\mrb{We denote by \mrb{$\mc{A}_{\sf{s}}$} the set of achievable triples $(\textsf{R},\textsf{E},\QQ_{USZXYV} )$, defined similarly as in Definition \ref{def:achievability}.}


%
\begin{theorem}[Two-sided state information]\label{theo:LeakageCE2SI} 
Consider a target distribution $\QQ_{USZXYV}$ that decomposes as $\QQ_{USZXYV} = \PP_{USZ}     \QQ_{X|US} \mc{T}_{Y|XS}      \QQ_{V|USZXY}$. \mrb{Then, $(\textsf{R},\textsf{E},\mrb{\QQ_{USZXYV}} )\in \mc{A}_{\sf{s}}$} if and only if there exist two auxiliary random variables $(W_1,W_2)$ with distribution $ \QQ_{USZW_1W_2XYV} \in \Q_{\sf{s}}$ \mrb{satisfying}
\begin{align}
I(U,S;W_1,W_2,Y,Z) \leq \textsf{E}&\leq H(U,S), \label{eq:theoremCE2SIachie2} \\
\textsf{R} + \textsf{E} &\leq  I(W_1,U,S; Y,Z),\label{eq:theoremCE2SIachie3}
\end{align}
where $\Q_{\sf{s}}$ is the set of distributions $\QQ_{USZW_1W_2XYV}$ that decompose as
\begin{align}
  &   \PP_{USZ}    \QQ_{W_1}    \QQ_{W_2|USW_1}
 \QQ_{X|USW_1}     \mc{T}_{Y|XS}    \QQ_{V|YZW_1W_2},
 \end{align}
and \mrb{such that}  $\max(|\mc{W}_1|,|\mc{W}_2|) \leq   d+1 $ with $d = |\mc{U}\times\mc{S}\times\mc{Z}\times  \mc{X}\times \mc{Y} \times \mc{V}  |$.
\end{theorem}

The achievability proof follows directly from the proof of Theorem   \ref{theo:LeakageCE}, by replacing the random variable of the channel state $S$ by the pair $(U,S)$ and the random variable of the channel output $Y$ by the pair $(Y,Z)$. The converse proof is \mrb{provided in the Supplementary Materials.}



\begin{remark}\label{remark:MarkovCausalSI}
The Markov chains $X -\!\!\!\!\minuso\!\!\!\!- (U,S ,W_1) -\!\!\!\!\minuso\!\!\!\!-  W_2$, $Y -\!\!\!\!\minuso\!\!\!\!- (X , S ) -\!\!\!\!\minuso\!\!\!\!-  ( U,Z,W_1,W_2 )$ and $Z -\!\!\!\!\minuso\!\!\!\!- (U , S ) -\!\!\!\!\minuso\!\!\!\!-  ( X, Y,W_1,W_2 )$ imply another Markov chain property $(Y,Z) -\!\!\!\!\minuso\!\!\!\!- (W_1 , U,S ) -\!\!\!\!\minuso\!\!\!\!-  W_2$. Indeed, for all $(u,s,z, w_1,w_2, x,y)$ we have
\begin{align*}
&\PP(y,z | w_1,w_2, u,s ) \nonumber\\
=& \sum_{x \in \mc{X} } \QQ(x | u,s,w_1)  \mc{T}(y|x,s)   \PP(z|u,s) = \PP(y,z | w_1,u,s ).
\end{align*}
By combining~\eqref{eq:theoremCE2SIachie2} and~\eqref{eq:theoremCE2SIachie3} with the Markov chain $(Y,Z) -\!\!\!\!\minuso\!\!\!\!- (W_1 , U,S ) -\!\!\!\!\minuso\!\!\!\!-  W_2$, we recover the information constraint of \cite[Theorem V.1]{LeTreust(ISIT-TwoSided)15}:
\begin{align}
\textsf{R} &\leq I(W_1,W_2;Y,Z) - I(W_2;U,S|W_1). \label{eq:theoremCE2SIachie1}
\end{align}
\end{remark}


%


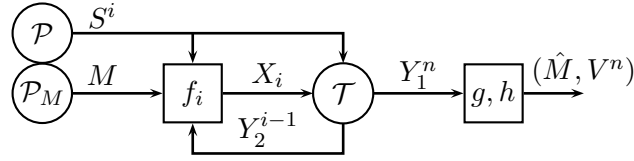
\begin{figure}
\begin{center}
 \psset{xunit=0.8cm,yunit=0.8cm}
 \begin{pspicture}(-0.5,-0.5)(9,1.7)
 \pscircle(0,1.5){0.395}
 \pscircle(0,0.5){0.395}
 \psframe(2,0)(3,1)
 \pscircle(5,0.5){0.40}
 \psframe(7,0)(8,1)
 \psline[linewidth=1pt]{->}(0.5,0.5)(2,0.5)
 \psline[linewidth=1pt]{->}(3,0.5)(4.5,0.5)
 \psline[linewidth=1pt]{->}(5.5,0.5)(7,0.5)
 \psline[linewidth=1pt]{->}(8,0.5)(9,0.5)
 \psline[linewidth=1pt]{->}(0.5,1.5)(5,1.5)(5,1)
 \psline[linewidth=1pt]{->}(2.5,1.5)(2.5,1)
 \psline[linewidth=1pt]{->}(5,0)(5,-0.5)(2.5,-0.5)(2.5,0)
 \rput[u](1,0.8){$M$}
 \rput[u](1,1.8){$S^{i}$}
 \rput[u](3.75,0.8){$X_i $}
 \rput[u](6.25,0.8){$Y_1^n$}
 \rput[u](3.75,-0.1){$Y_2^{i-1}$}
 \rput[u](9,0.9){$(\hat{M}, V^n)$}
 \rput(0,0.5){$\PP_{M}$}
 \rput(0,1.5){$\PP$}
 \rput(2.5,0.5){$f_i$}
 \rput(5,0.5){$\mc{T}$}
 \rput(7.5,0.5){$g,h$}
 \end{pspicture}
\caption{The noisy feedback \mrb{sequence $Y_2^{i-1}$ is drawn i.i.d. according to} $\mc{T}_{Y_1Y_2|XS} $. The encoding \mrb{is} $f_i: \mc{M} \times \mc{S}^{i}  \times \mc{Y}_2^{i-1}\rightarrow \mc{X}$, $\forall i\in\{1,\ldots,n\}$.}
\label{fig:CENoisyFeedback}
\end{center}
\end{figure}

\subsection{Noisy channel feedback observed by the encoder}\label{sec:NoisyFeedback}

\mrb{In this section, we consider that} the encoder has noisy feedback $Y_2$ from the state-dependent  channel  $\mc{T}_{Y_1Y_2|XS} $, as depicted in Fig. \ref{fig:CENoisyFeedback}. The encoding function \mrb{becomes} $f_i: \mc{M} \times \mc{S}^i \times \mc{Y}_2^{i-1} \rightarrow \mc{X}$, $\forall i\in\{1,\ldots,n\}$ \mrb{while} the  decoding functions and the \mrb{definition of the} state leakage remain unchanged. \mrb{The corresponding set of achievable triples $(\textsf{R},\textsf{E},\mrb{\QQ_{SXY_1Y_2V}} )$ is denoted $\mc{A}_{\sf{f}}$.}

\begin{theorem}[Noisy channel feedback]\label{theo:LeakageCENoisyFeedback}
We consider a target distribution $\QQ_{SXY_1Y_2V}$ that decomposes as $\QQ_{SXY_1Y_2V} =\PP_S     \QQ_{X|S} \mc{T}_{Y_1Y_2|XS}     \QQ_{V|SXY_1Y_2}$. \mrb{Then, $(\textsf{R},\textsf{E},\mrb{\QQ_{SXY_1Y_2V}} )\in\mc{A}_{\sf{f}}$} if and only if there exist two auxiliary random variables $(W_1,W_2)$ with distribution $ \QQ_{SW_1W_2XY_1Y_2V} \in \Q_{\sf{f}}$ that satisfy
\begin{align}
\textsf{R} \leq I(W_1&,W_2;Y_1) - I(W_2;S, Y_2 |W_1), \label{eq:theoremCENoisyFeedback1} \\
I(S;W_1,W_2,Y_1) \leq \textsf{E}&\leq H(S), \label{eq:theoremCENoisyFeedback2} \\
\textsf{R} + \textsf{E} &\leq  I(W_1,S; Y_1),\label{eq:theoremCENoisyFeedback3}
\end{align}
where $\Q_{\sf{f}}$ is the set of distributions with marginal $\QQ_{SW_1W_2XY_1Y_2V}$ that decompose as 
\begin{align*}
\PP_S      \QQ_{W_1}  \QQ_{X|SW_1}     \mc{T}_{Y_1Y_2|XS}
    \QQ_{W_2|SW_1Y_2}     \QQ_{V|Y_1W_1W_2},
 \end{align*}
and \mrb{such that} $\max(|\mc{W}_1|,|\mc{W}_2|) \leq   d+1 $ with $d = |  \mc{S}\times  \mc{X}\times \mc{Y}_1 \times \mc{Y}_2 \times \mc{V} | $.
\end{theorem}

The achievability proof of Theorem  \ref{theo:LeakageCENoisyFeedback} follows directly from the proof of Theorem \ref{theo:LeakageCE}, by replacing the pair $(S^n,W_1^n)$ by the triple $(S^n,W_1^n,Y_2^n)$  \mrb{in order to select} $W_2^n$. The decoding functions and the leakage analysis remain unchanged. The converse proof is stated in \mrb{the Supplementary Materials.}


\begin{remark}[Noisy feedback improves coordination]\label{remark:FeedbackIncreases}
The channel feedback increases the set of achievable triples, \mrb{i.e. $\mc{A}_{\sf{c}}\subset \mc{A}_{\sf{f}}$}, since the conditional distribution $ \QQ_{W_2|SW_1Y_2} $ depends on channel outputs $Y_2$. The information constraints of Theorem \ref{theo:LeakageCENoisyFeedback} are reduced to that of Theorem \ref{theo:LeakageCE} since $ \QQ_{W_2|SW_1Y_2} =\QQ_{W_2|SW_1} \Longleftrightarrow W_2  -\!\!\!\!\minuso\!\!\!\!-  ( S,W_1 ) -\!\!\!\!\minuso\!\!\!\!- Y_2 \Longleftrightarrow I(W_2; Y_2 |S, W_1)=0$. \mrb{This was already pointed out for the coordination problem} in  \cite{LeTreust(ISITfeedbacks)15}, \mrb{and for the rate-and-state capacity problem in~\cite{BrossLapidoth18}.}
\end{remark}

%
%

\begin{figure}[!ht]
\begin{center}
 \psset{xunit=0.9cm,yunit=0.9cm}
 \begin{pspicture}(0,0)(9,1.7)
 \pscircle(0,1.5){0.445}
 \pscircle(0,0.5){0.445}
 \psframe(2,0)(3,1)
 \pscircle(5,0.5){0.45}
 \psframe(7,0)(8,1)
 \psline[linewidth=1pt]{->}(0.5,0.5)(2,0.5)
 \psline[linewidth=1pt]{->}(3,0.5)(4.5,0.5)
 \psline[linewidth=1pt]{->}(5.5,0.5)(7,0.5)
 \psline[linewidth=1pt]{->}(8,0.5)(9,0.5)
 \psline[linewidth=1pt]{->}(0.5,1.5)(5,1.5)(5,1)
 \psline[linewidth=1pt]{->}(2.5,1.5)(2.5,1)
 \rput[u](1,0.8){$M$}
 \rput[u](1,1.8){$S^{i-1}$}
 \rput[u](3.75,0.8){$X_i $}
 \rput[u](6.25,0.8){$Y^n$}
 \rput[u](8.8,0.8){$(\hat{M}, V^n)$}
 \rput(0,0.5){$\PP_{M}$}
 \rput(0,1.5){$\PP$}
 \rput(2.5,0.5){$f_i$}
 \rput(5,0.5){$\mc{T}$}
 \rput(7.5,0.5){$g,h$}
 \end{pspicture}
\caption{The strictly causal encoding function \mrb{is} $f_i: \mc{M} \times \mc{S}^{i-1} \rightarrow \Delta(\mc{X})$, for all $i\in\{1,\ldots,n\}$ and the non-causal decoding functions \mrb{are} $g: \mc{Y}^n  \rightarrow \mc{M}$ and  $h: \mc{Y}^n  \rightarrow \Delta(\mc{V}^n)$.}
\label{fig:SCE}
\end{center}
\end{figure}
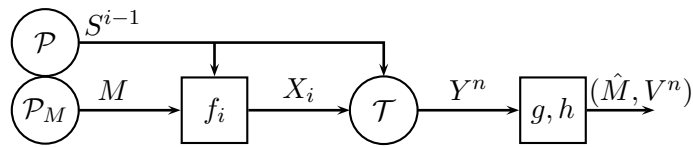

\section{Strictly causal encoding}\label{sec:StrictlyCausal}

\begin{definition}\label{def:CodeLeakageSCE}
\mrb{A code with strictly causal encoding consists of stochastic encoding functions $f_i : \mc{M} \times \mc{S}^{i-1}  \longrightarrow \Delta(\mc{X})$ $\forall i \in \{1, \ldots,n\}$, a deterministic decoding function $g : \mc{Y}^n  \longrightarrow  \mc{M} $, and a stochastic receiver action function $h : \mc{Y}^n  \longrightarrow   \Delta(\mc{V}^n)$.  The set of codes with strictly causal encoding with length $n$ and message set $\mc{M}$ is denoted $\mc{C}_{\sf{sc}}(n,\mc{M})$. The corresponding set of achievable triples $(\textsf{R},\textsf{E},\QQ_{SXYV} )$ is defined similarly as in Definition~\ref{def:achievability} and is denoted $\mc{A}_{\sf{sc}}$.}
\end{definition}

\begin{theorem}[Strictly causal encoding]\label{theo:LeakageSCE}
Consider a target distribution $\QQ_{SXYV}$ that decomposes as $\QQ_{SXYV} = \PP_S      \QQ_X \mc{T}_{Y|XS}      \QQ_{V|SXY}$. \mrb{Then, $(\textsf{R},\textsf{E},\mrb{\QQ_{SXYV}} )\in \mc{A}_{\sf{sc}}$} if and only if there exists an auxiliary random variable $W_2$ with distribution $\QQ_{SW_2XYV}  \in \Q_{\sf{sc}}$ that satisfies
\begin{align}
I(S;X,W_2,Y) \leq \textsf{E}&\leq H(S), \label{eq:theoremSCEachie2} \displaybreak[0]\\
\textsf{R} + \textsf{E} &\leq  I(X,S; Y),\label{eq:theoremSCEachie3}
\end{align}
where $\Q_{\sf{sc}}$ is the set of distributions $\QQ_{SW_2XYV}$ with marginal $\QQ_{SW_2XYV}$ that decompose as
\begin{align}
 \QQ_{SW_2XYV} = \PP_S    \QQ_X    \QQ_{W_2|SX}
    \mc{T}_{Y|XS}    \QQ_{V|XYW_2}
     \end{align}
and \mrb{such that} $|\mc{W}_2| \leq   | \mc{S}\times  \mc{X}\times \mc{Y}| + 1  $.
\end{theorem}

The achievability \mrb{proof is obtained from Theorem \ref{theo:LeakageCE}} by replacing the auxiliary random variable $W_1$ by the channel input $X$. The converse proof is \mrb{provided in the Supplementary Materials.}

\begin{remark}\label{remark:MarkovCausalSCE}
Equation~\eqref{eq:theoremSCEachie3}, the first inequality of \eqref{eq:theoremSCEachie2}, the Markov chain $Y -\!\!\!\!\minuso\!\!\!\!- (X , S ) -\!\!\!\!\minuso\!\!\!\!- W_2$, and the independence between $S$ and $X$ imply
\begin{align}
\textsf{R} &\leq I(X,W_2;Y) - I(W_2;S|X). \label{eq:theoremSCEachie1} 
\end{align}
\end{remark}



%

\begin{corollary}[Without receiver's outputs]\label{cor:LeakageSCEb}
A pair of rate and state leakage $(\textsf{R},\textsf{E})$ is achievable if and only if  there exists a distribution $  \QQ_X $ that satisfies
\begin{align}
 I(S; Y |X) \leq \textsf{E}&\leq H(S),   \label{eq:theoremSCEachiev2b}   \\
\textsf{R} + \textsf{E} &\leq I(X,S; Y).   \label{eq:theoremSCEachiev3b} 
\end{align}
\end{corollary}

The achievability proof of Corollary~\ref{cor:LeakageSCEb} comes from the achievability proof of Theorem \ref{theo:LeakageSCE}. The converse proof is based on standard arguments. Equations  \eqref{eq:theoremSCEachiev2b} and \eqref{eq:theoremSCEachiev3b} imply $\textsf{R} \leq I(X;Y)  $.

\appendices\label{sec:App}

%
\section{Achievability proof of Theorem \ref{theo:LeakageCE}}\label{sec:AchievabilityLeakageCE}

\subsection{Random coding}\label{sec:RandomCoding}

\mrb{The case $H(S)=0$ can be handled using a standard channel coding scheme and is detailed in the Supplementary Materials. In the remaining of the proof, we assume that $H(S)>0$ and we fix a rate, a state leakage, and a distribution} $(\textsf{R},\textsf{E},\mrb{\QQ_{SXYV}})$ for which there exists a distribution $ \QQ_{SW_1W_2XYV} \in \Q_{\sf{c}}$  \mrb{with marginal $\QQ_{SXYV}$, that satisfies}
\begin{align}
 I(S;W_1,W_2,Y) < \textsf{E}&\leq H(S),\label{eq:RateConstraint2ce}\\
 \textsf{R} + \textsf{E} &\leq I(W_1,S; Y), \label{eq:RateConstraint3ce}
\end{align}
\mrb{where the  inequality in \eqref{eq:RateConstraint2ce} is strict. The case of equality has to be treated with care because the channel capacity might be zero; a detailed analysis is available in the Supplementary Materials.} We show that $(\textsf{R},\textsf{E},\mrb{\QQ_{SXYV}})$ is achievable by introducing the rate parameters $\textsf{R}_{\sf{L}}$, $\textsf{R}_{\sf{J}}$, $\textsf{R}_{\sf{K}}$ and by considering a block-Markov random code $c\in \mc{C}(n B , \mc{M} )$ defined over $B\in \N^{\star}$ blocks of length $n\in \N^{\star}$. The codebook is defined over one block of length $n\in \N^{\star}$  and the total length of the code is denoted by $N = nB\in \N^{\star}$. \mrb{In the following, the notation $T_{\delta}(\QQ)$ stands for the set of typical sequences with respect to the distribution $\QQ=\QQ_{SW_1W_2XYV}$.}

\textbf{Random Codebook.}\\
1)  We draw $2^{n (H(S)+ \varepsilon)}$ sequences $S^n(l,j)$ according to \mrb{the} i.i.d. distribution $\PP_S$, with \mrb{indices} $(l,j)\in \mc{M}_{\textsf{L}} \times \mc{M}_{\textsf{J}} $ \mrb{such that}  $|\mc{M}_{\textsf{L}}|= 2^{n  \textsf{R}_{\sf{L}}} $ and $|\mc{M}_{\textsf{J}}|= 2^{n  \textsf{R}_{\sf{J}}} $. \\
2) We draw $2^{n(\textsf{R} + \textsf{R}_{\sf{L}} + \textsf{R}_{\sf{K}} )}$ sequences $W_1^n(m,l,k)$ according to the i.i.d. distribution $\QQ_{W_1}$ with \mrb{indices} $(m,l,k)\in\mc{M} \times \mc{M}_{\textsf{L}} \times \mc{M}_{\textsf{K}}$.\\
3)  For each triple of \mrb{indices} $(m,l,k)\in\mc{M} \times \mc{M}_{\textsf{L}}\times \mc{M}_{\textsf{K}}$, we draw the same number $2^{n(\textsf{R} + \textsf{R}_{\sf{L}} + \textsf{R}_{\sf{K}} )}$ of sequences $W_2^n(m,l,k,\hat{m},\hat{l},\hat{k})$ with \mrb{indices} $(\hat{m},\hat{l},\hat{k})\in\mc{M} \times \mc{M}_{\textsf{L}}\times \mc{M}_{\textsf{K}}$ according to  the i.i.d. conditional distribution $\mrb{\QQ_{W_2|W_1}}$ depending on $W_1^n(m,l,k)$.

  \textbf{Encoding function at the beginning of block $b\in \{2,\ldots B-1\}$.}\\
 1) The encoder observes the sequence of channel states $S^n_{b-1}$ corresponding to the block $b-1$ and finds the \mrb{indices} $(l_{b-1},j_{b-1})\in\mc{M}_{\textsf{L}} \times \mc{M}_{\textsf{J}}$ such that $\big( S^n(l_{b-1},j_{b-1}) , S^n_{b-1}\big) \in T_{\delta}(\mrb{\PP})$ for the  distribution \mrb{ $\PP(s,\hat{s}) = \PP(s) \UN(\hat{s} = s)$ $\forall(s,\hat{s})\in S\times S$}.\\
2) The encoder observes the message $m_b$ and the index $l_{b-1}$ and recalls $W_1^n(m_{b-1}, l_{b-2}, k_{b-1}) $ corresponding to the block $b-1$. It finds the index $k_b \in \mc{M}_{\textsf{K}}$ such that 
$ \big(S^n_{b-1}, W_1^n(m_{b-1}, l_{b-2}, k_{b-1})  ,\\W_2^n(m_{b-1}, l_{b-2}, k_{b-1}, m_{b}, l_{b-1}, k_{b}) \big)\in T_{\delta}(\QQ)$.\\
3) The encoder sends $X^n_b$ drawn from the i.i.d. conditional distribution $\QQ_{X|SW_1}$  depending on  $W_1^n(m_{b}, l_{b-1}, k_{b}) $ and $S^n_b$ observed causally on the current block $b\in \{2,\ldots B-1\}$. 

  \textbf{Decoding function at the end of block $b\in \{2,\ldots B-1\}$.}\\ 
1) The receiver recalls $Y^n_{b-1} $ and the \mrb{indices} $(m_{b-1}, l_{b-2}, k_{b-1}) $ corresponding to $W_1^n(m_{b-1}, l_{b-2}, k_{b-1}) $ decoded at the end of the block $b-1$. \\
2) The receiver observes $Y^n_b$ and finds the triple of \mrb{indices} $(m_b,l_{b-1}, k_{b})$ such that  $\big(Y^n_b ,  W_1^n(m_b,l_{b-1}, k_{b}) \big) \in T_{\delta}(\QQ)$  and  $\big(Y^n_{b-1} ,  W_1^n(m_{b-1}, l_{b-2}, k_{b-1})  ,\\W_2^n(m_{b-1}, l_{b-2}, k_{b-1}, m_{b}, l_{b-1}, k_{b}) \big)  \in T_{\delta}(\QQ)$.\\
3) The receiver returns the message $m_b$ corresponding to block $b$. \\
4) The receiver returns $V^n_{b-1}$ drawn from the conditional distribution $\QQ_{V|YW_1W_2}$  depending on  $\big(Y^n_{b-1} ,  W^n_{1}(m_{b-1}, l_{b-2}, k_{b-1}) ,\\W_2^n(m_{b-1}, l_{b-2}, k_{b-1}, m_{b}, l_{b-1}, k_{b}) \big)$.\\
5) The receiver knows that over block $b-1$, the sequences 
$\big(S^n_{b-1},  W_1^n(m_{b-1}, l_{b-2}, k_{b-1})  ,W_2^n(m_{b-1}, l_{b-2}, k_{b-1}, \\m_{b}, l_{b-1}, k_{b}), X^n_{b-1} , Y^n_{b-1} , V^n_{b-1}  \big) \in T_{\delta}(\QQ)$ and the sequence of states $S^n_{b-1}$  belongs to the bin with index $l_{b-1} \in \mc{M}_{\textsf{L}} $.

 \textbf{Initialization of the encoder.} Arbitrary \mrb{indices} $(m_1,l_0,k_1) $ are given to both encoder and decoder. The encoder sends the sequence $X^n_{b_1}$ drawn according to the conditional distribution $\QQ_{X|SW_1}$  depending on  $\big(S^n_{b_1}, W_1^n(m_{1},l_0,k_1)  \big)$. At the beginning of the second block  $b_2$, encoder recalls $ W_1^n(m_{1},l_0,k_1) $, observes \mrb{the} message $m_2$, finds the index $l_1$ such that   $\big(S^n_{b_1}, S^n(l_1,j_1) \big)\in T_{\delta}(\mrb{\PP})$ and finds the index $k_2$ such that   $\big(S^n_{b_1}, W_1^n(m_{1},l_0,k_1) ,W_2^n(m_{1},l_0,k_1,m_2,l_1,k_2) \big)\in T_{\delta}(\QQ)$. The encoder sends $X^n_{b_2}$ drawn from the conditional distribution $\QQ_{X|SW_1}$  depending on  $\big(S^n_{b_2}, W_1^n(m_{2},l_1,k_2)  \big)$. The index $m_1$ does not correspond to \mrb{an informational} message.

 \textbf{Initialization of the decoder.}  At the end of second block $b_2$, the decoder finds the triple of \mrb{indices} $(m_2,l_1,k_2)$ such that $\big(Y^n_{b_2} ,W_1^n(m_2,l_1,k_2) \big) \in T_{\delta}(\QQ)$  and  $\big(Y^n_{b_1} , W_1^n(m_{1},l_0,k_1), W_2^n(m_{1},l_0,k_1 , m_2,l_1,k_2) \big)\\ \in T_{\delta}(\QQ)$. The decoder returns the message $m_2$ corresponding to the block $b_2$ and the sequence $V^n_{b_1} \in \mc{V}^n$ drawn from the conditional probability $\QQ_{V|YW_1W_2}$  depending on  $\big(Y^n_{b_1} , W^n_{1}(m_{1},l_0,k_1), W^n_{2}(m_{1},l_0,k_1 ,  m_2, l_1,k_2)\big)$.  The decoder knows that over the first block $b_1$, the sequence $S^n_{b_1}$ belongs to the bin $l_1$. The sequences  $\big(S^n_{b_1} ,  W_1^n(m_{1},l_0,k_1), W_2^n(m_{1},l_0,k_1 , m_2,l_1,k_2) ,  X^n_{b_1} , Y^n_{b_1},\\ V^n_{b_1} \big) \in T_{\delta}(\QQ)$ and $\big(S^n_{b_2},  W_1^n(m_{2},l_1,k_2),\\ W_2^n(m_{2},l_1,k_2 , m_3,l_2,k_3)  ,  X^n_{b_2} , Y^n_{b_2} , V^n_{b_2} \big) \in T_{\delta}(\QQ)$.

 \textbf{Last block  at the encoder.} At the beginning of the last block $B$, the encoder recalls
$W_1^n(m_{B-1}$, observes message $m_B$, finds the index $l_{B-1}$ such that   $\big(S^n_{B-1}, S^n( l_{B-1} , j_{B-1} ) \big)\in T_{\delta}(\mrb{\PP})$. It finds the index $k_B$ such that  $ \big(S^n_{B-1}, W_1^n(m_{B-1}, l_{B-2}, k_{B-1})  ,W_2^n(m_{B-1}, l_{B-2},\\ k_{B-1}, m_{B}, l_{B-1}, k_{B}) \big)\in T_{\delta}(\QQ)$. The encoder sends the sequence $X^n_B$ drawn from the conditional distribution $\QQ_{X|SW_1}$  depending on $\big(S^n_{B}, W_1^n(m_{B},l_{B-1},k_B)  \big)$.

 \textbf{Last block  at the decoder.} At the end of the last block $B$, the decoder finds the triple of \mrb{indices} $(m_B,l_{B-1},k_B)$ such that $\big(Y^n_{B} ,W_1^n(m_B,l_{B-1},k_B) \big) \in T_{\delta}(\QQ)$  and  $\big(Y^n_{B-1} , W_1^n(m_{B-1},l_{B-2},k_{B-1}), \\W_2^n(m_{B-1},l_{B-2},k_{B-1} , m_B,l_{B-1},k_B) \big) \in T_{\delta}(\QQ)$. The decoder returns the message $m_B$ corresponding to the last block $B$ and the sequence $V^n_{B-1} \in \mc{V}^n$ drawn from the conditional probability $\QQ_{V|YW_1W_2}$  depending on  $\big(Y^n_{B-1} , W_1^n(m_{B-1},l_{B-2},k_{B-1}) , \\W_2^n(m_{B-1},l_{B-2},k_{B-1} , m_B,l_{B-1},k_B) \big)$. The decoder knows that over the block $B-1$, the sequence $S^n_{B-1}$ belongs to the bin $l_{B-1}$. Then  $\big(S^n_{B-1} ,  W_1^n(m_{B-1},l_{B-2},k_{B-1}),\\ W_2^n(m_{B-1},l_{B-2},k_{B-1} , m_B,l_{B-1},k_B),  X^n_{B-1},  Y^n_{B-1} ,  V^n_{B-1}  \big)\\ \in T_{\delta}(\QQ)$ but $(S^n_{B} ,W^n_{1,B},W^n_{2,B},X^n_{B} ,Y^n_{B} ,V^n_{B} ) \notin T_{\delta}(\QQ)$ on the last block $B$. The decoder does not know the index $l_B$ of the bin corresponding to sequence $S^n_{B}$.

In the following, we introduce the notation $W^n_{1,b} =W_1^n(m_{b}, l_{b-1}, k_{b}) $ and $W^n_{2,b} =W_2^n(m_{b}, l_{b-1}, k_{b}, m_{b+1}, l_{b}, k_{b+1}) $, with $b \in \{1,\ldots,B-1\}$. If there is no error in the coding scheme, the messages $(m_2,\ldots, m_{B})$ are correctly decoded and the decoder knows the bin \mrb{indices} $(l_1,\ldots,l_{b-1})$ of the sequences $(S^n_1,\ldots,S^n_{b-1})$, \mrb{and} $(S_b^n,W_{1,b}^n,W_{2,b}^n,X_b^n,Y_b^n,V_b^n) \in T_{\delta}(\QQ)$, for each blocks $b \in \{1,\ldots,B-1\}$.

\subsection{Rate parameters $\textsf{R}$,  $\textsf{R}_{\sf{L}}$ and $\textsf{R}_{\sf{K}}$ }\label{sec:RateParametersCE}
1)  At the end of block $b\in \{2,\ldots B\}$, the decoder observes $(Y^n_{b-1} ,Y^n_{b})$ and decodes $(W^n_{1,b-1}, W^n_{2,b-1})$ corresponding to the block $b-1$. Intuitively, the observation of $(W^n_{1,b-1},W^n_{2,b-1},Y^n_{b-1})$ leaks $n I(S;W_1,W_2,Y) = n I(S; W_2,Y | W_1)$ bits of information regarding $S^n_{b-1}$. By fixing the rate parameter $ \textsf{R}_{\sf{L}} = \textsf{E}-I(S;W_1,W_2,Y)$, the encoder will transmit $ n\textsf{R}_{\sf{L}} =  n \big(\textsf{E}-I(S;W_1,W_2,Y)\big)$ additional bits of information corresponding to the state sequence $S^n_{b-1}$. As it will be proven in the Section \ref{sec:LeakageCE}, the leakage rate $I(S^n_{b-1};Y^n_{b-1}) $ over block $b-1$ is close to $ n  \big( I(S;W_1,W_2,Y) +  \textsf{E}-I(S;W_1,W_2,Y) \big)= n    \textsf{E}$. We fix the rate parameter $\textsf{R}_{\sf{L}}$ equal to:
\begin{align}
\textsf{R}_{\sf{L}}=&\textsf{E} - I(S;W_1,W_2,Y) - 2 \varepsilon \geq 0.\label{eq:rateRLce}
\end{align}
The first inequality $I(S;W_1,W_2,Y) \leq \textsf{E}$ in \eqref{eq:RateConstraint2ce} implies there exists a positive rate parameter $\textsf{R}_{\sf{L}} $. In case of equality $\textsf{E} = I(S;W_1,W_2,Y)$, then the rate $\textsf{R}_{\sf{L}}=0$ and no index $l\in \mc{M}_{\sf{L}}$ is transmitted to the decoder.

2) The rates parameters $\textsf{R}_{\sf{L}}$, $\textsf{R}_{\sf{J}}$ corresponding to the \mrb{indices} $(l_{b-1},j_{b-1})$, guarantee that almost every sequences $S^n_{b-1}$ appear in the codebook.
\begin{align}
\textsf{R}_{\sf{L}} + \textsf{R}_{\sf{J}} =& H(S) + \varepsilon, \label{eq:AchExactLeak1ce} \\
\Longrightarrow \textsf{R}_{\sf{J}}=&H(S) -  \textsf{E} + I(S;W_1,W_2,Y) +3\varepsilon.
\end{align}
The second inequality $ \textsf{E} \leq H(S)$ in  \eqref{eq:RateConstraint2ce}   implies there exists a positive rate parameter $\textsf{R}_{\sf{J}} $.

3) The rates parameter $ \textsf{R}_{\sf{K}}$ corresponding to the index $k_b$, is used by the encoder in order to correlate $\big(W_1^n(m_{b-1}, l_{b-2}, k_{b-1})  ,W_2^n(m_{b-1}, l_{b-2}, k_{b-1}, m_{b}, l_{b-1}, k_{b})\big)$ with $S^n_{b-1}$.
\begin{align}
\textsf{R}_{\sf{K}}=& I(W_2;S|W_1)  + \varepsilon,\label{eq:AchExactLeak6ce}
\end{align}
Since the random variables $W_1$ and $S$ are independent, we have  $ I(W_1,W_2;S) =  I(W_2;S|W_1)$.

4) The rate parameters $\textsf{R} $, $\textsf{R}_{\sf{L}} $, $\textsf{R}_{\sf{K}} $ are correctly decoded if
\begin{align}
&\textsf{R} + \textsf{R}_{\sf{L}} +  \textsf{R}_{\sf{K}} \leq I(W_1;Y) + I(W_2;Y|W_1)   - \varepsilon \label{eq:AchExactLeak7ce}\\
\Longleftrightarrow &\textsf{R} +\textsf{E} - I(S;W_1,W_2,Y)  - 2 \varepsilon+ I(W_2;S|W_1)  + \varepsilon\nonumber\\& \leq I(W_1,W_2;Y)   - \varepsilon\\
\Longleftrightarrow &\textsf{R} +\textsf{E}  \leq I( W_1,W_2,S;Y ) \label{eq:rateSCE5ce}\\
\Longleftrightarrow &\textsf{R} +\textsf{E}  \leq I( W_1,S;Y ), \label{eq:rateSCE4ce}
\end{align}
where \eqref{eq:rateSCE5ce} comes from the independence between $W_1$ and $S$; 
 \eqref{eq:rateSCE4ce} comes from the Markov chain  $Y -\!\!\!\!\minuso\!\!\!\!- (W_1 , S ) -\!\!\!\!\minuso\!\!\!\!-  W_2$ stated in Remark \ref{remark:MarkovCausal}.

Equation \eqref{eq:RateConstraint3ce} implies that for each block $b\in \{2,\ldots B\}$, the \mrb{indices} with rates $\textsf{R} $, $\textsf{R}_{\sf{L}} $, $\textsf{R}_{\sf{K}} $ are recovered by the decoder, with large probability. Hence the rate of the total code of length $N = n B$  is given by
\begin{align}
&\frac{1}{n B} \sum_{b=2}^{B} \log_2 |\mc{M}| 
= \frac{B-1}{ B}  \textsf{R} = \textsf{R} -  \frac{1}{ B}  \textsf{R} \nonumber\\&
\geq \textsf{R} -  \frac{1}{ B}   \log_2 |\mc{Y}| 
\geq \textsf{R} -  \varepsilon.\label{eq:totalRatece}
\end{align}
The last equation is satisfied when the number of blocks is sufficiently large \mrb{such that} $ \frac{1}{ B}   \log_2 |\mc{Y}| \leq \varepsilon$.


%

\subsection{Expected error probability by block} \label{sec:ErrorProbaCE}
For each block $b\in \{2,\ldots ,B\}$, we consider the expected probability of the following error events. The properties of the typical sequences, see \cite[pp. 27]{ElGamalKim(book)11}, implies that \mrb{for all $\varepsilon>0$} there exists $n_1\in \N^{\star}$ such that for all $n\geq n_1$, the expected probability of the error event is bounded by 
\begin{align}
\E_c\bigg[ \prob\bigg( S^n_{b-1}   \notin T_{\delta}(\QQ) \bigg)\bigg]  \leq \varepsilon.
\end{align}
From the covering Lemma, \mrb{ see \cite[pp. 62]{ElGamalKim(book)11}}, 
 \begin{align*}
 \textsf{R}_{\sf{L}}  +  \textsf{R}_{\sf{J}}    \geq& H(S) + \varepsilon   
\end{align*}
implies that \mrb{$\forall \varepsilon>0$,} $ \exists n_2\in \N^{\star}$ such that $ \forall n\geq n_2$, 
\begin{align}
&\E_c\bigg[ \prob\bigg( \forall  (L_{b-1},J_{b-1})\in \mc{M}_{\sf{L}} \times \mc{M}_{\sf{J}} ,\nonumber\\
&\qquad ( S^n(L_{b-1},J_{b-1}), S^n_{b-1})  \notin T_{\delta}(\QQ)  \bigg)\bigg]  \leq \varepsilon.
\end{align}

 From the covering Lemma,  \mrb{ see \cite[pp. 62]{ElGamalKim(book)11}},  \begin{align*}
 \textsf{R}_{\sf{K}}\geq& I(W_2;S|W_1)  + \varepsilon  
\end{align*}
implies that \mrb{$\forall \varepsilon>0$,} $ \exists n_3\in \N^{\star}$ such that $ \forall n\geq n_3$, 
\begin{align}
&\E_c\bigg[ \prob\bigg( \forall  K_b \in \mc{M}_{\sf{K}},\; (S^n_{b-1}, W_1^n(M_{b-1}, L_{b-2}, K_{b-1})  ,\nonumber\\
&W_2^n(M_{b-1}, L_{b-2}, K_{b-1}, M_{b}, L_{b-1}, K_{b}) ) \notin T_{\delta}(\QQ )  \bigg)\bigg]  \leq \varepsilon.
\end{align}

 From the packing Lemma,  \mrb{ see \cite[pp. 46]{ElGamalKim(book)11}}, 
 \begin{align*}
\textsf{R} + \textsf{R}_{\sf{L}} +  \textsf{R}_{\sf{K}} \leq& I(W_1;Y) + I(W_2;Y|W_1)   - \varepsilon  
\end{align*}
implies that \mrb{$\forall \varepsilon>0$,} $ \exists n_4\in \N^{\star}$ such that $ \forall n\geq n_4$, 
\begin{align}
&\E_c\bigg[ \prob\bigg( \exists (M_b,L_{b-1}, K_{b}) \neq (M_b',L_{b-1}', K_{b}'),\text{ s.t. } \nonumber\\
&\Big\{(Y^n_b ,  W_1^n(M_b',L_{b-1}', K_{b}') ) \in T_{\delta}(\QQ) \Big\} \cap \nonumber\\
& \Big\{ (Y^n_{b-1} ,  W_1^n(M_{b-1}, L_{b-2}, K_{b-1})  ,\nonumber\\
&W_2^n(M_{b-1}, L_{b-2}, K_{b-1}, M_{b}', L_{b-1}', K_{b}') ) \in T_{\delta}(\QQ)\Big\}  \bigg)\bigg]  \leq \varepsilon.
\end{align}

\mrb{Thus} for each block $b\in \{2,\ldots,B \}$ and for all $n\geq \bar{n} \geq \max(n_1,n_2,n_3,n_4)$, the expected probability of non-decoding the \mrb{indices} $(M_b , L_{b-1},K_b)$ is bounded by:
 \begin{align}
\E_{c} \bigg[\prob\Big((M_b , L_{b-1},K_b) \neq ( \hat{M}_b , \hat{L}_{b-1} ,  \hat{M}_b)    \Big) \bigg]   \leq& 4 \varepsilon.
\end{align}


 \subsection{Expected error probability of the block-Markov code}
 
We evaluate the expected probability of error for the random \mrb{indices} $(M_b,L_{b-1},K_b)$, for $b \in \{2,\ldots,B\}$ of the block-Markov random code:
\begin{align}
 & \E_{c} \bigg[\prob\bigg(  \Big(M_2,L_1,K_2\ldots, M_{B},L_{B-1},K_B\Big) \neq \nonumber\\
 &\Big(\hat{M}_2,\hat{L}_1,\hat{K}_2, \ldots, \hat{M}_{B}, \hat{L}_{B-1} , \hat{K}_{B} \Big) \bigg)  \bigg]\nonumber \\
 =& 1 - \E_{c} \bigg[\prob\bigg(   (M_2,L_1,K_2) = (\hat{M}_2 , \hat{L}_1,\hat{K}_2)  \bigg) \bigg]    \nonumber \\
&\times \E_{c} \bigg[\prob\bigg(    ( M_B , L_{B-1},K_B) = ( \hat{M}_B , \hat{L}_{B-1}, \hat{K}_{B}) \Big| \nonumber\\
& \Big\{  (M_2,L_1,K_2) = (\hat{M}_2 , \hat{L}_1,\hat{K}_2) \Big\} \cap  \ldots \cap  \nonumber\\
& \Big\{   ( M_{B-1} , L_{B-2} , K_{B-1} ) = ( \hat{M}_{B-1} , \hat{L}_{B-2}, \hat{K}_{B-1} )\Big\}\bigg) \bigg]   \nonumber\\
\leq& 1 - \bigg( 1 - 4 \varepsilon \bigg)^{B-1}.\label{eq:ErrorProbaTotalce}
 \end{align}

We denote by $\widetilde{Q}^N \in \Delta( \mc{S} \times \mc{X}\times \mc{Y}\times \mc{V} )$, the empirical distribution of symbols over every blocks $b \in \{1,\ldots ,B-1\}$ removing the last block. We show $\widetilde{Q}^N$ is close to the empirical distribution $Q^N$ over all the $B$ blocks, for a number of blocks $B\in \N^{\star}$ sufficiently large, \textit{i.e.} for which $ \frac{2}{B}  |  \mc{S} \times \mc{X} \times\mc{Y}  \times \mc{V}  | \leq \varepsilon$. We denote by $Q_B$, the empirical distribution of symbols over the last block.
\begin{align}
\Big|\Big|{Q}^N  -& \widetilde{Q}^N\Big|\Big|_{1} 
= \Big|\Big|\frac{1}{B} \Big( (B-1) \widetilde{Q}^N +  Q_B\Big)  - \widetilde{Q}^N\Big|\Big|_{1}\\
=& \frac{1}{B}  \Big|\Big|   Q_B - \widetilde{Q}^N \Big|\Big|_{1} \leq  \frac{2}{B}  \Big|   \mc{S} \times \mc{X} \times\mc{Y}  \times \mc{V}   \Big| \leq \varepsilon.
\end{align}
Then, the expected probability that $(S^N,W_1^N,W_2^N, X^N,Y^N,V^N) \notin T_{\delta}(\QQ)$, is upper bounded by:
\begin{align}
&  \E_c\bigg[ \prob\bigg(\Big|\Big|{Q}^N - \QQ  \Big|\Big|_{1}> 2\varepsilon\bigg) \bigg] \nonumber\displaybreak[0]\\
=&  \E_c\bigg[ \prob\bigg(\Big|\Big|{Q}^N  - \widetilde{Q}^N + \widetilde{Q}^N - \QQ \Big|\Big|_{1}> 2\varepsilon\bigg) \bigg] \displaybreak[0]\nonumber\\
\leq&  \E_c\bigg[ \prob\bigg(\Big|\Big|{Q}^N  - \widetilde{Q}^N\Big|\Big| + \Big|\Big|\widetilde{Q}^N - \QQ \Big|\Big|_{1}> 2\varepsilon\bigg) \bigg] \nonumber\displaybreak[0]\\
\leq& \E_c\bigg[ \prob\bigg(  \Big|\Big|\widetilde{Q}^N - \QQ \Big|\Big|_{1}> 2\varepsilon -  \frac{2}{B}  \Big| \mc{S} \times \mc{X} \times\mc{Y}  \times \mc{V}  \Big|  \bigg) \bigg] \displaybreak[0]\nonumber\\
\leq&  \E_c\bigg[ \prob\bigg(  \Big|\Big|\widetilde{Q}^N - \QQ \Big|\Big|_{1}> \varepsilon \bigg) \bigg] \leq  1 - \bigg(1 - 4{\varepsilon}   \bigg)^{B-1}.
\end{align}
Hence, we obtain the following bound on the expected error probability:
\begin{align}
\E_{c} \bigg[\PP_{\textsf{e}}(c)  \bigg] 
=& \E_{c} \bigg[ \prob\bigg( M \neq \hat{M} \bigg) + \prob\bigg(\Big|\Big|Q^n - \QQ \Big|\Big|_{1}> \varepsilon\bigg) \bigg] \nonumber\displaybreak[0]\\
\leq&  2 - 2\bigg(1 - 4{\varepsilon}   \bigg)^{B-1}.\label{eq:ErrorProbaTotalceBis}
\end{align}
This implies \mrb{that for all $\varepsilon>0$, there exists} a code $c^{\star} \in \mc{C}(N)$ with \mrb{$N\geq B \bar{n} $ such that $\PP_{\textsf{e}}(c^{\star})\leq2 - 2\Big(1 - 4{\varepsilon} \Big)^{B-1}$.}


%

\subsection{Expected state leakage rate} \label{sec:LeakageCE}
In this section, we provide \mrb{upper and lower bounds} on the expected state leakage rate, depending on the parameters $\varepsilon_1$ and $\varepsilon_2$ given by equations \eqref{eq:UpperBoundLeakce} and  \eqref{eq:LowerBoundLeakce}.
\begin{small}
 \begin{align}
\textsf{E} -\varepsilon_1 - \varepsilon_2 \leq& \E_{c} \bigg[ \mc{L}_{\textsf{e}}(c)  \bigg]   = \E_{c} \bigg[ \frac{1}{nB} I(S^{nB};Y^{nB}| C = c)    \bigg]  \nonumber\displaybreak[0]\\
=&   \frac{1}{nB} I(S^{nB};Y^{nB}| C)   \leq \textsf{E} + \varepsilon_1 + \varepsilon_2.
\end{align}
\end{small}
\mrb{The notation} $S^{nB}$ denotes the sequence of random variables of channel states of length $N = n B$, whereas $S^{n}_b$ denotes the sub-sequence  of length $n\in\N^{\star}$ over the block $b\in\{1,\ldots,B\}$.

 \textbf{Upper bound.} We provide an upper bound on the expected state leakage rate by \mrb{using} the chain rule from one block to another. 
\begin{align}
&I(S^{nB};Y^{nB} | C)   \nonumber\displaybreak[0]\\
 \leq&   \sum_{b = b_1}^{B-1}  I(S^{n}_b;Y^{nB} | S^{n}_{b+1} , \ldots, S^{n}_B,C) + n \log_2 |\mc{S}|\label{eq:upperB00ce}  \displaybreak[0]\\
\leq&   \sum_{b = b_1}^{B-1}  I(S^{n}_b;Y^{nB},W^n_{1,b},W^n_{2,b},L_b, M_{b+1} , S^{n}_{b+1} , \ldots, S^{n}_B|C)\nonumber\displaybreak[0]\\&+ n \log_2 |\mc{S}|\label{eq:upperB0ce} \displaybreak[0]\\
 = &  \sum_{b = b_1}^{B-1}  I(S^{n}_b;W^n_{1,b},W^n_{2,b},Y^{n}_b,L_b,M_{b+1} | C)+ n \log_2 |\mc{S}| \nonumber \displaybreak[0]\\
 + &  \sum_{b = b_1}^{B-1}  I(S^{n}_b;S^{n}_{b+1} , \ldots, S^{n}_B| W^n_{1,b},W^n_{2,b},Y^{n}_b, L_b,M_{b+1},C) \label{eq:upperB1ce}\displaybreak[0]\\
 + &\!\!  \sum_{b = b_1}^{B-1}  I(S^{n}_b;Y^{nB} | W^n_{1,b},W^n_{2,b} , Y^{n}_b,L_b,M_{b+1},S^{n}_{b+1} , \ldots, S^{n}_B,C) \label{eq:upperB2ce} \displaybreak[0]\\
 = &  \sum_{b = b_1}^{B-1}  I(S^{n}_b;W^n_{1,b},W^n_{2,b},L_b,M_{b+1}| C) + n \log_2 |\mc{S}| \nonumber\displaybreak[0]\\
 &+\sum_{b = b_1}^{B-1} I(S^{n}_b;Y^{n}_b|W^n_{1,b},W^n_{2,b},L_b,M_{b+1}, C). \label{eq:upperB3ce}
\end{align}
The term at line \eqref{eq:upperB1ce} is equal to zero since the causal encoding and the i.i.d. property of the channel state, imply that the random variables $(S^{n}_{b+1} , \ldots, S^{n}_B)$ are independent of the message $M_{b+1}$ and of the random variables $(S^n_b,W^n_{1,b},W^n_{2,b},Y^{n}_b,L_b)$ of the block $b$, and the term at line \eqref{eq:upperB2ce} is equal to zero since in the encoding process, the sequence $S^n_b$ only affects the choice of the bin index $L_b$ and of the sequences $(W^n_{1,b},W^n_{2,b}, Y^{n}_b)$ of the current block $b$. This induces the following Markov chain: $S^n_b -\!\!\!\!\minuso\!\!\!\!- (W^n_{1,b},W^n_{2,b}, Y^{n}_b,L_b,C)   -\!\!\!\!\minuso\!\!\!\!-   (Y^{nB} ,M_{b+1},S^{n}_{b+1} , \ldots, S^{n}_B) $, that is valid for each block $b\in \{1,\ldots,B-1\}$. The sequence $S^n_b$ is correlated with the random variables of the other blocks $b'\neq b$ only through $(W^n_{1,b},W^n_{2,b}, Y^{n}_b,L_b,C)$. 

For each block $b \in \{1, \ldots,B-1\}$, the first term in equation \eqref{eq:upperB3ce} satisfies: 
\begin{align} 
&I(S^{n}_b;W^n_{1,b},W^n_{2,b},L_b ,M_{b+1} | C) \nonumber\displaybreak[0]\\
=& I(S^{n}_b;W^n_{2,b},L_b  | W^n_{1,b} ,M_{b+1} ,C) \label{eq:AchSCEleak1ce}\displaybreak[0]\\
\leq& H(W^n_{2,b},L_b  | W^n_{1,b} ,M_{b+1} ,C) \label{eq:AchSCEleak2ce}\displaybreak[0]\\
\leq& \log_2|\mc{M}_{\textsf{L}}| +  H(W^n_{2,b}| W^n_{1,b},L_b ,M_{b+1} ,C)    \label{eq:AchSCEleak3ce}\displaybreak[0]\\
\leq& \log_2|\mc{M}_{\textsf{L}}| +  \log_2|\mc{M}_{\textsf{K}}|   \label{eq:AchSCEleak4ce}\displaybreak[0]\\
=&  n \bigg(\textsf{E} - I(S;W_1,W_2,Y)  - 2 \varepsilon + I(S;W_2|W_1) + \varepsilon \bigg)\label{eq:AchSCEleak5ce}\displaybreak[0]\\ 
=& n \bigg(\textsf{E} - I(S;Y|W_1,W_2)  - \varepsilon \bigg), \label{eq:AchSCEleak6ce}
\end{align}
where \eqref{eq:AchSCEleak1ce} comes from the causal encoding that induces the independence between the sequence auxiliary random variables  $W^n_{1,b}$  and the sequence of channel states $S^{n}_b$, in block $b \in \{1,\ldots, B-1\}$. Hence $S^n_b$ is independent of $(W^n_{1,b} ,M_{b+1} ,C)$; \eqref{eq:AchSCEleak3ce} comes from the cardinality of the set of \mrb{indices} $\mc{M}_{\textsf{L}}$; \eqref{eq:AchSCEleak5ce} comes from the coding scheme described in Appendix \ref{sec:AchievabilityLeakageCE}. By considering a fixed sequence $W^n_{1,b}$ and fixed \mrb{indices} $(L_b ,M_{b+1})$, the encoder chooses an index $K_{b+1} \in \mc{M}_{\sf{K}} $ corresponding to the sequence $W^n_{2,b}$. Hence, the sequence $W^n_{2,b}$ belongs to the bin of cardinality $|\mc{M}_{\sf{K}}|$. The  rate parameters  are given by $\textsf{R}_{\sf{L}} =\textsf{E} - I(S;W_1,W_2,Y) - 2 \varepsilon $ and $\textsf{R}_{\sf{K}} = I(W_2;S|W_1) +  \varepsilon $; \eqref{eq:AchSCEleak6ce} comes from the independence between $W_1$ and $S$ that induces $I(S;W_2|W_1) = I(S;W_1,W_2)$.


For each block $b \in \{1, \ldots,B-1\}$, we introduce the random event of error $E_b \in \{0,1\}$ defined with respect to the achievable distribution $\QQ_{SXW_1W_2YV}$, as follows:
\begin{align}
E_b = \left\{
\begin{array}{lll}
0 \text{ if}&  (S_b^n,X_b^n,W_{1,b}^n,W_{2,b}^n,Y_b^n,V_b^n) \in T_{\delta}(\QQ) \text{ and }\\& (\hat{M}_{b+1},\hat{L}_{b}, \hat{K}_{b+1}) = (M_{b+1},L_{b}, K_{b+1}),\\
1 \text{ if}&  (S_b^n,X_b^n,W_{1,b}^n,W_{2,b}^n,Y_b^n,V_b^n) \notin T_{\delta}(\QQ) \text{ or }\\& (\hat{M}_{b+1},\hat{L}_{b}, \hat{K}_{b+1}) \neq (M_{b+1},L_{b}, K_{b+1}).
\end{array}
\right.\label{eq:ErrorEventBlockB}
\end{align}
The second term in equation \eqref{eq:upperB3ce} satisfies: 
\begin{align}
&I(S^{n}_b;Y^{n}_b|W^n_{1,b},W^n_{2,b},L_b,M_{b+1}  , C) \nonumber \displaybreak[0]\\
=&  H(Y^{n}_b|W^n_{1,b},W^n_{2,b},L_b ,M_{b+1}  , C) \nonumber\displaybreak[0]\\
&-  H(Y^{n}_b|S^{n}_b, W^n_{1,b},W^n_{2,b},L_b ,M_{b+1}  , C) \label{eq:BchSCEleak1ce}  \displaybreak[0]\\
=&  H(Y^{n}_b|W^n_{1,b},W^n_{2,b},L_b  ,M_{b+1}  , C) -  n H(Y|W_1,W_2,S) \label{eq:BchSCEleak2ce}  \displaybreak[0]\\
\leq&  H(Y^{n}_b|W^n_{1,b},W^n_{2,b},L_b ,M_{b+1}, C,E_b=0) +  1 \nonumber\displaybreak[0]\\
&+  \prob(E_b=1) n \log_2|\mc{Y}| -  n H(Y|W_1,W_2,S) \label{eq:BchSCEleak3ce}  \displaybreak[0]\\
\leq&   n \bigg( H(Y|W_1,W_2) + \varepsilon \bigg)  + 1 +  \prob(E_b=1) n \log_2|\mc{Y}| \nonumber\displaybreak[0]\\
&-  n H(Y|W_1,W_2,S) \label{eq:BchSCEleak4ce}  \displaybreak[0]\\
=& n \bigg( I(S;Y|W_1,W_2) + \varepsilon + \frac{1}{n}  +  \prob(E_b=1) \log_2|\mc{Y}|  \bigg) ,\label{eq:BchSCEleak5ce}
\end{align}
where \eqref{eq:BchSCEleak1ce} comes from the properties of the mutual information; \eqref{eq:BchSCEleak2ce} comes from the cascade of memoryless channels $\QQ_{X|W_1S}   \mc{T}_{Y|XS}$ of the coding scheme $C$ described in Appendix \ref{sec:AchievabilityLeakageCE}, that implies $H(Y^{n}_b|S^{n}_b, W^n_{1,b},W^n_{2,b}  ,L_b ,M_{b+1} , C) = nH(Y|W_1,S)= nH(Y|W_1,W_2,S)$; \eqref{eq:BchSCEleak3ce} is inspired by the proof of Fano's inequality, see \cite[pp. 19]{ElGamalKim(book)11}; \eqref{eq:BchSCEleak4ce} comes from the \mrb{cardinality bound for} the set of sequences $y^n$ such that $(w_1^n,w_2^n,y^n) \in T_{\delta}(\QQ)$, see \cite[pp. 26]{ElGamalKim(book)11}; \eqref{eq:BchSCEleak5ce} comes from the properties of the mutual information.

Hence we have the following upper bound on the leakage rate:
\begin{align*}
&n B \E_{c} \bigg[\mc{L}_{\textsf{e}}(c)  \bigg] 
 \leq   \sum_{b = b_1}^{B-1}  I(S^{n}_b;W^n_{1,b},W^n_{2,b},L_b,M_{b+1}| C) \nonumber\displaybreak[0]\\
 &+ \sum_{b = b_1}^{B-1} I(S^{n}_b;Y^{n}_b|W^n_{1,b},W^n_{2,b},L_b,M_{b+1}, C) + n \log_2 |\mc{S}| \displaybreak[0]\\
 \leq & n B \bigg(\textsf{E} - I(S;Y|W_1,W_2)  - \varepsilon +  I(S;Y|W_1,W_2) + \varepsilon + \frac{1}{n} \nonumber\displaybreak[0]\\
 & +  \prob(E_b=1) \log_2|\mc{Y}|  + \frac{1}{B} \log_2|\mc{S}| \bigg)  \displaybreak[0]\\
 \leq & n B \bigg(\textsf{E} + \frac{1}{n}  +  \prob(E_b=1) \log_2|\mc{Y}| + \frac{1}{B} \log_2|\mc{S}| \bigg).
\end{align*}
 \textbf{Lower bound.} We provide a lower bound on the expected state leakage rate. 
\begin{align}
&n B  \E_{c} \bigg[\mc{L}_{\textsf{e}}(c)  \bigg]  =  I(S^{nB};Y^{nB}|C)  \label{eq:CchSCEleak1ce}\\
 =&  n B H(S) - H(S^{nB}|Y^{nB},C)  \label{eq:CchSCEleak2ce}\displaybreak[0]\\
 \geq & nB H(S)   - \prob(E_b=1) n  B\log_2 |\mc{S}|- 1  \nonumber\displaybreak[0]\\
 & - H(S^{nB}|Y^{nB},C,E_b=0)  \label{eq:CchSCEleak5ce}\displaybreak[0]\\
 \geq & nB H(S)   - \prob(E_b=1) n  B\log_2 |\mc{S}|- 1- n   \log_2 |\mc{S}|\nonumber\displaybreak[0]\\
 &-  \sum_{b=b_1}^{B-1} H(S^{n}_b|Y^{nB},S^{n}_{b+1},\ldots, S^{n}_B,C,E_b=0)  \label{eq:CchSCEleak6ce}\displaybreak[0]\\
 = & nB H(S)   - \Big(\prob(E_b=1) B + 1 \Big)  n \log_2 |\mc{S}|- 1\nonumber \displaybreak[0]\\
-&  \sum_{b=b_1}^{B-1} H(S^{n}_b|W^n_{1,b},W^n_{2,b},L_b,Y^{nB},S^{n}_{b+1},\ldots, S^{n}_B,C,E_b=0) \nonumber\displaybreak[0] \\
-&  \sum_{b=b_1}^{B-1} I(S^{n}_b ;W^n_{1,b},W^n_{2,b},L_b |Y^{nB},S^{n}_{b+1},\ldots, S^{n}_B,C,E_b=0) \label{eq:CchSCEleak7ce}\displaybreak[0]\\
 = & nB H(S)   - \Big(\prob(E_b=1) B + 1 \Big)  n \log_2 |\mc{S}|- 1\nonumber\displaybreak[0]\\
 -&  \sum_{b=b_1}^{B-1} H(S^{n}_b|W^n_{1,b},W^n_{2,b},L_b,Y^{nB},S^{n}_{b+1},\ldots, S^{n}_B,C,E_b=0) \label{eq:CchSCEleak8ce}\displaybreak[0]\\
 \geq & nB H(S)   - \Big(\prob(E_b=1) B + 1 \Big)  n \log_2 |\mc{S}|- 1 \nonumber\displaybreak[0]\\
 -&  \sum_{b=b_1}^{B-1} H(S^{n}_b|W^n_{1,b},W^n_{2,b},L_b,Y^{n}_b,E_b=0) , \label{eq:CchSCEleak9ce} 
\end{align}
where  \eqref{eq:CchSCEleak2ce} comes from the i.i.d. property of the channel states $S$;  \eqref{eq:CchSCEleak5ce} is inspired by the proof of Fano's inequality, see \cite[pp. 19]{ElGamalKim(book)11}; \eqref{eq:CchSCEleak8ce} comes from the non-error event $E_b=0$, that implies for all block $b\in \{1,\ldots,B-1\}$, the sequences $(W^n_{1,b},W^n_{2,b},L_b)$ are correctly decoded based on the observation of $Y^{nB}$, \mrb{hence} $  I(S^{n}_b ;W^n_{1,b},W^n_{2,b},L_b |Y^{nB},S^{n}_{b+1},\ldots, S^{n}_B,C,E_b=0) =0$; \eqref{eq:CchSCEleak9ce} comes from removing the conditioning over the sequences $(Y^{n}_{b_1},\ldots,Y^{n}_{b-1},Y^{n}_{b+1},\ldots,Y^{n}_{B},S^{n}_{b+1},\ldots, S^{n}_B,C)$ and the random code in the conditional entropy $H(S^{n}_b|W^n_{1,b},W^n_{2,b},L_b,Y^{n}_b,E_b=0) $.

In order to provide an upper bound on $H(S^{n}_b|W^n_{1,b},W^n_{2,b},L_b,Y^{n}_b,E_b=0) $, we fix an index $l \in \mc{M}_{\sf{L}}$, some typical sequences $(w_1^n,w_2^n,y^n) \in T_{\delta}(\QQ)$. We \mrb{define} the set 
\begin{align}
&\mc{S}^{\star}(w_1^n,w_2^n,y^n,l) = \bigg\{s^n \in  \mc{S}^n, \text{ s.t. }\nonumber\displaybreak[0]\\
& \Big\{( s^n, w_1^n,w_2^n,y^n) \in T_{\delta}(\QQ) \Big\}\cap \Big\{s^n \in \mc{B}(l) \Big\}\bigg\}
\end{align}
Since the code $C$ is random, the  above set $\mc{S}^{\star}(w_1^n,w_2^n,y^n,l) $ is \mrb{also} random  and \mrb{we have}
\begin{align}
&\E_c\bigg[ \Big|\mc{S}^{\star}(W_1^n,W_2^n,Y^n,L)\Big| \bigg] \nonumber\\
=& \E_c\Bigg[ \bigg| \bigg\{s^n \in  \mc{S}^n, \text{ s.t. }  \Big\{  ( S^n, W_1^n,W_2^n,Y^n)  \in T_{\delta}(\QQ)  \Big\}\nonumber\displaybreak[0]\\
&\cap \Big\{S^n \in \mc{B}(L) \Big\} \bigg\}\bigg| \Bigg]\displaybreak[0]\\
=&    \sum_{ S^n  \in  T_{\delta}(\QQ)  (W_1^n,W_2^n, Y^n)   } \E_c\Bigg[ \UN\Big( S^n \in \mc{B}(L)\Big) \Bigg]\displaybreak[0]\\
\leq&  \sum_{ S^n  \in \atop T_{\delta}(\QQ)  (W_1^n,W_2^n, Y^n)  } 2^{-n \cdot\textsf{R}_{\sf{L}} } 
\leq  2^{n \cdot(H(S|W_1,W_2,Y)  - \textsf{R}_{\sf{L}} +\varepsilon) },\label{eq:AchMarkovIneq2ce}
\end{align}
\mrb{where} \eqref{eq:AchMarkovIneq2ce} comes from the definition of the random code that induces a uniform distribution over the bins $ \mc{B}(l)$ and the properties of the typical sequences, see \cite[pp. 27]{ElGamalKim(book)11}.

By Markov's inequality, for all $l\in \mc{M}_L$ and for all $( w_1^n,w_2^n,y^n)  \in T_{\delta}(\QQ)  $ \mrb{we have}
\begin{align}
&\prob\bigg[ \Big|\mc{S}^{\star}(W_1^n,W_2^n,Y^n,L)\Big| \geq 2^{n \cdot(H(S|W_1,W_2,Y)  - \textsf{R}_{\sf{L}} +2\varepsilon) } \bigg] \nonumber\displaybreak[0]\\
\leq& \frac{\E_c\bigg[ \Big|\mc{S}^{\star}(W_1^n,W_2^n,Y^n,L)\Big| \bigg] }{2^{n \cdot(H(S|W_1,W_2,Y)  - \textsf{R}_{\sf{L}} +2\varepsilon) }}\displaybreak[0]\\
\leq& \frac{2^{n \cdot(H(S|W_1,W_2,Y)  - \textsf{R}_{\sf{L}} +\varepsilon) } }{2^{n \cdot(H(S|W_1,W_2,Y)  - \textsf{R}_{\sf{L}} +2\varepsilon) }}
\leq  2^{ - n \varepsilon } \leq \varepsilon.\label{eq:LeakMarkovce}
\end{align}
The last equation is valid for $n\geq n_5\mrb{= \frac{1}{\varepsilon} \log_2 \frac{1}{\varepsilon}}$. For each block $b \in \{1, \ldots,B-1\}$ we \mrb{define the event}
\begin{align}
F_b =& 
\begin{cases}
1 \text{ if } &\big|\mc{S}^{\star}(W_{1,b}^n,W_{2,b}^n,Y_b^n,L_b)\big| \\
&\geq 2^{n \cdot(  H(S|W_1,W_2,Y)  - \textsf{R}_{\sf{L}} +2\varepsilon) },\\
0 \text{ if }& \big|\mc{S}^{\star}(W_{1,b}^n,W_{2,b}^n,Y_b^n,L_b)\big| \\
&<      2^{n \cdot(  H(S|W_1,W_2,Y)  - \textsf{R}_{\sf{L}} +2\varepsilon) }.\\
\end{cases}
\end{align}
\mrb{Then for} each block $b \in \{b_1,\ldots,B-1\}$ we have
\begin{align}
&H(S^n_b|W^n_{1,b},W^n_{2,b},Y^{n}_b,L_b ,E_b=0) \nonumber \\
\leq& 1 + \prob(F_b=1)n \log_2|\mc{S}| \nonumber\displaybreak[0]\\
&+  H(S^n_b|W^n_{1,b},W^n_{2,b},Y^{n}_b,L_b ,E_b=0,F_b=0) \label{eq:MarkovLemma0ce}\displaybreak[0]\\
\leq& 1 + \varepsilon \log_2|\mc{S}| \nonumber\displaybreak[0]\\
&+  H(S^n_b|W^n_{1,b},W^n_{2,b},Y^{n}_b,L_b ,E_b=0,F_b=0) \label{eq:MarkovLemma1ce}\displaybreak[0]\\
\leq& 1 + \varepsilon \log_2|\mc{S}| +  \log_2 2^{n \cdot(H(S|W_1,W_2,Y)  - \textsf{R}_{\sf{L}} +2\varepsilon) } \label{eq:MarkovLemma2ce}\displaybreak[0]\\
=& n \cdot\bigg( H(S|W_1,W_2,Y)  - \textsf{R}_{\sf{L}}  + \frac{1}{n} + \varepsilon \log_2|\mc{S}|    +2\varepsilon \bigg) \label{eq:MarkovLemma3ce}\displaybreak[0]\\
=& n \cdot\bigg( H(S|W_1,W_2,Y)  - \textsf{E} + I(S;W_1,W_2,Y) + 2 \varepsilon  + \frac{1}{n}\nonumber\displaybreak[0]\\
& + \varepsilon \log_2|\mc{S}|    +2\varepsilon \bigg)\label{eq:MarkovLemma4ce}\displaybreak[0]\\
=& n \cdot\bigg( H(S)  - \textsf{E}  + \frac{1}{n} + \varepsilon  \log_2|\mc{S}|    +4\varepsilon \bigg),\label{eq:MarkovLemma5ce}
\end{align}
where \eqref{eq:MarkovLemma0ce} is inspired by the proof of Fano's inequality, see \cite[pp. 19]{ElGamalKim(book)11};  \eqref{eq:MarkovLemma1ce} comes from equation \eqref{eq:LeakMarkovce}, that corresponds to $\prob(F_b=1) \leq \varepsilon$; \eqref{eq:MarkovLemma2ce} comes from the definition of event $F_b=0$, that implies the set $\mc{S}^{\star}(w_1^n,w_2^n,y^n,l)$ has  cardinality bounded by $2^{n \cdot(  H(S|W_1,W_2,Y)  - \textsf{R}_{\sf{L}} +2\varepsilon) }$; \eqref{eq:MarkovLemma4ce} comes from the definition of the rate $\textsf{R}_{\sf{L}} = \textsf{E} - I(S;W_1,W_2,Y) - 2 \varepsilon$ stated in equation \eqref{eq:rateRLce}.

This provides the following lower bound:
\begin{align}
&n B \E_{c} \bigg[\mc{L}_{\textsf{e}}(c)  \bigg] \nonumber\displaybreak[0]\\
\geq& nB H(S)   - \Big(\prob(E_b=1) B + 1 \Big)  n \log_2 |\mc{S}|- 1\nonumber\displaybreak[0]\\
&-   \sum_{b=b_1}^{B-1} H(S^{n}_b|W^n_{1,b},W^n_{2,b},Y^{n}_b,E_b=0) \label{eq:LowerBoundFinal1ce}\displaybreak[0] \\
\geq& nB H(S)   - \Big(\prob(E_b=1) B + 1 \Big)  n \log_2 |\mc{S}|- 1\nonumber\displaybreak[0]\\
&-  (B-1) n \bigg( H(S)  - \textsf{E}  + \frac{1}{n} + \varepsilon  \log_2|\mc{S}|    +4\varepsilon \bigg) \label{eq:LowerBoundFinal2ce} \displaybreak[0]\\
\geq& nB H(S)   - \Big(\prob(E_b=1) B + 1 \Big)  n \log_2 |\mc{S}|- 1 \nonumber\displaybreak[0]\\
&-  B n \bigg( H(S)  - \textsf{E}  + \frac{1}{n} + \varepsilon  \log_2|\mc{S}|    +4\varepsilon \bigg) - n \log_2 |\mc{S}|\label{eq:LowerBoundFinal3ce}\displaybreak[0] \\
\geq& nB \bigg(\textsf{E} - \bigg(\prob(E_b=1) + \varepsilon + \frac{2 }{B}\bigg) \log_2 |\mc{S}| - \frac{2}{nB} - 4\varepsilon \bigg),\label{eq:LowerBoundFinal4ce}
\end{align}
where \eqref{eq:LowerBoundFinal1ce} comes from equation \eqref{eq:CchSCEleak9ce}; \eqref{eq:LowerBoundFinal2ce} comes from equation \eqref{eq:MarkovLemma5ce}; \eqref{eq:LowerBoundFinal3ce} comes the lower bound: 
\begin{align}
n\Big(H(S)  - \textsf{E}  + \frac{1}{n} + \varepsilon  \log_2|\mc{S}|    +4\varepsilon\Big) \geq - n  \textsf{E}\geq - n   \log_2 |\mc{S}|.
\end{align}
Equation \eqref{eq:LowerBoundFinal4ce} provides the lower bound on the expected state leakage rate.

 \textbf{Conclusion for Exact state leakage Rate:} 
 \begin{align}
\E_{c} \bigg[\mc{L}_{\textsf{e}}(c)  \bigg]   \leq & \textsf{E} + \frac{1}{n}  +  \max_b\prob(E_b=1) \log_2|\mc{Y}| + \frac{1}{B} \log_2|\mc{S}| \nonumber\displaybreak[0]\\
=&  \textsf{E} + \varepsilon_2, \label{eq:UpperBoundLeakce} \\
\E_{c} \bigg[\mc{L}_{\textsf{e}}(c)  \bigg] \geq&   \textsf{E} - \bigg(\max_b\prob(E_b=1) + \varepsilon + \frac{2 }{B}\bigg) \log_2 |\mc{S}| \nonumber\displaybreak[0]\\ 
&- \frac{2}{nB} - 4\varepsilon  = \textsf{E} - \varepsilon_1. \label{eq:LowerBoundLeakce} 
 \end{align}
\mrb{By defining $\varepsilon_1$ and $\varepsilon_2$ as in \eqref{eq:LowerBoundLeakce} and \eqref{eq:UpperBoundLeakce}, we have}
  \begin{align}
\Bigg| \E_{c} \Big[\mc{L}_{\textsf{e}}(c)  \Big]  -   \textsf{E} \Bigg| \leq&    \varepsilon_1 +  \varepsilon_2.\label{eq:BoundLeakce} 
  \end{align}

  %
 
\subsection{Conclusion of the achievability proof of Theorem \ref{theo:LeakageCE} }\label{sec:ConclusionTotalLengthCE}

 By choosing \mrb{appropriately the number of blocks $B$, the} error probability $\max_b\prob(E_b=1)$, \mrb{the block-length} $n$ and \mrb{the tolerance for the typical sequences $\delta$, we show the existence of} a code $c^{\star} \in   \C(n B,\mc{M})$  such that
\begin{align}
\frac{1}{n B} \sum_{b=2}^{B} \log_2 |\mc{M}| \geq& \textsf{R} -  \varepsilon_3,\\
 \PP_{\textsf{e}}(c^{\star})     \leq& \varepsilon_3,\\
\Bigg|  \mc{L}_{\textsf{e}}(c^{\star})     -   \textsf{E} \Bigg| \leq&   \varepsilon_3,
 \end{align}
 with $\varepsilon_3 =  \frac{1}{ B}  \log_2 |\mc{Y}| + 2 - 2 \Big( 1 - 4 \varepsilon \Big)^{B-1} +  \varepsilon_1 +  \varepsilon_2$. This concludes the achievability proof of  Theorem \ref{theo:LeakageCE}.
 

%

\section{Converse proof of Theorem \ref{theo:LeakageCE} }\label{sec:ConverseLeakageCE}

\subsection{Information Constraints }\label{sec:InformationConstraints}

Consider that the triple of rate, state leakage and distribution $(\textsf{R},\textsf{E},\mrb{\QQ_{SXYV}})$ is achievable \mrb{by using a sequence of codes with causal encoding}. \mrb{We simplify the notation by using $\QQ$ in place of $\QQ_{SXYV}$ and} we introduce the random  event of error $E \in \{0,1\}$ defined with respect to $\QQ$ as follows:
\begin{align}
E = \Bigg\{
\begin{array}{lll}
0 \text{ if }& 
(S^n,X^n,Y^n,V^n) \in T_{\delta}(\QQ) ,\\
1 \text{ if }& 
(S^n,X^n,Y^n,V^n) \notin T_{\delta}(\QQ).
\end{array}
\Bigg.\label{eq:ErrorEvent}
\end{align}
The event $E=1$ occurs if the sequences $(S^n,X^n,Y^n,V^n)\notin T_{\delta}(\QQ)$ for the target distribution $\QQ$. By definition \ref{def:CodeLeakageCE}, for all $\varepsilon>0$, there exists $\bar{n}\in\N^{\star}$ such that for all $n\geq \bar{n}$, there exists a code $c^{\star} \in \C(n,\mc{M})$ that satisfies the three following equations:
\begin{align}
\frac{\log_2 |\mc{M}|}{n}  \geq& \textsf{R} - \varepsilon, \label{eq:converse1ce}  \\
 \bigg| \mc{L}_{\textsf{e}}(c^{\star})  - \textsf{E} \bigg|  =& \bigg|  \frac{1}{n} I(S^n;Y^n)  - \textsf{E} \bigg| \leq  \varepsilon,   \label{eq:converse3ce} \\
\PP_{\textsf{e}}(c^{\star})  =& \prob\bigg( M \neq \hat{M} \bigg) + \prob\bigg(\Big|\Big|Q^n - \QQ \Big|\Big|_{1}> \varepsilon\bigg) \leq \varepsilon. \label{eq:converse2ce}  
\end{align}
We introduce the auxiliary random variables $ W_{1,i} =  ( M , S^{i-1} )$ and $ W_{2,i} =  Y_{i+1}^n $ that satisfy the Markov chains of the set of distribution $\Q_{\sf{c}}$ for all $i \in \{1,\ldots,n\}$:
\begin{align}
&S_i \text{ independent of } W_{1,i}, \label{eq:ConverseCEMarkov1} \\
&X_i -\!\!\!\!\minuso\!\!\!\!- (S_i ,W_{1,i})-\!\!\!\!\minuso\!\!\!\!-  W_{2,i} ,  \label{eq:ConverseCEMarkov2}  \\ 
&Y_i -\!\!\!\!\minuso\!\!\!\!- (X_i , S_i ) -\!\!\!\!\minuso\!\!\!\!-  ( W_{1,i},W_{2,i} ) ,  \label{eq:ConverseCEMarkov3}  \\
&V_i -\!\!\!\!\minuso\!\!\!\!- (Y_i,  W_{1,i},W_{2,i} ) -\!\!\!\!\minuso\!\!\!\!-  (S_i ,X_i) ,  \label{eq:ConverseCEMarkov4} 
\end{align}
where \eqref{eq:ConverseCEMarkov1} comes from the i.i.d. property of the source that induces the independence between $S_i$ and  $  ( M ,  S^{i-1} ) = W_{1,i} $; \eqref{eq:ConverseCEMarkov2} comes from Lemma  \ref{lemma1:ConverseCEMarkov1} \mrb{since the encoding is causal and the channel is memoryless}; \eqref{eq:ConverseCEMarkov3} \mrb{also} comes from the memoryless property of the channel $ \mc{T}_{Y|XS}$; \eqref{eq:ConverseCEMarkov4} comes from Lemma \ref{lemma2:ConverseCEMarkov1} \mrb{since the encoding is causal, the decoding is non-causal and the channel is memoryless.}

We introduce the random variable $T$ that is uniformly distributed over the indices $\{1 , \ldots, n\}$ and \mrb{we denote by} $W_{1,T}$, $W_{2,T}$, $S_T$, $X_T$, $Y_T$, $V_T$ \mrb{the corresponding averaged random variables}. The  auxiliary random variables $W_1=(W_{1,T} , T)$ and $W_2 = W_{2,T}$ belong to the set of distributions $\Q_{\sf{c}}$  and satisfy the three information constraints of Theorem \ref{theo:LeakageCE}:
\begin{align}
 I(S;W_1,W_2,Y) \leq \textsf{E}&\leq H(S),\label{eq:theoremCEconverse2b}\\
 \textsf{R} + \textsf{E} &\leq I(W_1,S; Y). \label{eq:theoremCEconverse3b}
\end{align}

 \textbf{First Constraint:}
\begin{align}
n \textsf{E} \geq & I(S^n ; Y^n) -  n\varepsilon \label{eq:SecondConverse00ce}\displaybreak[0]\\
=&I(S^n ; Y^n,M) - I(S^n ; M|Y^n) -  n\varepsilon \displaybreak[0]\\
\geq& \sum_{i=1}^n  I(S_i;Y^n , M | S^{i-1}) -  H(M | Y^n)-  n\varepsilon \label{eq:SecondConverse3ce}\displaybreak[0]\\
\geq& \sum_{i=1}^n  I(S_i;Y^n , M | S^{i-1}) -  n2\varepsilon \label{eq:SecondConverse4ce}\displaybreak[0]\\
=& \sum_{i=1}^n  I(S_i;Y^n , M , S^{i-1}) -    n2\varepsilon 
\displaybreak[0]\\
\geq& \sum_{i=1}^n  I(S_i;Y_{i+1}^n , M , S^{i-1},Y_i) -    n2\varepsilon \label{eq:SecondConverse6ce}\displaybreak[0]\\
=& \sum_{i=1}^n  I(S_i; W_{1,i} , W_{2,i},Y_i ) -    n2\varepsilon \label{eq:SecondConverse7ce}\displaybreak[0]\\
=& n   I(S_T; W_{1,T} , W_{2,T} ,Y_T| T ) -    n2\varepsilon \label{eq:SecondConverse8ce}\displaybreak[0]\\
=& n   I(S_T; W_{1,T} , W_{2,T},Y_T, T ) -    n2\varepsilon \label{eq:SecondConverse9ce}\displaybreak[0]\\
= & n  \bigg(  I(S_T ; W_1,W_2,Y_T)  -  2\varepsilon \bigg)  \label{eq:SecondConverse10ce}\displaybreak[0]\\
\geq& n  \bigg(  I(S_T ; W_1,W_2,Y_T | E=0)  -  3\varepsilon \bigg)  \label{eq:SecondConverse11ce}\displaybreak[0]\\
\geq & n  \bigg(  I(S ; W_1,W_2,Y)  -  4\varepsilon \bigg)  \label{eq:SecondConverse12ce},
\end{align}
where  \eqref{eq:SecondConverse00ce} comes from the definition of achievable state leakage rate $\textsf{E}$ in \eqref{eq:converse3ce}; \eqref{eq:SecondConverse4ce} comes from equation \eqref{eq:converse2ce} and  Fano's inequality, see \cite[pp. 19]{ElGamalKim(book)11}; \eqref{eq:SecondConverse6ce}  comes from the i.i.d. property of the channel states that implies $S_i$ is independent of $S^{i-1}$; \eqref{eq:SecondConverse7ce}  comes from the introduction of the auxiliary random variables $ W_{1,i} =  ( M,S^{i-1}  )$ and $ W_{2,i} =   Y_{i+1}^n $, for all $i \in \{1,\ldots,n\}$; \eqref{eq:SecondConverse8ce}   comes from the introduction of the uniform random variable $T$ over $\{1 , \ldots, n\}$ and the corresponding mean random variables  $S_T$,  $W_{1,T}$, $W_{2,T}$,  $Y_T$; \eqref{eq:SecondConverse9ce}   comes from the independence between $T$ and $S_T$; \eqref{eq:SecondConverse10ce}   comes from \mrb{the identification of} the auxiliary random variables $W_1=(W_{1,T} , T)$ and $W_2 = W_{2,T}$; \eqref{eq:SecondConverse11ce}  comes from the empirical coordination requirement as stated in Lemma \ref{lemma:ErrorEventCoordination01}. The sequences of symbols $(S^n,X^n,Y^n,V^n)$ are not jointly typical with small error probability $\prob(E=1)$; \eqref{eq:SecondConverse12ce} comes from Lemma \ref{lemma:3} since the distribution $ \PP_{S_TX_TY_TV_T|E=0}$ is close to the target distribution $\QQ_{SXYV}$. The \mrb{result of} \cite[Lemma 2.7, pp. 19]{CsiszarKorner(Book)11} concludes.

 \textbf{Second Constraint:}
\begin{align}
n \textsf{E} \leq & I(S^n ; Y^n) + n\varepsilon \leq  H(S^n)  + n\varepsilon \label{eq:ThirdConverse2ce}\\
= & n  \bigg( H(S) + \varepsilon \bigg), \label{eq:ThirdConverse3ce}
\end{align}
where \eqref{eq:ThirdConverse2ce} comes from the definition of the achievable state leakage rate $\textsf{E}$, stated in equation \eqref{eq:converse3ce}; \eqref{eq:ThirdConverse3ce} comes from the i.i.d. property of the channel states $S$.

 \textbf{Third Constraint:}
\begin{align}
&n \Big( \textsf{E} +  \textsf{R} \Big) \nonumber\displaybreak[0]\\
\leq & I(S^n ; Y^n)  + H(M) +  n2\varepsilon \displaybreak[0]\label{eq:FourthConverse1ce}\\
= &  I(S^n ; Y^n)  + I(M ; Y^n ) +  H(M |Y^n) +  n2\varepsilon \label{eq:FourthConverse2ce}\\
\leq &  I(S^n ; Y^n)  + I(M ; Y^n ) +  n3\varepsilon \displaybreak[0]\label{eq:FourthConverse3ce}\\
\leq &  I(S^n ; Y^n)  + I(M ; Y^n | S^n) +  n3\varepsilon \displaybreak[0]\label{eq:FourthConverse4ce}\\
=&  I(S^n ,M ; Y^n)  +  n3\varepsilon 
\displaybreak[0]\label{eq:FourthConverse5ce}\\
\leq & \sum_{i=1}^n I(S_i ,M , S^{i-1}  ; Y_i  ) \nonumber\\
&+  \sum_{i=1}^n I( S^n_{i+1}   , Y^n_{i+1} ; Y_i  | S_i ,M , S^{i-1}  )  +  n3\varepsilon \label{eq:FourthConverse8ce}\\
= & \sum_{i=1}^n I(S_i ,M , S^{i-1}  ; Y_i  )   +  n3\varepsilon \displaybreak[0]\label{eq:FourthConverse9ce}\\
= &  \sum_{i=1}^n I(S_i ,W_{1,i} ; Y_i)  +  n\cdot3\varepsilon  \displaybreak[0]\label{eq:FourthConverse11ce}\\
= &  n   I(S_T ,W_{1,T} ; Y_T | T)   +  n\cdot3\varepsilon \displaybreak[0]\label{eq:FourthConverse12ce}\\
\leq &  n   I(S_T ,W_{1,T} ,T; Y_T)   +  n\cdot3\varepsilon \displaybreak[0]\label{eq:FourthConverse13ce}\\
\leq &  n   \bigg( I(S_T ,W_{1} ; Y_T)   +  3\varepsilon \bigg)\displaybreak[0]\label{eq:FourthConverse14ce}\\
\leq &  n   \bigg( I(S_T ,W_{1} ; Y_T |E = 0)   +  4\varepsilon \bigg)\displaybreak[0]\label{eq:FourthConverse15ce}\\
\leq &  n   \bigg( I(S ,W_{1} ; Y )   +  5\varepsilon \bigg)\label{eq:FourthConverse16ce},
\end{align}
where \eqref{eq:FourthConverse1ce} comes from the definition of achievable rate and information  leakage $( \textsf{R} ,  \textsf{E})$, stated in equations \eqref{eq:converse1ce} and  \eqref{eq:converse3ce}; \eqref{eq:FourthConverse3ce} comes from equation \eqref{eq:converse2ce} and  Fano's inequality, stated in \cite[pp. 19]{ElGamalKim(book)11}; \eqref{eq:FourthConverse5ce}  comes from the independence between the message $M$ and the channel states $S^n$, hence $I(M ; Y^n ) \leq I(M ; Y^n ,S^n ) = I(M ; Y^n |S^n ) $; \eqref{eq:FourthConverse8ce}  comes from the properties of the mutual information; \eqref{eq:FourthConverse9ce}  comes from the Markov chain $Y_i   -\!\!\!\!\minuso\!\!\!\!- ( S_i ,M , S^{i-1} )   -\!\!\!\!\minuso\!\!\!\!- (S^n_{i+1}   , Y^n_{i+1} )$, stated in Lemma \ref{lemma3:ConverseCEMarkov1}; \eqref{eq:FourthConverse11ce}  comes from the \mrb{identification} of the auxiliary random variable $ W_{1,i} =  ( M,S^{i-1}  )$, $\forall i \in \{1,\ldots,n\}$; \eqref{eq:FourthConverse12ce}  comes from the introduction of $T$, $S_T$,  $W_{1,T}$, $Y_T$;  \eqref{eq:FourthConverse14ce}   comes from \mrb{the identification} $W_1=(W_{1,T} , T)$ and $W_2 = W_{2,T}$; \eqref{eq:FourthConverse15ce}  comes from the empirical coordination requirement as stated in Lemma \ref{lemma:ErrorEventCoordination0}; \eqref{eq:FourthConverse16ce} comes from Lemma \ref{lemma:3} \mrb{since} $(S^n,X^n,Y^n,V^n)$ are jointly typical, hence $ \PP_{S_TX_TY_TV_T|E=0}$ is close to $\QQ_{SXYV}$. The \mrb{result of} \cite[Lemma 2.7, pp. 19]{CsiszarKorner(Book)11} concludes.

 \textbf{Conclusion:} 
If the triple of rate, state leakage and distribution$(\textsf{R},\textsf{E},\mrb{\QQ_{SXYV}})$ is achievable with causal encoding, then  the following equations are satisfied  for all $\varepsilon>0$:
\begin{align}
 I(S;W_1,W_2,Y) - 4 \varepsilon \leq \textsf{E}&\leq H(S) + \varepsilon,\\
\textsf{R} + \textsf{E} &\leq I(S ,W_{1} ; Y)+ 5 \varepsilon.
\end{align}
\mrb{We recover \eqref{eq:theoremCEachie2} and  \eqref{eq:theoremCEachie3}  as $\varepsilon\to0$  and this} concludes the converse proof of Theorem \ref{theo:LeakageCE}.

\begin{remark}
For the converse proof of Theorem \ref{theo:LeakageCE}, the \mrb{code with} causal encoding is not necessarily deterministic. The same optimal performances can be obtained by considering a stochastic \mrb{code with} causal encoding.
\end{remark}



%

\subsection{Lemmas for the converse proof with empirical coordination}\label{lemma:CE}

\begin{lemma}\label{lemma1:ConverseCEMarkov1}
The causal encoding function and the memoryless property of the channel induce the Markov chain property
$ X_i -\!\!\!\!\minuso\!\!\!\!- (S_i ,  W_{1,i}) -\!\!\!\!\minuso\!\!\!\!-    W_{2,i}$.
This Markov chain is satisfied with $ W_{1,i} =( M , S^{i-1} )$ and $ W_{2,i} =  Y_{i+1}^n $, for all $i \in \{1,\ldots,n\}$.
\end{lemma}

\begin{proof}[Lemma \ref{lemma1:ConverseCEMarkov1}] 
 \mrb{We identify $w_{1,i} =  ( m , s^{i-1} )$, $ w_{2,i} = y_{i+1}^n $ and} for all $(s^n,x^n,w_1^n,w_2^n,y^n,m)$
\begin{align}
&\PP( w_{2,i}   | s_i , w_{1,i} ,x_i) =\PP( y_{i+1}^n   | s_i , m , s^{i-1} ,x_i) \\
=&\sum_{s_{i+1}^n,x_{i+1}^n}  \PP( s_{i+1}^n,x_{i+1}^n  | s_i , m , s^{i-1} ,x_i) \nonumber\displaybreak[0]\\
&\times\PP( y_{i+1}^n  |s_{i+1}^n,x_{i+1}^n ,s_i , m , s^{i-1} ,x_i)   \label{eq:lemma1ConverseCE2}  \\
=&\sum_{s_{i+1}^n,x_{i+1}^n}  \PP( s_{i+1}^n,x_{i+1}^n  | s_i , m , s^{i-1} )  \PP( y_{i+1}^n   | s_{i+1}^n,x_{i+1}^n)   \label{eq:lemma1ConverseCE4}  \\
=&\sum_{s_{i+1}^n,x_{i+1}^n} \PP( s_{i+1}^n,x_{i+1}^n ,y_{i+1}^n   | s_i , m , s^{i-1} )\displaybreak[0]\\
&   = \PP( y_{i+1}^n   | s_i , m , s^{i-1} ) = \PP( w_{2,i}   | s_i , w_{1,i}) ,  \label{eq:lemma1ConverseCE5}
\end{align}
where \eqref{eq:lemma1ConverseCE4} comes from the causal encoding function that induces the Markov chain $X_i -\!\!\!\!\minuso\!\!\!\!- (S_i , M , S^{i-1}) -\!\!\!\!\minuso\!\!\!\!-  (S_{i+1}^n,X_{i+1}^n)$
\mrb{and} the memoryless channel $ Y_{i+1}^n -\!\!\!\!\minuso\!\!\!\!- (S_{i+1}^n,X_{i+1}^n ) -\!\!\!\!\minuso\!\!\!\!-  ( S_i , M , S^{i-1} ,X_i)$. 
\end{proof}

%

\begin{lemma}\label{lemma2:ConverseCEMarkov1}
The causal encoding function, the non-causal decoding function and the memoryless channel induce 
$ V_i  -\!\!\!\!\minuso\!\!\!\!- ( Y_i, W_{1,i} , W_{2,i}   )   -\!\!\!\!\minuso\!\!\!\!-  ( S_i, X_i)$.
This Markov chain is satisfied with $ W_{1,i} =( M , S^{i-1} )$ and $ W_{2,i} =  Y_{i+1}^n $, for all $i \in \{1,\ldots,n\}$.
\end{lemma}

\begin{proof}[Lemma \ref{lemma2:ConverseCEMarkov1}] 
 \mrb{We identify $w_{1,i} =  ( m , s^{i-1} )$, $ w_{2,i} = y_{i+1}^n $ and for all $(s^n,x^n,w_1^n,w_2^n,y^n,v^n,m)$}
\begin{align}
&\PP(v_i | y_i , w_{1,i} , w_{2,i}  ,  s_i, x_i) 
=\PP(v_i | y_i , m, s^{i-1}     ,  y^n_{i+1}  ,  s_i, x_i) \nonumber \\
=&   \sum_{x^{i-1}, y^{i-1}}  \PP(x^{i-1}  | y_i , m,  s^{i-1}  ,  y^n_{i+1}  ,  s_i, x_i )\nonumber\\
&\times  \PP(  y^{i-1}  |y_i , m, s^{i-1}  ,  y^n_{i+1}  ,  s_i, x_i ,  x^{i-1} ) \nonumber\\
&\times \PP(v_i  |  y_i , m, s^{i-1}  ,  y^n_{i+1}  ,  s_i, x_i  , x^{i-1} , y^{i-1}  ) .
 \end{align}
\mrb{We} can remove $( s_i , x_i)$, in the three conditional distributions
 \begin{align}
&\PP(x^{i-1}  | y_i ,m,  s^{i-1}  ,  y^n_{i+1}  ,  s_i, x_i )  = \PP(x^{i-1}  |m, s^{i-1} ) , \label{eq:lemma:CausalEnc2}\\
&\PP(  y^{i-1}  |y_i , m, s^{i-1}  ,  y^n_{i+1}  ,  s_i, x_i  , x^{i-1} )   \nonumber\\
&=  \PP(  y^{i-1}   | s^{i-1}  , x^{i-1}) ,  \label{eq:lemma:CausalEnc3}\\
&\PP(v_i  | y_i , m, s^{i-1}  ,  y^n_{i+1}  ,  s_i, x_i  , x^{i-1}  , y^{i-1}  )\nonumber\\
& = \PP(v_i  | y_i  ,  y^n_{i+1}   , y^{i-1}  ),  \label{eq:lemma:CausalEnc4}
\end{align}
\normalsize
where \eqref{eq:lemma:CausalEnc2} comes from the causal encoding that induces  $X^{i-1} -\!\!\!\!\minuso\!\!\!\!- (M,  S^{i-1})    -\!\!\!\!\minuso\!\!\!\!- (Y_i ,  Y^n_{i+1}  , X_i,   S_i )$; \eqref{eq:lemma:CausalEnc3} comes from the memoryless channel; \eqref{eq:lemma:CausalEnc4} comes from the non-causal decoding  that induces $V_i -\!\!\!\!\minuso\!\!\!\!- (Y_i  ,  Y^n_{i+1}   , Y^{i-1} )    -\!\!\!\!\minuso\!\!\!\!- (M, S^{i-1}, S_i, X_i, X^{i-1} )$. Hence we have for all $(s^n,x^n,w_1^n,w_2^n,y^n,v^n,m)$:
\begin{align}
& \PP(v_i | y_i , w_{1,i} , w_{2,i}  ,  s_i, x_i)\nonumber\\
=&   \sum_{x^{i-1}, y^{i-1} } \PP( v_i , x^{i-1}  , y^{i-1}   | y_i , m, s^{i-1}  ,  y^n_{i+1}  )\\
=& \PP(v_i | y_i  ,m  , s^{i-1}  ,  y^n_{i+1}   )  =  \PP(v_i | y_i  ,w_{1,i} , w_{2,i}  ).
 \end{align}
\mrb{This} concludes the proof of Lemma \ref{lemma2:ConverseCEMarkov1}.
\end{proof}

%

\begin{lemma}\label{lemma3:ConverseCEMarkov1}
The causal encoding function and the memoryless property of the channel  induce 
$ Y_i  -\!\!\!\!\minuso\!\!\!\!-  ( S_i ,M ,  S^{i-1}) -\!\!\!\!\minuso\!\!\!\!- ( S_{i+1}^n , Y_{i+1}^n  )$, $\forall i \in \{1,\ldots,n\}$.
\end{lemma} 
 
\begin{proof}[Lemma \ref{lemma3:ConverseCEMarkov1}] 
We have the following equations for all $(s^n,x^n,y^n,m)$:
\begin{align}
&\PP( y_i  |   s_i ,m ,  s^{i-1} ,  s_{i+1}^n , y_{i+1}^n  ) \displaybreak[0]\\
=& \sum_{x_i}  \PP(x_i ,  y_i  |   s_i ,m ,  s^{i-1} ,  s_{i+1}^n , y_{i+1}^n ) \\
=&\sum_{x_i}  \PP(x_i   |   s_i ,m ,  s^{i-1} ,  s_{i+1}^n , y_{i+1}^n   )  \nonumber \\
&\times \PP(  y_i  | x_i ,  s_i ,m ,  s^{i-1} ,  s_{i+1}^n , y_{i+1}^n   ) \label{eq:lemma3ConverseCE1}  \\
=&\sum_{x_i}  \PP(x_i   |   s_i ,m ,  s^{i-1}  )  
 \PP(  y_i  | x_i ,  s_i )  \displaybreak[0]\\
 =& \sum_{x_i}  \PP(x_i ,  y_i  |   s_i ,m ,  s^{i-1}  ) =   \PP(  y_i  |  s_i ,m ,  s^{i-1}  ),\label{eq:lemma3ConverseCE3}  
\end{align}
where \eqref{eq:lemma3ConverseCE3} comes from $X_i -\!\!\!\!\minuso\!\!\!\!- ( S_i ,M ,  S^{i-1}) -\!\!\!\!\minuso\!\!\!\!-  ( S_{i+1}^n , Y_{i+1}^n )$ \mrb{and} 
$Y_i -\!\!\!\!\minuso\!\!\!\!-  (X_i ,  S_i) -\!\!\!\!\minuso\!\!\!\!- (M ,  S^{i-1} ,  S_{i+1}^n , Y_{i+1}^n )$. This concludes the proof of Lemma \ref{lemma3:ConverseCEMarkov1}.
\end{proof}

%

The random \mrb{error} event $E \in \{0,1\}$ of the following Lemmas is defined in  \eqref{eq:ErrorEvent}, \mrb{in which} $E=1$ if $(S^n,X^n,Y^n,V^n)\notin T_{\delta}(\QQ)$.

%
\begin{lemma}\label{lemma:ErrorEventCoordination01}
Fix a distribution $\QQ_{S_TW_1W_2X_TY_TV_T} \in \Q_{\sf{c}}$ and suppose that the error probability $\prob(E=1)$ is small enough such that $\prob(E=1) \cdot \log_2 |\mc{S} |   + h_b\Big(\prob(E=1)\Big)\leq \varepsilon$, \mrb{where $h_b$ denotes the binary entropy function. Then} we have
\begin{align}
 I(S_T ; W_1,W_2,Y_T) \geq& I(S_T ; W_1,W_2,Y_T | E = 0)   -  \varepsilon.
\end{align}
\end{lemma}

\begin{proof}[Lemma \ref{lemma:ErrorEventCoordination01}]
\begin{align}
& I(S_T ; W_1,W_2,Y_T) =   I(S_T ; W_1,W_2,Y_T | E)  + I(S_T ; E ) \nonumber\\
&-  I(S_T ; E | W_1,W_2,Y_T ) \\
\geq&  \big(1-\prob(E=1)\big)\cdot I(S_T ; W_1,W_2,Y_T | E = 0)   \nonumber\\
&+  \prob(E=1) \cdot I(S_T ; W_1,W_2,Y_T | E = 1)  -  H(E)  \\
\geq& I(S_T ; W_1,W_2,Y_T | E = 0)   \nonumber\\
& -  \prob(E=1)  \cdot \log_2 |\mc{S} |   - h_b\Big(\prob(E=1)\Big) .
\end{align}
\end{proof}
%

\begin{lemma}\label{lemma:ErrorEventCoordination0}
Fix a distribution $\QQ_{S_TW_1W_2X_TY_TV_T} \in \Q_{\sf{c}}$ and suppose that the error probability $\prob(E=1)$ is small enough such that $\prob(E=1) \cdot \log_2 |\mc{Y} | +  h_b\Big(\prob(E=1)\Big)\leq \varepsilon$. Then we have
\begin{align}
I(S_T ,W_{1} ; Y_T) \leq & I(S_T ,W_{1} ; Y_T |E = 0)   +  \varepsilon.\label{eq:ErrorEventCoordination0Third}
\end{align}
\end{lemma}
\begin{proof}[Lemma \ref{lemma:ErrorEventCoordination0}]
\begin{align}
&I(S_T ,W_{1} ; Y_T) \nonumber\\
=& I(S_T ,W_{1} ; Y_T|E) +I(E ; Y_T   ) -I(E  ; Y_T| S_T,W_{1}) \displaybreak[0]\\
\leq&  \prob(E=0) \cdot I(S_T ,W_{1} ; Y_T|E=0)\nonumber\\
&+  \prob(E=1) \cdot I(S_T ,W_{1} ; Y_T|E=1) +H(E) \displaybreak[0]\\ 
\leq& I(S_T ,W_{1} ; Y_T|E=0)+  \prob(E=1) \cdot  \log_2 |\mc{Y} |   +H(E)\leq\varepsilon.
\end{align}
\end{proof}

We denote by $ \PP_{S_TX_TY_TV_T|E=0}$ the distribution induced by the random variables $(S_T,X_T,Y_T,V_T)$ knowing the event $E=0$ is realized.
\begin{lemma}\label{lemma:3}
\mrb{The} distribution defined by $ \PP_{S_TX_TY_TV_T|E=0}$ is close to the target distribution $\QQ_{SXYV}$, \mrb{i.e.} for all $(s,x,y,v)$
\begin{align}
 \bigg|  \PP_{S_TX_TY_TV_T|E=0}(s,x,y,v) -  \QQ(s,x,y,v)  \bigg|  \leq \varepsilon.
\end{align} 
\end{lemma}

\begin{proof}[Lemma \ref{lemma:3}]
We evaluate the probability \mrb{value} $\PP_{S_T|E=0}(s)$ and we show \mrb{that} it is close to the  probability \mrb{value} $\PP(s)$, for all $s\in \mc{S}$.
\begin{align}
&\PP_{S_T|E=0}(s)\nonumber\\
=&\sum_{s^n \in T_{\delta}(\PP) } \sum_{i = 1}^n  \prob\big(S^n = s^n , T = i , S_T = s  \big| E=0\big)  \label{eq:Lemma3_Eq1} \\
=&\sum_{s^n \in T_{\delta}(\PP) } \sum_{i = 1}^n   \prob\big(S^n = s^n   \big| E=0\big)   \prob\big( T = i   \big| S^n = s^n , E=0\big) \nonumber\\
&\times \prob\big(S_T = s  \big|S^n = s^n , T = i ,  E=0\big)  \label{eq:Lemma3_Eq2} \\
=&\sum_{s^n \in T_{\delta}(\PP) } \sum_{i = 1}^n   \prob\big(S^n = s^n   \big| E=0\big)  \prob\big( T = i \big) \UN(s_T = s)\label{eq:Lemma3_Eq4} \\
=&\sum_{s^n \in T_{\delta}(\PP) }   \prob\big(S^n = s^n   \big| E=0\big) \sum_{i = 1}^n  \frac{1}{n} \UN(s_T = s)\label{eq:Lemma3_Eq5} \\
=&\sum_{s^n \in T_{\delta}(\PP) }   \prob\big(S^n = s^n   \big| E=0\big)   \frac{N(s|s^n)}{n} ,\label{eq:Lemma3_Eq6}
\end{align}
where \eqref{eq:Lemma3_Eq4} comes from the independence of event $\{T = i\} $ with events $\{S^n = s^n\}$ and $\{E=0\}$; \eqref{eq:Lemma3_Eq6} comes from the definition of the number of occurrence $N(s|s^n) = \sum_{i = 1}^n   \UN(s_T = s )$.

Since the sequences $s^n \in T_{\delta}(\PP)$ are typical, we have for all $\mrb{ \tilde{s}}\in \mc{S}$:
\begin{align}
\PP(\mrb{ \tilde{s}}) - \varepsilon  \leq  \frac{N(\mrb{ \tilde{s}}|s^n)}{n}  \leq  \PP(\mrb{ \tilde{s}}) + \varepsilon,
\end{align}
which provides an upper bound and a lower bound on  $\prob(S_T = \mrb{ \tilde{s}} | E=0)$ as 
\begin{align}
 &\PP(\mrb{ \tilde{s}}) - \varepsilon 
 = \sum_{s^n \in T_{\delta}(\PP) }   \prob\big(S^n = s^n   \big| E=0\big) \Big( \PP(\mrb{ \tilde{s}}) - \varepsilon\Big) \nonumber\\
 \leq&   \prob(S_T = \mrb{ \tilde{s}} | E=0)\\
\leq &  \sum_{s^n \in T_{\delta} (\PP)}   \prob\big(S^n = s^n   \big| E=0\big) \Big( \PP(\mrb{ \tilde{s}}) + \varepsilon\Big) =   \PP(\mrb{ \tilde{s}}) + \varepsilon.
\end{align}
Using \mrb{similar} arguments, we \mrb{show} that the distribution $\PP_{S_TX_TY_TV_T|E=0}(s,x,y,v)$  is close to the target distribution $\QQ(s,x,y,v)$, \mrb{i.e.} for all $(s,x,y,v)$
\begin{align}
 \bigg|  \PP_{S_TX_TY_TV_T|E=0}(s,x,y,v) -  \QQ(s,x,y,v)\bigg|  \leq \varepsilon.
\end{align} 
This concludes the proof of Lemma \ref{lemma:3}.
\end{proof}


\section{Proof of Theorem \ref{theo:Estimation}}\label{sec:ProofTheoremKL}
The conditional distribution $\PP_{W_1^nW_2^nX^n|S^n}$ combined with $\PP_S $ and $\mc{T}_{Y|XS}$ \mrb{induces the} distribution
\begin{align}
 \PP_{S^nW_1^nW_2^nX^nY^n}&= \prod_{i=1}^n \PP_{S_i}    \PP_{W_1^nW_2^nX^n|S^n} \prod_{i=1}^n\mc{T}_{Y_i|X_iS_i}.
\end{align} 
We introduce the random event $F \in \{0,1\}$ depending on whether the random sequences $(S^n,W_1^n,W_2^n,Y^n)$ are jointly typical or not.
\begin{align}
F = \Bigg\{
\begin{array}{lll}
0 &\text{ if }&  (S^n,W_1^n, W_2^n,Y^n) \in T_{\delta}(\QQ) ,\\
1 &\text{ if }& (S^n,W_1^n, W_2^n,Y^n) \notin T_{\delta}(\QQ) .
\end{array}
\Bigg.
\end{align}
Since the target distribution $\QQ_{SW_1W_2XY}$ has full support, the distribution $\PP_{S^n|Y^n}$ is absolutely continuous with respect to $\prod_{i=1}^n \QQ_{S_i|Y_iW_{1,i}W_{2,i}}$, and the conditional \ac{KL}-divergence \mrb{writes}
\begin{align}
& \frac{1}{n} D\bigg(\PP_{S^n|Y^n} \bigg|\bigg| \prod_{i=1}^n \QQ_{S_i|Y_iW_{1,i}W_{2,i}} \bigg)  \nonumber \displaybreak[0]\\
=& \frac{1}{n} \sum_{w_1^n,w_2^n,y^n}\PP(w_1^n,w_2^n,y^n)\nonumber\displaybreak[0]\\
&\times  \sum_{s^n}\PP(s^n|y^n) \log_2 \frac{1}{\prod_{i=1}^n \QQ(s_i|y_i,w_{1,i}, w_{2,i} )} \nonumber \displaybreak[0]\\
-& \frac{1}{n} \sum_{w_1^n,w_2^n,y^n}\PP(w_1^n,w_2^n,y^n) \sum_{s^n}\PP(s^n|y^n) \log_2 \frac{1}{\PP(s^n|y^n)}  \label{eq:LeakageKL2}\displaybreak[0]\\
=&\frac{1}{n} \sum_{w_1^n,w_2^n,\atop y^n,s^n,F}\PP(w_1^n,w_2^n,y^n) \PP(s^n|y^n)\prob(F|s^n,w_1^n,w_2^n,y^n)\nonumber\displaybreak[0]\\
&\times \log_2 \frac{1}{\prod_{i=1}^n \QQ(s_i|y_i,w_{1,i}, w_{2,i} )} -  \frac{1}{n} H(S^n|Y^n) \label{eq:LeakageKL3} \displaybreak[0]\\
=& \prob(F=0) \frac{1}{n} \sum_{(w_1^n,w_2^n,y^n,s^n) \atop\in T_{\delta}(\QQ)}\prob(s^n,w_1^n,w_2^n,y^n|F=0)  \nonumber\displaybreak[0]\\
&\times\log_2 \frac{1}{\prod_{i=1}^n \QQ(s_i|y_i,w_{1,i}, w_{2,i} )} \nonumber \displaybreak[0]\\
+& \prob(F=1) \frac{1}{n} \sum_{w_1^n,w_2^n,y^n,s^n}\prob(s^n,w_1^n,w_2^n,y^n|F=1)  \nonumber\displaybreak[0]\\
&\times\log_2 \frac{1}{\prod_{i=1}^n \QQ(s_i|y_i,w_{1,i}, w_{2,i} )}  - \frac{1}{n} H(S^n|Y^n) \label{eq:LeakageKL4} \displaybreak[0]\\
\leq& \frac{1}{n} \sum_{(w_1^n,w_2^n,y^n,s^n) \atop\in T_{\delta}(\QQ)}\prob(s^n,w_1^n,w_2^n,y^n|F=0)  \nonumber\displaybreak[0]\\
&\times n\bigg(H(S|W_1,W_2,Y) + \delta  \sum_{s,w_1,\atop w_2,y} \log_2  \frac{1}{\QQ(s|w_1,w_2,y)} \bigg)  \nonumber\displaybreak[0]\\
+& \prob(F=1)  \log_2 \frac{1}{\min_{s,y,w_1,w_2} \QQ(s|y,w_{1}, w_{2} )}  - \frac{1}{n} H(S^n|Y^n) \label{eq:LeakageKL5} \displaybreak[0]\\
\leq& H(S|W_1,W_2,Y)  - \frac{1}{n} H(S^n|Y^n) \nonumber\displaybreak[0]\\
& + \delta  \sum_{s,w_1,\atop w_2,y} \log_2  \frac{1}{\QQ(s|w_1,w_2,y)} \nonumber\displaybreak[0]\\
&+  \prob(F=1)  \log_2 \frac{1}{\min_{s,y,w_1,w_2} \QQ(s|y,w_{1}, w_{2} )} \label{eq:LeakageKL6} \displaybreak[0]\\
=&\mrb{ \frac{1}{n} I(S^n;Y^n) - I(S;W_1,W_2,Y)  }\nonumber\displaybreak[0]\\
& + \delta  \sum_{s,w_1,\atop w_2,y} \log_2  \frac{1}{\QQ(s|w_1,w_2,y)} \nonumber\displaybreak[0]\\
&+  \prob(F=1)  \log_2 \frac{1}{\min_{s,y,w_1,w_2} \QQ(s|y,w_{1}, w_{2} )}\label{eq:LeakageKL7}\displaybreak[0]\\
\leq& \mrb{ \mc{L}_{\textsf{e}}(c) - I(S;W_1,W_2,Y)  } + \delta  \sum_{s,w_1,\atop w_2,y} \log_2  \frac{1}{\QQ(s|w_1,w_2,y)} \nonumber \displaybreak[0]\\
&+   \prob\Big((s^n,w_1^n,w_2^n,y^n) \notin T_{\delta}(\QQ) \Big) \nonumber\displaybreak[0]\\
&\times \log_2 \frac{1}{\min_{s,y,w_1,w_2} \QQ(s|y,w_{1}, w_{2} )}\label{eq:LeakageKL8},
\end{align}
where \eqref{eq:LeakageKL2} comes from the definition of the \ac{KL}-divergence; \eqref{eq:LeakageKL3} comes from the introduction of the random event $F \in \{0,1\}$; \eqref{eq:LeakageKL4} is a reformulation of \eqref{eq:LeakageKL3} where $ \prob(F)  \cdot\prob(s^n,w_1^n,w_2^n,y^n|F) = \PP(w_1^n,w_2^n,y^n) \PP(s^n|y^n)\prob(F|s^n,w_1^n,w_2^n,y^n) $ and where the event $F=0$ implies that  the sequences $(w_1^n,w_2^n,y^n,s^n)\in T_{\delta}(\QQ)$; \eqref{eq:LeakageKL5} comes from two properties: 1) the property $2^{-n\cdot\Big(H(S|Y,W_1,W_2) + \delta \sum_{s,w_1,\atop w_2,y} \log_2 \frac{1}{\QQ(s|w_1,w_2,y)} \Big)} \leq \prod_{i=1}^n \QQ_{S_i|Y_iW_{1,i}W_{2,i}}$ for typical sequences $(w_1^n,w_2^n,y^n,s^n)\in T_{\delta}(\QQ)$, stated in \cite[pp. 26]{ElGamalKim(book)11}; 2) the hypothesis $\QQ_{SW_1W_2XY}$ has full support, which implies that $\min_{s,y,w_1,w_2} \QQ(s|y,w_{1}, w_{2} )>0$; \eqref{eq:LeakageKL7} comes from the i.i.d. property of $S$ that implies $H(S) = \frac{1}{n} H(S^n)$; \eqref{eq:LeakageKL8} comes from the definition of the state leakage $\mc{L}_{\textsf{e}}(c)  = \frac{1}{n} I(S^n;Y^n)$ and the event $F=0$.




\section{Achievability proof of Theorem \ref{theo:DecoderEstimation}}\label{sec:AchievabilityTheoDecoderEstimation}

 We consider a pair of rate and distortion $(\textsf{R},\textsf{D}) \in \mc{A}_{\sf{g}}$ that \mrb{satisfies} \eqref{eq:DecoderEstimation1} and \eqref{eq:DecoderEstimation2} with distribution $\PP_S    \QQ_{W_1}   \QQ_{X|SW_1}      \mc{T}_{Y|XS}$ \mrb{and we} introduce the target leakage $\textsf{E} = I(S;W_1,Y)$. Theorem \ref{theo:LeakageCEd} guarantees that the triple $(\textsf{R},\textsf{E},\mrb{\QQ_{SXYV}} )$ is achievable\mrb{, i.e.} for all $\varepsilon>0$, there exists  $\bar{n}\in \N^{\star}$, for all $n \geq \bar{n}$, there exists a code with causal encoding $c\in\mc{C}(n,\mc{M})$ that satisfies
\begin{align}
&\frac{\log_2 |\mc{M}|}{n}  \geq \textsf{R} - \varepsilon,\label{eq:DecoderEstimationKL1}\displaybreak[0]\\
& \bigg| \mc{L}_{\textsf{e}}(c)  - \textsf{E} \; \bigg|  \leq\varepsilon, \;\; \text{\normalsize with } \;\; \mc{L}_{\textsf{e}}(c) = \frac{1}{n} I(S^n;Y^n)  , \label{eq:DecoderEstimationKL2}\displaybreak[0]\\
& \PP_{\textsf{e}}(c)  = \prob\bigg( M \neq \hat{M} \bigg)  
\nonumber\\
& +\prob\bigg(\Big|\Big|\mrb{Q^n_{SXYV} - \QQ_{SXYV}} \Big|\Big|_{1}> \varepsilon\bigg) \leq \varepsilon.  \label{eq:DecoderEstimationKL3}
 \end{align}  
\begin{remark}\label{remark:CorrectDecoding}
The achievability proof of Theorem \ref{theo:LeakageCEd} guarantees that the sequences $(S^n,W_1^n,X^n,Y^n)\in T_{\delta}(\QQ)$ with large probability and that the decoding of $W_1^n$ is correct with large probability, \mrb{i.e.} $ \prob(\hat{W}_1^n \neq W_1^n)\leq \varepsilon$.
\end{remark}

We assume that \mrb{the} distribution $\PP_S    \QQ_{W_1}   \QQ_{X|SW_1}    \mc{T}_{Y|XS}$ has full support. Otherwise, we would consider a sequence of distributions \mrb{with} full support, that converges to the target distribution. By replacing the pair $(W_1,W_2)$ by $W_1$ in Theorem \ref{theo:Estimation}, we obtain:
\begin{align}
& \frac{1}{n} D\bigg(\PP_{S^n|Y^n} \bigg|\bigg| \prod_{i=1}^n \QQ_{S_i|Y_iW_{1,i}} \bigg) \nonumber\\
 \leq&  \mrb{ \mc{L}_{\textsf{e}}(c) -I(S;W_1,Y)  }+ \delta \sum_{s,w_1,y} \log_2  \frac{1}{\QQ(s|w_1,y)}  \nonumber \\
 & +  \prob\Big((s^n,w_1^n,y^n) \notin T_{\delta}(\QQ) \Big)  \log_2 \frac{1}{\min_{s,y,w_1} \QQ(s|y,w_{1})}
 .\label{eq:DecoderEstimationKL0}
\end{align}
\mrb{The equations \eqref{eq:DecoderEstimationKL1}--\eqref{eq:DecoderEstimationKL3} and \eqref{eq:DecoderEstimationKL0} show} that for any $\varepsilon>0$, there exists $\bar{n}\in \N^{\star}$, for all $n \geq \bar{n}$, there exists a code with causal encoding $c\in\mc{C}(n,\mc{M})$ involving an auxiliary sequence $W_1^n$, such that
\begin{align}
\frac{\log_2 |\mc{M}|}{n}  \geq& \textsf{R} - \varepsilon,\label{eq:DecoderEstimationKL4}\\
 \prob\bigg( M \neq \hat{M} \bigg) \leq&\varepsilon,\label{eq:DecoderEstimationKL5}\\
 \frac{1}{n} D\bigg(\PP_{S^n|Y^n} \bigg|\bigg| \prod_{i=1}^n \QQ_{S_i|Y_iW_{1,i}} \bigg) 
 \leq& \varepsilon.\label{eq:DecoderEstimationKL6}
 \end{align}

Following \cite[Notations A.6 and A.7]{LeTreustTomala19}, we define the sets $J_{\alpha}(w_1^n,y^n)$ and $B_{\alpha,\gamma,\delta}$ depending on small parameters $\alpha>0$, $\gamma>0$ and $\delta>0$:
\begin{align}
&J_{\alpha}(w_1^n,y^n) = \bigg\{i \in \{1,\ldots,n\} , \;\;\text{ s.t. }\;\;\nonumber\displaybreak[0]\\
& \;\;D\Big(\PP_{S_i|Y^n}(\cdot|y^n)\Big|\Big|\QQ_{S_i|Y_iW_{1,i}}(\cdot|y_i,w_{1,i})\Big)\leq  \frac{\alpha^2}{2\ln 2} \bigg\},\\
&B_{\alpha,\gamma,\delta}= \bigg\{ (w_1^n,y^n) \;\;\text{ s.t. } \;\; \frac{|J_{\alpha}(w_1^n,y^n)|}{n} \geq 1 - \gamma\;\;\text{ and }\nonumber\displaybreak[0]\\
& \;\;(w_1^n,y^n) \in  T_{\delta}(\QQ)\bigg\}.\label{eq:SetTalpha}
\end{align}
The notation $B_{\alpha,\gamma,\delta}^c$ stands for the complementary set of $B_{\alpha,\gamma,\delta}\subset \mc{W}_1^n \times \mc{Y}^n $. The sequences $(w_1^n,y^n)$ belong to the set $B_{\alpha,\gamma,\delta}$ if they are jointly typical and if the decoder's posterior belief $\PP_{S_i|Y^n}(\cdot|y^n)$ is close in K-L divergence to the target belief $\QQ_{S_i|Y_iW_{1,i}}(\cdot|y_i,w_{1,i})$, for a large fraction of stages $i \in \{1,\ldots,n\}$.

\mrb{ 
\begin{lemma}\label{lemma:BoundSetB}
\begin{align}
&\prob(B_{\alpha,\gamma,\delta}^c) \leq  \frac{2 \ln 2}{\alpha^2 \gamma}  \frac{1}{n} D\Big(\PP_{S^n|Y^n}\Big|\Big|\prod_{i=1}^n\QQ_{S_i|Y_iW_{1,i}}\Big) \nonumber\\
&+  \prob((w_1^n,y^n) \notin  T_{\delta}(\QQ)).\label{eq:LemmaBoundSetB}
\end{align}
\end{lemma}
Lemma \ref{lemma:BoundSetB} is a reformulation of  \cite[equations (40) - (46)]{LeTreustTomala19}, a detailed proof is stated in Appendix \ref{sec:ProofLemmaBoundSetB}.
Then \cite[Lemma A.8]{LeTreustTomala19}
ensures that for each code with causal encoding $c\in\mc{C}(n,\mc{M})$ we have
\begin{align}
& \bigg| \min_{h_{V^n|Y^n}}  \frac{1}{n} \sum_{i=1}^n\E\Big[d(S_i,V_i)\Big]  - \min_{\PP_{V|W_1Y}} \E \Big[ d(S,V)\Big] \bigg| \nonumber \\
= &\bigg| \frac{1}{n} \sum_{i=1}^n \sum_{y^n} \PP(y^n)\min_{h_{V_i|Y^n}} \sum_{s_i} \PP(s_i|y^n)  d(s_i,v_i)
\nonumber\displaybreak[0]\\
& -  \sum_{w_1,y}\QQ(w_1,y) \min_{\PP_{V|W_1Y}}\sum_{s} \QQ(s|w_{1},y)  d(s,v)\bigg|  \\
\leq &\big(\alpha + 2 \gamma + \delta + \prob(B_{\alpha,\gamma,\delta}^c)\big) \bar{d},  \label{eq:BoundMinDistortion}
\end{align}
where $\bar{d} = \max_{s,v}d(s,v)$ is the maximal distortion value. 
We combine \eqref{eq:BoundMinDistortion} with \eqref{eq:LemmaBoundSetB} and \eqref{eq:DecoderEstimationKL6}, and then we choose the parameters $\alpha>0$, $\gamma>0$, $\delta>0$ small and $n$ large such as to obtain
\begin{align}
 \bigg| \min_{h_{V^n|Y^n}}  \frac{1}{n} \sum_{i=1}^n\E\Big[d(S_i,V_i)\Big]  - \min_{\PP_{V|W_1Y}} \E \Big[ d(S,V)\Big] \bigg|  \leq  \varepsilon.
\end{align}}
This concludes the achievability proof of Theorem \ref{theo:DecoderEstimation}.


\subsection{Proof of Lemma \ref{lemma:BoundSetB}}\label{sec:ProofLemmaBoundSetB}

The union bound implies
\begin{align}
\prob(B_{\alpha,\gamma,\delta}^c) 
=& \prob\bigg((W_1^n,Y^n)  \;\;\text{ s.t. } \;\; \frac{|J_{\alpha}(W_1^n,Y^n)|}{n} < 1 - \gamma\;\; \nonumber\\
&\text{ or }\;\;(W_1^n,Y^n) \notin  T_{\delta}(\QQ)\bigg)\\
\leq& \prob\bigg((W_1^n,Y^n)  \;\;\text{ s.t. } \;\; \frac{|J_{\alpha}(W_1^n,Y^n)|}{n} < 1 - \gamma \bigg) \nonumber \\
+& \prob\Big((W_1^n,Y^n) \notin  T_{\delta}(\QQ)\Big). \label{eq:ProofLemmaB1} 
\end{align}
Moreover, 
\begin{align}
&\prob\bigg((W_1^n,Y^n)  \;\;\text{ s.t. } \;\; \frac{|J_{\alpha}(W_1^n,Y^n)|}{n} < 1 - \gamma \bigg)\displaybreak[0]\nonumber\\
 =& \prob\bigg( \frac{1}{n} \bigg| \bigg\{i \text{ s.t. }
D\Big(\PP_{S_i|Y^n}(\cdot|Y^n)\Big|\Big|\QQ_{S_i|Y_iW_{1,i}}(\cdot|Y_i,W_{1,i})\Big)\nonumber\displaybreak[0]\\
&\leq  \frac{\alpha^2}{2\ln 2} \bigg\}\bigg|  < 1 -   \gamma   \bigg) \label{eq:ProofLemmaB2}\displaybreak[0]\\
=& \prob\bigg( \frac{1}{n} \bigg| \bigg\{i \text{ s.t. }D\Big(\PP_{S_i|Y^n}(\cdot|Y^n)\Big|\Big|\QQ_{S_i|Y_iW_{1,i}}(\cdot|Y_i,W_{1,i})\Big)\nonumber\displaybreak[0]\\
&>  \frac{\alpha^2}{2\ln 2} \bigg\}\bigg|  \geq  \gamma \bigg) \label{eq:ProofLemmaB4}\displaybreak[0]\\
\leq &\frac{2 \ln 2}{\alpha^2 \gamma} \E\bigg[ \frac{1}{n} \sum_{i=1}^n D\Big(\PP_{S_i|Y^n}(\cdot|Y^n)\Big|\Big|\QQ_{S_i|Y_iW_{1,i}}(\cdot|Y_i,W_{1,i})\Big)
\bigg]\label{eq:ProofLemmaB5} \displaybreak[0]\\
\leq &\frac{2 \ln 2}{\alpha^2 \gamma}  \frac{1}{n} D\Big(\PP_{S^n|Y^n}\Big|\Big|\prod_{i=1}^n\QQ_{S_i|Y_iW_{1,i}}\Big),\label{eq:ProofLemmaB6}
\end{align}
where  \eqref{eq:ProofLemmaB2}-\eqref{eq:ProofLemmaB4} are reformulations;~\eqref{eq:ProofLemmaB5} comes from \cite[Lemma A.21]{LeTreustTomala19};~\eqref{eq:ProofLemmaB6} comes from Lemma \ref{lemma:KLproductC}.


\begin{lemma}\label{lemma:KLproductC}
We consider the distributions $\PP_{A_1A_2B_1B_2}$, $\QQ_{B_1|A_1}$ and $\QQ_{B_2|A_2}$. We have
\begin{align}
&D\Big(\PP_{B_1B_2|A_1A_2} \Big|\Big| \QQ_{B_1|A_1}  \QQ_{B_2|A_2} \Big)\nonumber \\
=&D\Big(\PP_{B_1|A_1A_2} \Big|\Big| \QQ_{B_1|A_1}  \Big) + D\Big(\PP_{B_2|A_1A_2} \Big|\Big| \QQ_{B_2 |A_2} \Big) \nonumber\displaybreak[0]\\
&+ I(B_1;B_2|A_1,A_2),
\end{align}
where the mutual information $ I(B_1;B_2|A_1,A_2)$ is evaluated with respect to $\PP_{A_1A_2B_1B_2}$. In particular, this implies for all $n\geq1$
 \begin{align}
D\bigg(\PP_{B^n|A^n} \bigg|\bigg| \prod_{i=1}^n \QQ_{B_i|A_i} \bigg)  \geq \sum_{i=1}^n D\Big(\PP_{B_i|A^n} \Big|\Big| \QQ_{B_i|A _i}  \Big).
\end{align}
 \end{lemma}

 \begin{proof}[Lemma \ref{lemma:KLproductC}]
 \begin{align}
&D\Big(\PP_{B_1B_2|A_1A_2} \Big|\Big| \QQ_{B_1|A_1} \times  \QQ_{B_2|A_2} \Big) \nonumber \\
=&\sum_{a_1,a_2}\PP(a_1,a_2) \sum_{b_1} \PP(b_1| a_1,a_2)  \log_2 \frac{ 1}{\QQ(b_1|a_1 )}\nonumber\\
&+\sum_{a_1,a_2}\PP(a_1,a_2) \sum_{b_2} \PP(b_2| a_1,a_2)  \log_2 \frac{1}{   \QQ(b_2|a_2 )} \nonumber\\
&- H(B_1,B_2|A_1,A_2)\displaybreak[0]\\
=&\sum_{a_1,a_2}\PP(a_1,a_2) \sum_{b_1} \PP(b_1| a_1,a_2)  \log_2 \frac{ \PP(b_1| a_1,a_2)}{\QQ(b_1|a_1 )}\nonumber\\
&+\sum_{a_1,a_2}\PP(a_1,a_2) \sum_{b_2} \PP(b_2| a_1,a_2)  \log_2 \frac{\PP(b_2| a_1,a_2) }{   \QQ(b_2|a_2 )}\nonumber\\
&+H(B_1|A_1,A_2) + H(B_2|A_1,A_2) - H(B_1,B_2|A_1,A_2)\\
=&D\Big(\PP_{B_1|A_1A_2} \Big|\Big|  \QQ_{B_1|A_1} \Big) + D\Big(\PP_{B_2|A_1A_2} \Big|\Big|  \QQ_{B_2|A_2} \Big) \nonumber\displaybreak[0]\\
&+ I(B_1;B_2|A_1,A_2).
\end{align}
\end{proof}


\section{Converse proof of Theorem \ref{theo:DecoderEstimation}}\label{sec:ConverseTheoDecoderEstimation}

We introduce the random event $F \in \{0,1\}$ indicating whether $M$ is correctly decoded or not.
We assume that there exists a code with causal encoding $c\in\mc{C}(n,\mc{M})$ that satisfies
\begin{align}
\frac{\log_2 |\mc{M}|}{n}  \geq& \textsf{R} - \varepsilon,\label{eq:DefinitionEstimation1}\\
 \PP_{\textsf{e}}(c)  = \prob\bigg( M \neq \hat{M} \bigg) \leq& \varepsilon, \label{eq:DefinitionEstimation2} \\
\bigg|\min_{h_{V^n|Y^n}} \frac{1}{n} \sum_{i=1}^n\E\Big[d(S_i,V_i)\Big] - \textsf{D}\bigg| \leq &\varepsilon. \label{eq:DefinitionEstimation3}
 \end{align}
\mrb{We first show the constraint \eqref{eq:DecoderEstimation1}.}
\begin{align}
\textsf{R} \leq& \frac{\log_2 |\mc{M}|}{n}  + \varepsilon  =  \frac{1}{n} H(M)  + \varepsilon\nonumber\displaybreak[0]\\
 =& \frac{1}{n} I(M;Y^n) + \frac{1}{n} H(M|Y^n)  + \varepsilon  \label{eqConverseFirst1} \displaybreak[0]\\
\leq& \frac{1}{n} I(M;Y^n) + 2\varepsilon  = \frac{1}{n} \sum_{i=1}^n I(M;Y_i|Y^{i-1}) + 2\varepsilon \nonumber\displaybreak[0]\\
\leq&  \frac{1}{n} \sum_{i=1}^n I(M,Y^{i-1},Y^n_{i+1};Y_i) + 2\varepsilon\label{eqConverseFirst5} \displaybreak[0]\\
=& \frac{1}{n} \sum_{i=1}^n I(W_{1,i};Y_i) + 2\varepsilon =   I(W_{1,T};Y_T|T) + 2\varepsilon \nonumber\displaybreak[0]\\
 \leq&  I(W_{1,T},T;Y_T) + 2\varepsilon \leq I(W_{1};Y) + 2\varepsilon\label{eqConverseFirst9},
\end{align}
where \eqref{eqConverseFirst1} comes from assumption \eqref{eq:DefinitionEstimation1} and \mrb{the} uniform distribution of the message $M$; \eqref{eqConverseFirst5} comes from Fano's inequality \cite[Theorem 2.10.1]{cover-book-2006} and assumption \eqref{eq:DefinitionEstimation2} and adding the mutual informations $ I(Y^{i-1};Y_i)$ and $ I(Y^n_{i+1};Y_i|M,Y^{i-1})$; 
\eqref{eqConverseFirst9} comes from the identification of the auxiliary random variable $W_{1,i} = (M,Y^{i-1},Y^n_{i+1})$ and the introduction of the uniform random variable $T$ over the indices $\{1,\ldots,n\}$ and $(W_{1,T},Y_T)$ and by identifying $W_{1} = (W_{1,T},T)$ and $Y_T=Y$. 

\mrb{We now show the constraint \eqref{eq:DecoderEstimation2}.} The notation $y^{-i}=(y^{i-1},y_{i+1}^n)$ stands for the subsequence of $y^n$ where $y_i$ has been removed.  We assume that the event $F=0$ is realized. Then, the average distortion satisfies
\begin{align}
&\!\!\min_{h_{V^n|Y^n}}\!\! \sum_{s^n,y^n,\atop m,v^n} \prob(s^n,y^n,m|F=0)  h(v^n|y^n) \bigg[\frac{1}{n} \sum_{i=1}^n d(s_i,v_i)\bigg]\displaybreak[0]\\
=&   \sum_{y^n,m} \prob(y^n,m|F=0)  \!\! \min_{\tilde{h}:\mc{Y}^n\times   \mc{M} \to \mc{V}^n} \sum_{s^n} \prob(s^n|y^n,m,F=0)\nonumber\\
&\times \bigg[\frac{1}{n} \sum_{i=1}^n d(s_i,v_i)\bigg]  \displaybreak[0]\label{eq:ConverseSecond0}\\
=&  \sum_{y_i,y^{-i},m} \prob(y^n,m|F=0)   \frac{1}{n} \sum_{i=1}^n\nonumber\\
&\times \min_{\hat{h}:\mc{Y}^n\times \mc{M} \to \mc{V}_i} \sum_{s_i} \prob(s_i|y_i,y^{-i},m,F=0)  d(s_i,v_i) \label{eq:ConverseSecond2}\displaybreak[0]\\
=& \sum_{y_i,w_{1,i}} \PP(y_i,w_{1,i})   \frac{1}{n} \sum_{i=1}^n \nonumber\\
&\times\min_{h:\mc{Y}_i \times \mc{W}_{1,i} \to \mc{V}_i} \sum_{s_i} \PP(s_i|y_i,w_{1,i})  d(s_i,v_i) \label{eq:ConverseSecond3}\displaybreak[0]\\
=& \sum_{y_T,w_{1,T},T} \PP(y_T,w_{1,T},T)  \nonumber\\
&\times  \min_{h:\mc{Y} \times \mc{W}_{1} \to \mc{V}} \sum_{s_T} \PP(s_T|y_T,w_{1,T},T)  d(s_T,v_T) \label{eq:ConverseSecond4}\displaybreak[0]\\
=& \sum_{y,w_{1}} \PP(y,w_{1})    \min_{h:\mc{Y} \times \mc{W}_{1} \to \mc{V}} \sum_{s} \PP(s|y,w_{1})  d(s,v) \displaybreak[0]\\
&= \min_{\PP_{V|W_1Y}} \E \Big[ d(S,V)\Big] \label{eq:ConverseSecond6},
\end{align}
where \eqref{eq:ConverseSecond0}, comes from the hypothesis $F=0$, which guarantees the correct decoding of the message $m$ based on the observation of $y^n$; \eqref{eq:ConverseSecond2} is a reformulation; 
\eqref{eq:ConverseSecond3} comes from the identification of the auxiliary random variable $w_{1,i} = (m,y^{-i})$; 
\eqref{eq:ConverseSecond4} comes from the introduction of the uniform random variable $T$; 
\eqref{eq:ConverseSecond6} comes from \mrb{the identification of} the auxiliary random variables $w_{1} = (w_{1,T},T)$. \mrb{We have}
\begin{align}
&\bigg| \min_{\PP_{V|W_1Y}} \E \Big[ d(S,V)\Big] - \textsf{D}  \bigg|\\
\leq &\bigg| \min_{\PP_{V|W_1Y}} \E \Big[ d(S,V)\Big]  - \min_{h_{V^n|Y^n}} \frac{1}{n} \sum_{i=1}^n\E\Big[d(S_i,V_i)\Big]\bigg|\nonumber\\
& + \bigg| \min_{h_{V^n|Y^n}} \frac{1}{n} \sum_{i=1}^n\E\Big[d(S_i,V_i)\Big] -   \textsf{D}  \bigg| \label{eq:ConverseThird1}\displaybreak[0]\\
\leq &\bigg| \min_{\PP_{V|W_1Y}} \E \Big[ d(S,V)\Big]  - \min_{h_{V^n|Y^n}} \frac{1}{n} \sum_{i=1}^n\E\Big[d(S_i,V_i)\Big]\bigg| + \varepsilon \label{eq:ConverseThird2}\displaybreak[0]\\
= &\bigg| \min_{\PP_{V|W_1Y}} \E \Big[ d(S,V)\Big]  -  \sum_{y^n,m, F} \prob(y^n,m, F) \nonumber\\
& \times  \min_{h_{V^n|Y^n}} \sum_{s^n} \prob(s^n|y^n,F) \bigg[\frac{1}{n} \sum_{i=1}^n d(s_i,v_i)\bigg] \bigg| + \varepsilon \label{eq:ConverseThird3}\displaybreak[0]\\
\leq &\bigg| \min_{\PP_{V|W_1Y}} \E \Big[ d(S,V)\Big]  -  \sum_{y^n,m} \prob(y^n,m| F=0)\nonumber\\
&  \times  \min_{h_{V^n|Y^n}} \sum_{s^n} \prob(s^n|y^n,F=0) \bigg[\frac{1}{n} \sum_{i=1}^n d(s_i,v_i)\bigg] \bigg| \nonumber \\&+ \prob(F=1) 2\bar{d} + \varepsilon \label{eq:ConverseThird4}\displaybreak[0]\\
= &\bigg| \min_{\PP_{V|W_1Y}} \E \Big[ d(S,V)\Big]  -  \min_{\PP_{V|W_1Y}} \E \Big[ d(S,V)\Big] \bigg|\nonumber\\
&  + \prob(F=1) 2\bar{d} + \varepsilon \label{eq:ConverseThird5}\displaybreak[0]\\
= & \prob(F=1) 2\bar{d} + \varepsilon  \leq \varepsilon \cdot(2\bar{d} + 1)\label{eq:ConverseThird6},
\end{align}
where \eqref{eq:ConverseThird1} comes from the triangle inequality; 
\eqref{eq:ConverseThird2} comes from assumption \eqref{eq:DefinitionEstimation3}; 
\eqref{eq:ConverseThird3} is a reformulation that introduces the random event $F\in \{0,1\}$; 
\eqref{eq:ConverseThird4} comes from removing the \mrb{term} $\prob(F=1) \bar{d}$ from the  triangle inequality; 
\eqref{eq:ConverseThird5} comes from \eqref{eq:ConverseSecond6}; 
\eqref{eq:ConverseThird6} comes from assumption \eqref{eq:DefinitionEstimation2}. \mrb{This concludes the converse proof of Theorem \ref{theo:DecoderEstimation}.}


 



\begin{IEEEbiographynophoto}{Ma\"{e}l~Le~Treust}
(Member, IEEE) received the M.Sc. degree in optimization, game theory, and economics (OJME) from the Universit\'{e} de Paris VI (UPMC), France, in 2008, and the Ph.D. degree from the Laboratoire des signaux et syst\`{e}mes UMR 8506, Universit\'{e} de Paris Sud XI, Gif-sur-Yvette, France, in 2011. From 2012 to 2013, he was a Post-Doctoral Researcher with the Institut d'\'{e}lectronique et d'informatique Gaspard Monge, Universit\'{e} Paris- Est, Marne-la-Vall\'{e}e, France, and with the Centre \'{E}nergie, Mat\'{e}riaux et T\'{e}l\'{e}communication, Universit\'{e} INRS, Montr\'{e}al, Canada. He has been a CNRS Researcher with the ETIS Laboratory UMR 8051, CY Cergy Paris Universit\'{e}, ENSEA, CNRS, Cergy, France, since October 2013. His research interests include information theory, wireless communications, game theory, economics, and optimization.
\end{IEEEbiographynophoto}

\begin{IEEEbiographynophoto}{Matthieu~R.~Bloch}
 (Senior Member, IEEE) received the Engineering degree from Sup\'{e}lec, Gif-sur-Yvette, France, the M.S. degree in electrical engineering from the Georgia Institute of Technology, Atlanta, GA, USA, in 2003, the Ph.D. degree in engineering science from the Universit\'{e} de Franche-Comt\'{e}, Besan\c{c}on, France in 2006, and the Ph.D. degree in electrical engineering from the Georgia Institute of Technology in 2008. From 2008 to 2009, he was a Post-Doctoral Research Associate with the University of Notre Dame, South Bend, IN, USA. Since July 2009, he has been on the faculty of the School of Electrical and Computer Engineering. From 2009 to 2013, he was based with Georgia Tech Lorraine. He is currently an Associate Professor with the School of Electrical and Computer Engineering, Georgia Institute of Technology. He is a coauthor of the textbook \emph{Physical-Layer Security: From Information Theory to Security Engineering} (Cambridge University Press). 
His research interests are in the areas of information theory, error-control coding, wireless communications, and security. He was a co-recipient of the IEEE Communications Society and the IEEE Information Theory Society 2011 Joint Paper Award. He has served on the organizing committee of several international conferences. He was the Chair of the Online Committee of the IEEE Information Theory Society from 2011 to 2014. He has been on the Board of Governors of the IEEE Information Theory Society since 2016. He was an Associate Editor of the IEEE TRANSACTIONS ON INFORMATION THEORY from 2016 to 2019. He has been an Associate Editor of the IEEE TRANSACTIONS ON INFORMATION FORENSICS AND SECURITY since 2019.
\end{IEEEbiographynophoto}


%
%

\newpage
\setcounter{page}{1}

\begin{center}
\begin{large}
\textsc{\textbf{State Leakage and Coordination\\ with Causal State Knowledge at the Encoder}\\
-- Supplementary Materials -- }\\
\end{large}
Ma\"{e}l~Le~Treust and Matthieu~R.~Bloch\\
\end{center}

\mrb{
\section{Case of $H(S)=0$ in the achievability proof of Theorem \ref{theo:LeakageCE}}\label{sec:HSzeroAchievabilityCE}
In this section, we consider the case where $H(S)=0$. We choose a target distribution $\QQ_{XYV}$ that decomposes as $\QQ_{XYV} =   \QQ_{X} \mc{T}_{Y|X}      \QQ_{V|XY}$. Standard channel coding arguments, see \cite[Sec. 3.1]{ElGamalKim(book)11}, show that a pair  $(\textsf{R},\QQ_{XYV} )$ is achievable if and only if 
\begin{align}
\textsf{R} &\leq  I(X; Y).\label{eq:theoremCEachie3_HSzero}
\end{align}
We recover the information constraints of Theorem \ref{theo:LeakageCE} by setting $\textsf{E}=0$, $W_1=X$ and $|\mc{W}_2|=1$.}

\mrb{
\section{Case of equality in \eqref{eq:RateConstraint2ce}}\label{sec:EqualityAchievabilityCE}
In this section, we consider $(\textsf{R},\textsf{E})$ and a distribution $ \QQ_{SW_1W_2XYV} \in \Q_{\sf{c}}$ that satisfies
\begin{align}
 I(S;W_1,W_2,Y) = \textsf{E}&\leq H(S),\label{eq:RateConstraint2ceEquality}\\
 \textsf{R} + \textsf{E} &\leq I(W_1,S; Y). \label{eq:RateConstraint3ceEquality}
\end{align}
\begin{itemize}
\item[$\bullet$] If $ I(S;W_1,W_2,Y)<I(W_1,S; Y)$, then  we select $\widetilde{\textsf{E}}$ close to $\textsf{E}$ such that $I(S;W_1,W_2,Y) < \widetilde{\textsf{E}}\leq H(S)$ and ${\textsf{R}} + \widetilde{\textsf{E}} \leq I(W_1,S; Y)$, and thus the achievability proof of Appendix \ref{sec:AchievabilityLeakageCE} applies.
\item[$\bullet$] If $ I(S;W_1,W_2,Y)=I(W_1,S; Y)$ and the channel capacity, see  \cite[pp. 176]{ElGamalKim(book)11}, is strictly positive 
\begin{align}
\max_{\widetilde{\PP}_{W_1} ,\widetilde{\PP}_{X|W_1S}} I(W_1;Y) >0,\label{eq:ChannelCapacity}
\end{align}
then we refer to Sec. \ref{sec:StrictlyPositiveChannelCapacity}.
\item[$\bullet$] If $ I(S;W_1,W_2,Y)=I(W_1,S; Y)$ and the channel capacity in \eqref{eq:ChannelCapacity} is equal to zero, then we refer to Sec. \ref{sec:ZeroChannelCapacity}.
\end{itemize}}

\mrb{
\subsection{Strictly positive channel capacity}\label{sec:StrictlyPositiveChannelCapacity}
In this section, we assume that $ I(S;W_1,W_2,Y)=I(W_1,S; Y)$ and the channel capacity is strictly positive. 
We denote by $\PP^{\star}_{W_1}$ and $\PP^{\star}_{X|W_1S}$ the distributions that achieve the maximum in \eqref{eq:ChannelCapacity}. We define $\QQ^{\star}_{SW_1W_2XYV} = \PP_S     \PP^{\star}_{W_1}   \PP_{W_2}  \PP^{\star}_{X|W_1S} \mc{T}_{Y|XS} \QQ_{V} $ where  $V$ is independent of $(S,W_1,W_2,X,Y)$ and $W_2$ is independent of $(S,W_1,X,Y)$, hence $I(W_2;S|W_1)=I(W_2;Y|W_1)=0$, and therefore
\begin{align}
I_{\QQ^{\star}}(W_1,W_2;Y) - I_{\QQ^{\star}}(W_2;S|W_1)  = I_{\QQ^{\star}}(W_1;Y) >0.\label{eq:DistributionStar}
\end{align}
We assume that the distributions $\QQ_{SW_1W_2XYV}$ and $\QQ^{\star}_{SW_1W_2XYV}$ are selected according to the random variable $T\in\{0,1\}$, where $T=0$ corresponds to the distribution $\QQ_{SW_1W_2XYV}$ and $T=1$ corresponds to the distribution $\QQ^{\star}_{SW_1W_2XYV}$. By hypothesis, we have
\begin{align}
0<&\prob(T=1) \cdot I_{\QQ^{\star}}(W_1;Y) \\
=& \prob(T=0) \cdot\Big(I_{\QQ}(W_1,W_2;Y|T=0) - I_{\QQ}(W_2;S|W_1,T=0)\Big)\\
& + \prob(T=1) \cdot\Big(I_{\QQ^{\star}}(W_1,W_2;Y|T=1) - I_{\QQ^{\star}}(W_2;S|W_1,T=1)\Big)\\
=&I(W_1,W_2;Y|T) - I(W_2;S|W_1,T)\\
\leq &I(W_1,T,W_2;Y) - I(W_2;S|W_1,T)\label{eq:ICstrict}.
\end{align}
This show the existence of a distribution $\widetilde{\QQ}_{SW_1TW_2XYV} $ that decomposes as
\begin{align}
 \PP_S  \widetilde{\QQ}_{W_1T}    \widetilde{\QQ}_{W_2|SW_1T}  \widetilde{\QQ}_{X|SW_1T}
     \mc{T}_{Y|XS} \widetilde{\QQ}_{V|YW_1TW_2} , \label{eq:distributionT}
\end{align}
that converges to the distribution $\QQ_{SW_1W_2XYV}$ as $\prob(T=1)$ goes to zero, and that satisfies the information constraint \eqref{eq:ICstrict} with strict inequality. Thus the achievability proof of Appendix \ref{sec:AchievabilityLeakageCE} applies.}

\mrb{
\subsection{Strictly positive channel capacity}\label{sec:ZeroChannelCapacity}
In this section, we assume that $ I(S;W_1,W_2,Y)=I(W_1,S; Y)$ and the channel capacity is equal to zero, hence
\begin{align}
0 =& \max_{\PP_{W_1}, \PP_{X|W_1S}}  I(W_1;Y)  \label{eq:ZeroCapacity1}\\
\geq&  I(W_1 ; Y  ) \label{eq:ZeroCapacity2}\\
=&  I(S;W_1,W_2,Y) - I(S; Y|W_1) \label{eq:ZeroCapacity3}\\
=&   I(S;W_2|W_1,Y)  \geq 0, \label{eq:ZeroCapacity4}
\end{align}
where \eqref{eq:ZeroCapacity3} comes from the hypothesis $  I(S;W_1,W_2,Y)=I(W_1,S; Y)$ and  \eqref{eq:ZeroCapacity4} comes from the independence between $W_1$ and $S$. }

\mrb{
Moreover, the two Markov chains $W_2 -\!\!\!\!\minuso\!\!\!\!-   (W_1 ,S) -\!\!\!\!\minuso\!\!\!\!- X$ and $Y -\!\!\!\!\minuso\!\!\!\!-   (X,S) -\!\!\!\!\minuso\!\!\!\!- (W_1, W_2)$ induce $W_2 -\!\!\!\!\minuso\!\!\!\!-   (W_1 ,S) -\!\!\!\!\minuso\!\!\!\!- Y$. Therefore, we have
\begin{align}
W_2 -\!\!\!\!\minuso\!\!\!\!-   (W_1 ,S) -\!\!\!\!\minuso\!\!\!\!- Y  \Longleftrightarrow &I( W_2 ; Y | W_1,S) = 0 ,\label{eq:W2MarkovChain1}\\
W_2 -\!\!\!\!\minuso\!\!\!\!-   (W_1 ,Y) -\!\!\!\!\minuso\!\!\!\!- S  \Longleftrightarrow &I( W_2 ; S | W_1,Y) = 0 .\label{eq:W2MarkovChain2}
\end{align}
As a consequence, for all $(w_1,s,y)\in\mc{W}_1\times \mc{S}\times \mc{Y}$ such that $\QQ(w_1,s,y)>0$, we have 
\begin{align}
\QQ(w_2|w_1,s,y) = \QQ(w_2|w_1,s) = \QQ(w_2|w_1,y).\label{eq:TwoMarkovProperties}
\end{align}
We design a coding scheme for this special case.\\
 \emph{Codebook.} We choose a typical sequence $W_1^n$ that is known by both the encoder and the decoder.\\
 \emph{Encoder.}  At stage $i \in \{1, \ldots,n\}$, the encoder observes the symbol $S_i$, recalls the sequence $W_1^n$ and generates $X^n$ according to the i.i.d. conditional distribution $\QQ_{X|SW_1}$.\\
 \emph{Decoder.}  The decoder observes the sequence of channel output $Y^n$, recalls the pre-defined sequence $W_1^n$ and generates sequence $W_2^n$ according to the i.i.d. conditional distribution 
 $\QQ_{W_2|YW_1}$ which, according to \eqref{eq:TwoMarkovProperties}, satisfies $\QQ_{W_2}(\cdot|w_1,y)=\QQ_{W_2}(\cdot|w_1,s)$ for all $(w_1,s,y)\in\mc{W}_1\times \mc{S}\times \mc{Y}$ such that $\QQ(w_1,s,y)>0$. Then, the decoder generates $V^n$ by using the conditional distribution $\QQ_{V|W_1W_2Y}$.}

\mrb{This coding scheme achieves the target distribution 
\begin{align}
&\PP_S    \QQ_{W_1}    \QQ_{W_2|SW_1}  \QQ_{X|SW_1}     \mc{T}_{Y|XS}    \QQ_{V|YW_1W_2}. \end{align}}

\section{Cardinality Bound for Theorem \ref{theo:LeakageCE}}\label{sec:Cardinality Bound}

\begin{lemma}\label{lemma:CardinalityBound}
We consider the following information constraints  with two auxiliary random variables $(W_1,W_2) $
\begin{align}
I(S;W_1,W_2,Y) \leq \textsf{E}&\leq H(S), \\
\textsf{R} + \textsf{E} &\leq  I(W_1,S; Y).
\end{align}
The cardinality of the supports of the auxiliary random variables $(W_1,W_2)$ are bounded by
\begin{align*}
\max(|\mc{W}_1|,|\mc{W}_2|) \leq   d+1 ,\qquad \text{ with } \qquad d= | \mc{S}\times  \mc{X}\times \mc{Y}\times \mc{V} | .
\end{align*}
\end{lemma}
This result is based on the Lemma of Fenchel-Eggleston-Carath\'eodory, see \cite[pp. 631]{ElGamalKim(book)11}, \mrb{and the proof is provided below. From} Lemma \ref{lemma:CardinalityBound}, \mrb{we deduce} the cardinality bounds of Theorems \ref{theo:LeakageCE2SI} and \ref{theo:LeakageCENoisyFeedback}.

\begin{proof}[Lemma \ref{lemma:CardinalityBound}]
We denote by $d = | \mc{S}\times \mc{X}\times \mc{Y}  \times \mc{V}  |   $, the cardinality of the product of the discrete sets. We consider the family  of continuous functions $h_i : \Delta(  \mc{S}\times \mc{X}\times \mc{Y}  \times \mc{V}  ) \mapsto \R$, with $i \in \{1,\ldots,d+1\}$, defined as follows:
\begin{align*}
h_i\Big(\PP_{SXYV|W_1W_2}\Big) = 
\begin{cases}
\PP_{SXYV|W_1W_2} (i) ,\; \text{ for } i \in \big\{ 1,\ldots, d -1\big\},\\
H(Y | S, W_1=w_1) ,\quad \text{ for } i= d,\\
H(S | Y, W_1=w_1,W_2=w_2) ,\quad \text{ for } i= d+1.
\end{cases}
\end{align*}
The support Lemma, see \cite[pp. 631]{ElGamalKim(book)11}, implies that there exists a pair of auxiliary random variables $(W'_1, W'_2) \sim \PP_{W_1'W_2'}$  defined on a  set $\mc{W}'_1 \times \mc{W}'_2$ with finite cardinality $\max( |\mc{W}'_1| ,  |\mc{W}'_2| ) \leq  d+1$ such that for all $i \in \{1 , \ldots,d+1\}$ we have:
\begin{align*}
\int_{\mc{W}_1 \times \mc{W}_2} h_i\Big(\PP_{SXYV|W_1W_2}\Big) dF(w_1,w_2) =& \sum_{(w'_1,w'_2) \in \mc{W}'_1 \times  \mc{W}'_2} h_i\Big(\PP_{SXYV|W_1'W_2'}\Big) \PP(w'_1w'_2).
\end{align*}
This implies that the probability $\PP_{SXYV}$ is preserved and we have:
\begin{align*}
\PP_{SXYV} (i) =& \int_{\mc{W}_1\times \mc{W}_2}  \PP_{SXYV|W_1W_2} (i) dF(w_1,w_2) \nonumber \\
=& \sum_{ (w'_1,w'_2) \in \mc{W}'_1 \times  \mc{W}'_2 }  \PP_{SXYV|W_1'W_2'} (i) \PP(w'_1,w'_2), \quad \text{ for } i \in \big\{ 1,\ldots,  d-1\big\}\\
H(Y | S, W_1) =& \int_{\mc{W}_1 }  H(Y| S, W_1 = w_1) dF(w_1) ) \nonumber \\
=& \sum_{(w'_1)  \in \mc{W}'_1  }  H(Y| S,W'_1 = w'_1) \PP(w'_1) = H(Y| S,W'_1),\\
H(S |Y, W_1 , W_2) =& \int_{\mc{W}_1 \times \mc{W}_2}  H(S|Y,W_1 = w_1 , W_2 = w_2) dF(w_1,w_2) ) \nonumber \\
=& \sum_{(w'_1,w'_2)  \in \mc{W}'_1 \times \mc{W}'_2 }  H(S|Y,W'_1 = w'_1, W'_2 = w'_2) \PP(w'_1,w'_2) = H(S|Y,W'_1,W'_2).
\end{align*}
Hence the three information constraints remain equal with $\max( |\mc{W}'_1| ,  |\mc{W}'_2| ) \leq  d+1$. 
\begin{align*}
I(S ; W_1,W_2,Y) =& H(S) - H(S|W'_1,W'_2,Y) = I(S ; W'_1,W'_2,Y),\\
I(W_1 , S ; Y )  =& H(Y) - H(Y | W'_1,S) = I(W'_1 , S ; Y ) .
\end{align*}
This concludes the proof of  Lemma \ref{lemma:CardinalityBound}. 
\end{proof}

\section{Converse proof of Theorem \ref{theo:LeakageCEd}}\label{sec:ConverseLeakageCEnoCO}

We consider that the pair of rate and state leakage $(\textsf{R},\textsf{E})$ is achievable with \mrb{a code with causal encoding}.  By definition \ref{def:CodeLeakageCEd}, for all $\varepsilon>0$, there exists $\bar{n}\in \N^{\star}$, for all $n\geq \bar{n}$, there exists a code $c^{\star} \in \C(n,\mc{M})$ that satisfies
\begin{align}
\frac{\log_2 |\mc{M}|}{n}  \geq& \textsf{R} - \varepsilon,\label{eq:converse1ceb} \\
 \bigg| \mc{L}_{\textsf{e}}(c)  - \textsf{E} \; \bigg|  =& \bigg|  \frac{1}{n} I(S^n;Y^n)  - \textsf{E} \;\bigg| \leq  \varepsilon,  \label{eq:converse3ceb} \\
 \PP_{\textsf{e}}(c)  =& \prob\bigg( M \neq \hat{M} \bigg)  \leq \varepsilon.\label{eq:converse2ceb} 
\end{align}
We introduce the auxiliary random variables $ W_{1,i} =  ( M , S^{i-1} )$  that satisfy the Markov chains of the set of distribution $\Q_{\sf{c}}$ for all $i \in \{1,\ldots,n\}$:
\begin{align}
&S_i \text{ independent of } W_{1,i}, \label{eq:ConverseCEbMarkov1} \\
&Y_i -\!\!\!\!\minuso\!\!\!\!- (X_i , S_i ) -\!\!\!\!\minuso\!\!\!\!-    W_{1,i} ,  \label{eq:ConverseCEbMarkov3}   
\end{align}
where \eqref{eq:ConverseCEbMarkov1} comes from the i.i.d. property of the source that induces the independence between $S_i$ and  $  ( M ,  S^{i-1} ) = W_{1,i} $; \eqref{eq:ConverseCEbMarkov3} comes from the memoryless property of the channel $ \mc{T}_{Y|XS}$.

We introduce the  random variable  $T$ that is uniformly distributed over the indices $\{1 , \ldots, n\}$ and the corresponding mean random variables $(S_T,X_T,W_{1,T},Y_T)$.  The  auxiliary random variables $W_1=(W_{1,T} , T)$   belongs to the set of distributions $\Q_{\sf{c}}$  and satisfies the three information constraints of Theorem \ref{theo:LeakageCEd}
\begin{align}
I(S;Y |W_1) \leq \textsf{E}&\leq H(S),\label{eq:theoremCEconverse2bbis}\\
\textsf{R} + \textsf{E} &\leq I(W_1,S; Y). \label{eq:theoremCEconverse3bbis}
\end{align}

 \textbf{First Constraint:}
\begin{align}
n \textsf{E} \geq & I(S^n ; Y^n) -  n\varepsilon\label{eq:SecondConverse1ceb}\\
= & I(S^n;Y^n , M)  -  I(S^n ; M | Y^n) -  n\varepsilon\label{eq:SecondConverse2ceb}\\
\geq& \sum_{i=1}^n  I(S_i;Y^n , M | S^{i-1}) -  H(M | Y^n)-  n\varepsilon \label{eq:SecondConverse3ceb}\\
\geq& \sum_{i=1}^n  I(S_i;Y^n , M | S^{i-1}) -  n2\varepsilon \label{eq:SecondConverse4ceb}\\
=& \sum_{i=1}^n  I(S_i;Y^n | M , S^{i-1}) -    n2\varepsilon \label{eq:SecondConverse5ceb}\\
\geq& \sum_{i=1}^n  I(S_i; Y_i | M , S^{i-1}) -    n2\varepsilon \label{eq:SecondConverse6ceb}\\
=& \sum_{i=1}^n  I(S_i; Y_i| W_{1,i}  ) -    n2\varepsilon \label{eq:SecondConverse7ceb}\\
=& n   I(S_T;  Y_T | W_{1,T} , T ) -    n2\varepsilon \label{eq:SecondConverse8ceb}\\
= & n  \bigg(  I(S ; Y | W_1)  -  2\varepsilon \bigg),  \label{eq:SecondConverse10ceb}
\end{align}
where  \eqref{eq:SecondConverse1ceb} comes from the definition of achievable state leakage rate $\textsf{E}$, stated in equation \eqref{eq:converse3ceb}; \eqref{eq:SecondConverse2ceb}   and  \eqref{eq:SecondConverse3ceb}  come from the properties of the mutual information; \eqref{eq:SecondConverse4ceb} comes from equation \eqref{eq:converse2ceb} and  Fano's inequality, see  \cite[pp. 19]{ElGamalKim(book)11}; \eqref{eq:SecondConverse5ceb}  comes from the independence between the message $M$ and the channel states $(S^{i-1} , S_i)$; \eqref{eq:SecondConverse6ceb}  comes from the properties of the mutual information; \eqref{eq:SecondConverse7ceb}  comes from the introduction of the auxiliary random variable $ W_{1,i} =  ( M,S^{i-1}  )$, for all $i \in \{1,\ldots,n\}$; \eqref{eq:SecondConverse8ceb}   comes from the introduction of the uniform random variable $T$ over $\{1 , \ldots, n\}$ and the corresponding mean random variables  $S_T$,  $W_{1,T}$,  $Y_T$; \eqref{eq:SecondConverse10ceb}   comes from identifying $W_1=(W_{1,T} , T)$ and $S = S_T$,  $Y = Y_T$.

\textbf{Second Constraint:}
From \eqref{eq:ThirdConverse3ce}, we have
\begin{align}
n \textsf{E} \leq &  n  \bigg( H(S) + \varepsilon \bigg). 
\end{align}


\textbf{Third Constraint:}
From \eqref{eq:FourthConverse1ce} - \eqref{eq:FourthConverse16ce}, we have
\begin{align}
n \bigg( \textsf{E} +  \textsf{R} \bigg) 
\leq n   \bigg( I(S ,W_1 ; Y)   +  3\varepsilon \bigg).
\end{align}

 \textbf{Conclusion:} 
If the pair of rate and state leakage $(\textsf{R},\textsf{E})$ is achievable with \mrb{a code with} causal encoding, then the following equations are satisfied  for all $\varepsilon>0$
\begin{align}
 I(S;Y | W_1) - 2 \varepsilon \leq \textsf{E}&\leq H(S) + \varepsilon,\\
\textsf{R} + \textsf{E} &\leq I(S ,W_1 ; Y)+ 3 \varepsilon.
\end{align}
This corresponds to equations  \eqref{eq:theoremCEachie2noCO} and  \eqref{eq:theoremCEachie3noCO} and \mrb{this} concludes the converse proof of Theorem \ref{theo:LeakageCEd}.

\begin{remark}
For the converse proof of Theorem \ref{theo:LeakageCEd}, the \mrb{code with} causal encoding is not necessarily deterministic. The same optimal performances can be obtained by considering a stochastic \mrb{code with} causal encoding.
\end{remark}

\section{Converse proof of Theorem \ref{theo:LeakageCE2SI}}\label{sec:ConverseTheoSCE2SI}

Consider that the triple of rate, state leakage and distribution $(\textsf{R},\textsf{E},\mrb{\QQ_{USZXYV}})$ is achievable with a \mrb{code with} causal \mrb{encoding}. \mrb{We simplify the notation by using $\QQ$ in place of $\QQ_{USZXYV}$, and we} introduce the random  event of error $E \in \{0,1\}$ defined with respect to the achievable distribution $\QQ$ \mrb{by}
\begin{align}
E = \Bigg\{
\begin{array}{lll}
0 \text{ if }& \big|\big|{Q}^n - \QQ  \big|\big|_{1} \leq \varepsilon\quad \Longleftrightarrow \quad (U^n,S^n,Z^n,X^n,Y^n,V^n) \in T_{\delta}(\QQ) ,\\
1 \text{ if }& \big|\big|{Q}^n - \QQ  \big|\big|_{1} > \varepsilon \quad \Longleftrightarrow \quad (U^n,S^n,Z^n,X^n,Y^n,V^n) \notin T_{\delta}(\QQ).
\end{array}
\Bigg.
\end{align}
The event $E=1$ occurs if the sequences $(U^n,S^n,Z^n,X^n,Y^n,V^n)\notin T_{\delta}(\QQ)$ for the target distribution $\QQ$. By definition \ref{def:CodeLeakageCE2SI}, for all $\varepsilon>0$, there exists $\bar{n}\in \N^{\star}$, for all $n\geq \bar{n}$, there exists a code $c^{\star} \in \C(n,\mc{M})$ that satisfies 
\begin{align}
\frac{\log_2 |\mc{M}|}{n}  \geq& \textsf{R} - \varepsilon, \label{eq:converse1ce2si}  \\
 \bigg| \mc{L}_{\textsf{e}}(c^{\star})  - \textsf{E} \bigg|  =& \bigg|  \frac{1}{n} I(U^n,S^n;Y^n,Z^n)  - \textsf{E} \bigg| \leq  \varepsilon,   \label{eq:converse3ce2si} \\
\PP_{\textsf{e}}(c^{\star})  =& \prob\bigg( M \neq \hat{M} \bigg) + \prob\bigg(\Big|\Big|Q^n - \QQ \Big|\Big|_{1}> \varepsilon\bigg) \leq \varepsilon. \label{eq:converse2ce2si}  
\end{align}
We introduce the auxiliary random variables $ W_{1,i} =  ( M , U^{i-1},S^{i-1} )$ and $ W_{2,i} =  (Y_{i+1}^n ,Z_{i+1}^n ) $, that satisfy the Markov chain of the set of distribution $\Q_{\textsf{s}}$ for all $i \in \{1,\ldots,n\}$:
\begin{align}
&(U_i,S_i) \text{ independent of } W_{1,i}, \label{eq:ConverseCE2SIMarkov1} \\
&X_i -\!\!\!\!\minuso\!\!\!\!- (U_i,S_i ,W_{1,i}) -\!\!\!\!\minuso\!\!\!\!-  W_{2,i} , \label{eq:ConverseCE2SIMarkov2} \\ 
&Y_i -\!\!\!\!\minuso\!\!\!\!- (X_i , S_i ) -\!\!\!\!\minuso\!\!\!\!-  ( U_i, Z_i,W_{1,i},W_{2,i} ) , \label{eq:ConverseCE2SIMarkov3}\\
&Z_i -\!\!\!\!\minuso\!\!\!\!- (U_i , S_i ) -\!\!\!\!\minuso\!\!\!\!-  ( X_i, Y_i,W_{1,i},W_{2,i} ) , \label{eq:ConverseCE2SIMarkov4} \\
&V_i -\!\!\!\!\minuso\!\!\!\!- (Y_i,Z_i,W_{1,i}, W_{2,i} ) -\!\!\!\!\minuso\!\!\!\!-  (U_i,S_i ,X_i) , \label{eq:ConverseCE2SIMarkov5}
\end{align}
where \eqref{eq:ConverseCE2SIMarkov1} comes from the i.i.d. property of the source that induces the independence between $(U_i,S_i)$ and  $  ( M , U^{i-1},S^{i-1} ) = W_{1,i} $; \eqref{eq:ConverseCE2SIMarkov2} comes from Lemma  \ref{lemma1:ConverseCE2SIMarkov1}. It is a direct consequence of the causal encoding function, the memoryless property of the channel and the i.i.d. property of the source; \eqref{eq:ConverseCE2SIMarkov3} comes from the memoryless property of the channel $ \mc{T}_{Y|XS}$; \eqref{eq:ConverseCE2SIMarkov4} comes from the i.i.d. property of the source $   \PP_{USZ}$; \eqref{eq:ConverseCE2SIMarkov5} comes from Lemma \ref{lemma2:ConverseCE2SIMarkov1}. It is a direct consequence of the causal encoding function, the non-causal decoding function, the memoryless property of the channel and the i.i.d. property of the source.

We introduce the random variable $T$ that is uniformly distributed over the indices $\{1 , \ldots, n\}$ and the corresponding mean random variables $W_{1,T}$, $W_{2,T}$, $U_T$, $S_T$, $Z_T$, $X_T$, $Y_T$, $V_T$. The  auxiliary random variables $W_1=(W_{1,T} , T)$ and $W_2 = W_{2,T}$ belong to the set of distributions $\Q_{\textsf{s}}$ and satisfy the three information constraints of Theorem \ref{theo:LeakageCE2SI}:
\begin{align}
 I(U,S;W_1,W_2,Y,Z) \leq \textsf{E}&\leq H(U,S),\label{eq:theoremCE2SIconverse2b}\\
 \textsf{R} + \textsf{E} &\leq I(W_1,U,S; Y,Z). \label{eq:theoremCE2SIconverse3b}
\end{align}

 \textbf{First Constraint:}
\begin{align}
n \textsf{E} \geq & I(U^n , S^n ; Y^n , Z^n) -  n\varepsilon\label{eq:SecondConverse1ce2si}\\
= & I(U^n , S^n;Y^n , Z^n , M)  -  I(U^n , S^n ; M | Y^n , Z^n) -  n\varepsilon\label{eq:SecondConverse2ce2si}\displaybreak[0]\\
\geq& \sum_{i=1}^n  I(U_i , S_i;Y^n , Z^n , M | U^{i-1} , S^{i-1}) -  H(M | Y^n , Z^n)-  n\varepsilon \label{eq:SecondConverse3ce2si}\displaybreak[0]\\
\geq& \sum_{i=1}^n  I(U_i , S_i;Y^n , Z^n , M | U^{i-1} , S^{i-1}) -  n2\varepsilon \label{eq:SecondConverse4ce2si}\displaybreak[0]\\
=& \sum_{i=1}^n  I(U_i , S_i;Y^n , Z^n , M , U^{i-1} , S^{i-1}) -    n2\varepsilon \label{eq:SecondConverse5ce2si}\displaybreak[0]\\
\geq& \sum_{i=1}^n  I(U_i , S_i;Y_{i+1}^n , Z_{i+1}^n , M , U^{i-1} , S^{i-1},Y_i , Z_i) -    n2\varepsilon \label{eq:SecondConverse6ce2si}\displaybreak[0]\\
=& \sum_{i=1}^n  I(U_i , S_i; W_{1,i} , W_{2,i},Y_i , Z_i ) -    n2\varepsilon \label{eq:SecondConverse7ce2si}\displaybreak[0]\\
=& n   I(U_T , S_T; W_{1,T} , W_{2,T} ,Y_T , Z_T| T ) -    n2\varepsilon \label{eq:SecondConverse8ce2si}\displaybreak[0]\\
=& n   I(U_T , S_T; W_{1,T} , W_{2,T},Y_T , Z_T, T ) -    n2\varepsilon \label{eq:SecondConverse9ce2si}\displaybreak[0]\\
= & n  \bigg(  I(U_T , S_T ; W_1,W_2,Y_T , Z_T)  -  2\varepsilon \bigg)  \label{eq:SecondConverse10ce2si}\displaybreak[0]\\
\geq& n  \bigg(  I(U_T , S_T ; W_1,W_2,Y_T , Z_T | E=0)  -  3\varepsilon \bigg)  \label{eq:SecondConverse11ce2si}\displaybreak[0]\\
\geq & n  \bigg(  I( U , S ; W_1,W_2,Y, Z)  -  4\varepsilon \bigg)  \label{eq:SecondConverse12ce2si},
\end{align}
where  \eqref{eq:SecondConverse1ce2si} comes from the definition of achievable state leakage rate $\textsf{E}$, stated in equation \eqref{eq:converse3ce2si}; \eqref{eq:SecondConverse2ce2si}   and  \eqref{eq:SecondConverse3ce2si}  come from the properties of the mutual information; \eqref{eq:SecondConverse4ce2si} comes from equation \eqref{eq:converse2ce2si} and  Fano's inequality, see \cite[pp. 19]{ElGamalKim(book)11}; \eqref{eq:SecondConverse5ce2si}  comes from the i.i.d. property of the channel states that implies $(U_i , S_i)$ is independent of $(U^{i-1} , S^{i-1})$; \eqref{eq:SecondConverse6ce2si}  comes from the properties of the mutual information; \eqref{eq:SecondConverse7ce2si}  comes from the introduction of the auxiliary random variables $ W_{1,i} =  ( M,U^{i-1} , S^{i-1}  )$ and $ W_{2,i} =   (Y_{i+1}^n , Z_{i+1}^n )$, for all $i \in \{1,\ldots,n\}$; \eqref{eq:SecondConverse8ce2si}   comes from the introduction of the uniform random variable $T$ over $\{1 , \ldots, n\}$ and the corresponding mean random variables  $U_T $, $S_T$,  $W_{1,T}$, $W_{2,T}$,  $Y_T$, $Z_T$; \eqref{eq:SecondConverse9ce2si}   comes from the independence between $T$ and $(U_T , S_T)$; \eqref{eq:SecondConverse10ce2si}   comes from \mrb{the identification of} the auxiliary random variables $W_1=(W_{1,T} , T)$ and $W_2 = W_{2,T}$; \eqref{eq:SecondConverse11ce2si}  comes from the empirical coordination requirement as stated in Lemma \ref{lemma:ErrorEventCoordination01}. The sequences of symbols $(U^n , S^n,Z^n, X^n,Y^n, V^n)$ are not jointly typical with small error probability $\prob(E=1)$; \eqref{eq:SecondConverse12ce2si} comes from Lemma \ref{lemma:3}. The sequences of symbols $(U^n , S^n,Z^n, X^n,Y^n, V^n)$ are jointly typical, hence the distribution of the mean random variables $ \PP_{U_TS_TZ_TX_TY_TV_T|E=0}$ is close to the target distribution $\QQ_{USZXYV}$. The \mrb{result of} \cite[Lemma 2.7, pp. 19]{CsiszarKorner(Book)11} concludes.

 \textbf{Second Constraint:}
\begin{align}
n \textsf{E} \leq & I(U^n , S^n ; Y^n, Z^n) + n\varepsilon \label{eq:ThirdConverse1ce2si}\\
\leq & H(U^n , S^n)  + n\varepsilon \label{eq:ThirdConverse2ce2si}\\
= & n  \bigg( H(U, S) + \varepsilon \bigg), \label{eq:ThirdConverse3ce2si}
\end{align}
where \eqref{eq:ThirdConverse1ce2si} comes from the definition of the achievable state leakage rate $\textsf{E}$, stated in equation \eqref{eq:converse3ce2si}; \eqref{eq:ThirdConverse2ce2si} comes from the properties of the mutual information; \eqref{eq:ThirdConverse3ce2si} comes from the i.i.d. property of the channel states $(U,S)$.

 \textbf{Third Constraint:}
\begin{align}
n \bigg( \textsf{E} +  \textsf{R} \bigg) 
\leq & I(U^n , S^n ; Y^n , Z^n)  + H(M) +  n2\varepsilon \label{eq:FourthConverse1ce2si}\displaybreak[0]\\
= &  I(U^n , S^n ; Y^n , Z^n)  + I(M ; Y^n , Z^n ) +  H(M |Y^n , Z^n) +  n2\varepsilon \label{eq:FourthConverse2ce2si}\displaybreak[0]\\
\leq &  I(U^n , S^n ; Y^n , Z^n)  + I(M ; Y^n , Z^n ) +  n3\varepsilon \label{eq:FourthConverse3ce2si}\displaybreak[0]\\
\leq &  I(U^n , S^n ; Y^n , Z^n)  + I(M ; Y^n , Z^n | U^n , S^n) +  n3\varepsilon \label{eq:FourthConverse4ce2si}\displaybreak[0]\\
= &  I(U^n , S^n ,M ; Y^n , Z^n)  +  n3\varepsilon \label{eq:FourthConverse5ce2si}\displaybreak[0]\\
= & \sum_{i=1}^n I(U^n , S^n ,M ; Y_i , Z_i | Y_{i+1}^n , Z_{i+1}^n )  +  n3\varepsilon \label{eq:FourthConverse6ce2si}\displaybreak[0]\\
\leq & \sum_{i=1}^n I(U^n , S^n ,M ,  Y_{i+1}^n , Z_{i+1}^n   ; Y_i , Z_i )  +  n3\varepsilon \label{eq:FourthConverse7ce2si}\displaybreak[0]\\
= & \sum_{i=1}^n I(U_i , S_i ,M , U^{i-1} , S^{i-1}   ; Y_i , Z_i )\nonumber \displaybreak[0]\\
+&  \sum_{i=1}^n I(U_{i+1}^n , S_{i+1}^n , Y_{i+1}^n , Z_{i+1}^n   ;  Y_i  | U_i , S_i ,M , U^{i-1} , S^{i-1} )\nonumber \displaybreak[0]\\
+&  \sum_{i=1}^n I(U_{i+1}^n , S_{i+1}^n , Y_{i+1}^n , Z_{i+1}^n   ;  Z_i  | U_i , S_i ,M , U^{i-1} , S^{i-1} , Y_i ) +  n3\varepsilon \label{eq:FourthConverse8ce2si}\displaybreak[0]\\
= & \sum_{i=1}^n I(U_i , S_i ,M , U^{i-1} , S^{i-1}   ; Y_i , Z_i )\nonumber \displaybreak[0]\\
+&  \sum_{i=1}^n I(U_{i+1}^n , S_{i+1}^n , Y_{i+1}^n , Z_{i+1}^n   ;  Z_i  | U_i , S_i ,M , U^{i-1} , S^{i-1} , Y_i ) +  n3\varepsilon \label{eq:FourthConverse9ce2si}\displaybreak[0]\\
= & \sum_{i=1}^n I(U_i , S_i ,M , U^{i-1} , S^{i-1}    ; Y_i , Z_i ) +  n3\varepsilon \label{eq:FourthConverse10ce2si}\displaybreak[0]\\
= &  \sum_{i=1}^n I(U_i , S_i ,W_{1,i} ; Y_i , Z_i)  +  n\cdot3\varepsilon  \label{eq:FourthConverse11ce2si}\displaybreak[0]\\
= &  n   I(U_T , S_T ,W_{1,T} ; Y_T , Z_T | T)   +  n\cdot3\varepsilon \label{eq:FourthConverse12ce2si}\displaybreak[0]\\
\leq &  n   I(U_T , S_T ,W_{1,T} ,T; Y_T , Z_T)   +  n\cdot3\varepsilon \label{eq:FourthConverse13ce2si}\displaybreak[0]\\
= &  n   \bigg( I(U_T , S_T ,W_{1} ; Y_T , Z_T)   +  3\varepsilon \bigg)\label{eq:FourthConverse14ce2si}\displaybreak[0]\\
\leq &  n   \bigg( I(U_T , S_T ,W_{1} ; Y_T , Z_T |E = 0)   +  4\varepsilon \bigg)\label{eq:FourthConverse15ce2si}\displaybreak[0]\\
\leq &  n   \bigg( I(U,S ,W_{1} ; Y,Z )   +  5\varepsilon \bigg)\label{eq:FourthConverse16ce2si},
\end{align}
where \eqref{eq:FourthConverse1ce2si} comes from the definition of achievable rate and information  leakage $( \textsf{R} ,  \textsf{E})$, stated in equations \eqref{eq:converse1ce2si} and  \eqref{eq:converse3ce2si};  \eqref{eq:FourthConverse3ce2si} comes from equation \eqref{eq:converse2ce2si} and  Fano's inequality, see \cite[pp. 19]{ElGamalKim(book)11}; \eqref{eq:FourthConverse4ce2si}  comes from the independence between the message $M$ and the channel states $(U^n , S^n)$, hence $I(M ; Y^n , Z^n ) \leq I(M ; Y^n , Z^n ,U^n , S^n ) = I(M ; Y^n , Z^n |U^n , S^n ) $; \eqref{eq:FourthConverse9ce2si}  comes from the Markov chain $ Y_i  -\!\!\!\!\minuso\!\!\!\!-  (U_i , S_i ,M , U^{i-1} , S^{i-1}) -\!\!\!\!\minuso\!\!\!\!- (U_{i+1}^n , S_{i+1}^n , Y_{i+1}^n , Z_{i+1}^n )$ stated in Lemma \ref{lemma3:ConverseCE2SIMarkov1}, \mrb{which} is a direct consequence of the causal encoding function and the memoryless property of the channel; \eqref{eq:FourthConverse10ce2si}  comes from the i.i.d. property of the source that induces the  Markov chain $ Z_i  -\!\!\!\!\minuso\!\!\!\!-   (U_i , S_i  ) -\!\!\!\!\minuso\!\!\!\!-  ( U_{i+1}^n , S_{i+1}^n , Y_{i+1}^n , Z_{i+1}^n ,M , U^{i-1} , S^{i-1} , Y_i) $, hence we have:  $ I(U_{i+1}^n , S_{i+1}^n , Y_{i+1}^n , Z_{i+1}^n   ;  Z_i  | U_i , S_i ,M , U^{i-1} , S^{i-1} , Y_i )  = 0$; \eqref{eq:FourthConverse11ce2si}  comes from the introduction of the auxiliary random variable $ W_{1,i} =  ( M,U^{i-1} , S^{i-1}  )$, for all $i \in \{1,\ldots,n\}$; \eqref{eq:FourthConverse12ce2si}  comes from the introduction of the uniform random variable $T$ over $\{1 , \ldots, n\}$ and the corresponding mean random variables  $U_T , S_T$,  $W_{1,T}$, $Y_T , Z_T$; \eqref{eq:FourthConverse14ce2si}   comes from \mrb{the identification of} the auxiliary random variable $W_1=(W_{1,T} , T)$; \eqref{eq:FourthConverse15ce2si}  comes from the empirical coordination requirement as stated in Lemma \ref{lemma:ErrorEventCoordination0}. The sequences of symbols $(U^n , S^n,Z^n, X^n,Y^n , V^n)$ are not jointly typical with small error probability $\prob(E=1)$; \eqref{eq:FourthConverse16ce2si} comes from Lemma \ref{lemma:3}, \mrb{since} the sequences of symbols $(U^n , S^n,Z^n, X^n,Y^n, V^n)$ are jointly typical, hence the distribution of the mean random variables $ \PP_{U_TS_TZ_TX_TY_TV_T|E=0}$ is close to the target distribution $\QQ_{USZXYV}$.  The \mrb{result of} \cite[Lemma 2.7, pp. 19]{CsiszarKorner(Book)11} concludes.

 \textbf{Conclusion:} If the triple of rate, state leakage and distribution$(\textsf{R},\textsf{E},\mrb{\QQ_{USZXYV}})$ is achievable with \mrb{a code with} causal encoding, then  the following equations are satisfied  for all $\varepsilon>0$:
\begin{align}
 I(U,S;W_1,W_2,Y,Z) - 4 \varepsilon \leq \textsf{E}&\leq H(S) + \varepsilon,\\
\textsf{R} + \textsf{E} &\leq I(U,S ,W_{1} ; Y,Z)+ 5 \varepsilon.
\end{align}
This corresponds to \eqref{eq:theoremCE2SIachie2} and  \eqref{eq:theoremCE2SIachie3} and \mrb{this} concludes the converse proof of Theorem \ref{theo:LeakageCE2SI}. 

\begin{remark}
For the converse proof of Theorem \ref{theo:LeakageCE2SI}, the \mrb{a code with} causal encoding is not necessarily deterministic. The same optimal performances can be obtained by considering a stochastic \mrb{code with} causal encoding.
\end{remark}


\begin{lemma}\label{lemma1:ConverseCE2SIMarkov1}
The causal encoding function, the memoryless property of the channel and the i.i.d. property of the source induce the Markov chain property 
 \begin{align}
  X_i -\!\!\!\!\minuso\!\!\!\!- (U_i , S_i ,  W_{1,i}) -\!\!\!\!\minuso\!\!\!\!-    W_{2,i} .
\end{align}
This Markov chain  is satisfied with $ W_{1,i} =  ( M , U^{i-1},S^{i-1} )$ and $ W_{2,i} =  (Y_{i+1}^n ,Z_{i+1}^n ) $, for all $i \in \{1,\ldots,n\}$.
\end{lemma}

\begin{proof}[Lemma \ref{lemma1:ConverseCE2SIMarkov1}] 
The auxiliary random variables $ W_{1,i} =  ( M , U^{i-1},S^{i-1} )$ and $ W_{2,i} =  (Y_{i+1}^n ,Z_{i+1}^n ) $ satisfy the following equations for all $(u^n,s^n,z^n,x^n,y^n,v^n,m)$:
\begin{align}
&\PP( w_{2,i}   | u_i,s_i , w_{1,i} ,x_i)  = \PP( y_{i+1}^n ,z_{i+1}^n  | u_i,s_i , m , u^{i-1},s^{i-1} ,x_i) \nonumber \displaybreak[0]\\
=&\sum_{u_{i+1}^n ,s_{i+1}^n,x_{i+1}^n}  \PP( u_{i+1}^n ,s_{i+1}^n,x_{i+1}^n ,y_{i+1}^n ,z_{i+1}^n  | u_i,s_i , m , u^{i-1},s^{i-1} ,x_i)  \label{eq:lemma1ConverseCE2sI1}   \displaybreak[0]\\
=&\sum_{u_{i+1}^n ,s_{i+1}^n,x_{i+1}^n}  \PP( u_{i+1}^n ,s_{i+1}^n,x_{i+1}^n  | u_i,s_i , m , u^{i-1},s^{i-1} ,x_i) \nonumber \displaybreak[0]\\
& \PP( y_{i+1}^n ,z_{i+1}^n |u_{i+1}^n ,s_{i+1}^n,x_{i+1}^n ,u_i,s_i , m , u^{i-1},s^{i-1} ,x_i)   \label{eq:lemma1ConverseCE2sI2}  \displaybreak[0]\\
=&\sum_{u_{i+1}^n ,s_{i+1}^n,x_{i+1}^n}  \PP( u_{i+1}^n ,s_{i+1}^n,x_{i+1}^n  | u_i,s_i , m , u^{i-1},s^{i-1} )  \nonumber \displaybreak[0]\\
& \PP( y_{i+1}^n ,z_{i+1}^n |u_{i+1}^n ,s_{i+1}^n,x_{i+1}^n ,u_i,s_i , m , u^{i-1},s^{i-1} ,x_i)    \label{eq:lemma1ConverseCE2sI3}  \displaybreak[0]\\
=&\sum_{u_{i+1}^n ,s_{i+1}^n,x_{i+1}^n}  \PP( u_{i+1}^n ,s_{i+1}^n,x_{i+1}^n  | u_i,s_i , m , u^{i-1},s^{i-1} )\PP( y_{i+1}^n   | s_{i+1}^n,x_{i+1}^n) \PP(z_{i+1}^n  | u_{i+1}^n ,s_{i+1}^n)   \label{eq:lemma1ConverseCE2sI4}  \displaybreak[0]\\
=&\sum_{u_{i+1}^n ,s_{i+1}^n,x_{i+1}^n} \PP( u_{i+1}^n ,s_{i+1}^n,x_{i+1}^n ,y_{i+1}^n ,z_{i+1}^n  | u_i,s_i , m , u^{i-1},s^{i-1} )   \label{eq:lemma1ConverseCE2sI44} \displaybreak[0]\\
=& \PP( y_{i+1}^n ,z_{i+1}^n  | u_i,s_i , m , u^{i-1},s^{i-1} ) =  \PP( w_{2,i}   | u_i,s_i , w_{1,i}  ) ,  \label{eq:lemma1ConverseCE2sI5}
\end{align}
where \eqref{eq:lemma1ConverseCE2sI3} comes from the causal encoding function that induces the Markov chain $X_i -\!\!\!\!\minuso\!\!\!\!- (U_i,S_i , M , U^{i-1},S^{i-1}) -\!\!\!\!\minuso\!\!\!\!-  (U_{i+1}^n ,S_{i+1}^n,X_{i+1}^n)$; \eqref{eq:lemma1ConverseCE2sI4} comes from the memoryless property of the channel $ Y_{i+1}^n -\!\!\!\!\minuso\!\!\!\!- (S_{i+1}^n,X_{i+1}^n ) -\!\!\!\!\minuso\!\!\!\!-  (U_{i+1}^n , U_i,S_i , M , U^{i-1},S^{i-1} ,X_i)$ and the i.i.d. property of the source $ Z_{i+1}^n -\!\!\!\!\minuso\!\!\!\!- (S_{i+1}^n,U_{i+1}^n ) -\!\!\!\!\minuso\!\!\!\!-  (U_{i+1}^n , U_i,S_i , M , U^{i-1},S^{i-1} ,X_i,Y_{i+1}^n)$.  This concludes the proof of Lemma \ref{lemma1:ConverseCE2SIMarkov1}.
\end{proof}


\begin{lemma}\label{lemma2:ConverseCE2SIMarkov1}
The causal encoding function, the non-causal decoding function, the memoryless property of the channel and the i.i.d. property of the source induce the Markov chain property
 \begin{align}
 V_i  -\!\!\!\!\minuso\!\!\!\!- ( Y_i ,Z_i, W_{1,i} , W_{2,i}   )   -\!\!\!\!\minuso\!\!\!\!-  (U_i, S_i, X_i).\end{align}
This Markov chain  is satisfied with $ W_{1,i} =  ( M , U^{i-1},S^{i-1} )$ and $ W_{2,i} =  (Y_{i+1}^n ,Z_{i+1}^n ) $, for all $i \in \{1,\ldots,n\}$.
\end{lemma}

\begin{proof}[Lemma \ref{lemma2:ConverseCE2SIMarkov1}] 
The auxiliary random variables $ W_{1,i} =  ( M , U^{i-1},S^{i-1} )$ and $ W_{2,i} =  (Y_{i+1}^n ,Z_{i+1}^n ) $ satisfy the following equations for all $(u^n,s^n,z^n,w_1^n,w_2^n,x^n,y^n,v^n,m)$:
\begin{align}
&\PP(v_i | y_i ,z_i, w_{1,i} , w_{2,i}  , u_i, s_i, x_i) \nonumber \\
=&\PP(v_i | y_i ,z_i, m,u^{i-1}, s^{i-1}     ,  y^n_{i+1} ,z^n_{i+1} , u_i, s_i, x_i) \nonumber \\
=&   \sum_{x^{i-1}, y^{i-1}, z^{i-1}}\PP(v_i , x^{i-1} , y^{i-1}  , z^{i-1}  | y_i ,z_i, m, u^{i-1}, s^{i-1}     ,  y^n_{i+1} ,z^n_{i+1} , u_i, s_i, x_i )\\
=&   \sum_{x^{i-1}, y^{i-1}, z^{i-1}} \PP(z^{i-1}  | y_i ,z_i, m, u^{i-1}, s^{i-1}  ,  y^n_{i+1} ,z^n_{i+1} , u_i, s_i, x_i)\nonumber \\
&\PP(x^{i-1}  | y_i ,z_i, m, u^{i-1}, s^{i-1}  ,  y^n_{i+1} ,z^n_{i+1} , u_i, s_i, x_i, z^{i-1} )\nonumber \\
&\PP(  y^{i-1}  |y_i ,z_i, m,u^{i-1}, s^{i-1}  ,  y^n_{i+1} ,z^n_{i+1} , u_i, s_i, x_i ,  z^{i-1} ,x^{i-1} ) \nonumber\\
&\PP(v_i  |  y_i ,z_i, m,u^{i-1}, s^{i-1}  ,  y^n_{i+1} ,z^n_{i+1} , u_i, s_i, x_i , z^{i-1} , x^{i-1} , y^{i-1}  ) .  \end{align}
\mrb{We} can remove $(u_i , s_i , x_i)$, in the four conditional distributions
 \begin{align}
\PP(z^{i-1}  | y_i ,z_i, m,u^{i-1}, s^{i-1}  ,  y^n_{i+1} ,z^n_{i+1} , u_i, s_i, x_i) 
=&\PP(z^{i-1}  | u^{i-1}, s^{i-1}), \label{eq:lemma:CausalEnc2sI1}\\
\PP(x^{i-1}  | y_i ,z_i,m, u^{i-1}, s^{i-1}  ,  y^n_{i+1} ,z^n_{i+1} , u_i, s_i, x_i, z^{i-1} )  
=& \PP(x^{i-1}  |m,u^{i-1}, s^{i-1} ) , \label{eq:lemma:CausalEnc2sI2}\\
\PP(  y^{i-1}  |y_i ,z_i, m,u^{i-1}, s^{i-1}  ,  y^n_{i+1} ,z^n_{i+1} , u_i, s_i, x_i, z^{i-1}  , x^{i-1} )   =&  \PP(  y^{i-1}   | s^{i-1}  ,  x^{i-1}) ,  \label{eq:lemma:CausalEnc2sI3}\\
\PP(v_i  | y_i ,z_i, m,u^{i-1}, s^{i-1}  ,  y^n_{i+1} ,z^n_{i+1} , u_i, s_i, x_i, z^{i-1}  , x^{i-1}  , y^{i-1}  )  
=& \PP(v_i  | y_i ,z_i, y^n_{i+1} ,z^n_{i+1} , z^{i-1}  , y^{i-1}  ),  \nonumber\\
&\label{eq:lemma:CausalEnc2sI4}
\end{align}
where \eqref{eq:lemma:CausalEnc2sI1} comes from the i.i.d. property of the information source: $Z^{i-1} $ only depends on $(U^{i-1} ,S^{i-1} )$; \eqref{eq:lemma:CausalEnc2sI2} comes from the causal encoding that induces the Markov chain $X^{i-1} -\!\!\!\!\minuso\!\!\!\!- (M, U^{i-1}, S^{i-1})    -\!\!\!\!\minuso\!\!\!\!- (Y_i ,Z_i,  Y^n_{i+1} ,Z^n_{i+1} , X_i, Z^{i-1},  U_i, S_i )$; \eqref{eq:lemma:CausalEnc2sI3} comes from the memoryless property of the channel:  $Y^{i-1} $ only depends  on $(X^{i-1} , S^{i-1})$; \eqref{eq:lemma:CausalEnc2sI4} comes from the non-causal decoding  that induces the Markov chain $V_i -\!\!\!\!\minuso\!\!\!\!- (Y_i ,Z_i, Y^n_{i+1} ,Z^n_{i+1} , Z^{i-1}  , Y^{i-1} )    -\!\!\!\!\minuso\!\!\!\!- (  M,U^{i-1}, S^{i-1}  ,U_i, S_i, X_i,   X^{i-1}  )$. Hence, \mrb{for all $(u^n,s^n,z^n,x^n,y^n,v^n,m)$ we have}
\begin{align}
& \PP(v_i | y_i ,z_i, w_{1,i} , w_{2,i}  , u_i, s_i, x_i)\nonumber\\
=&   \sum_{x^{i-1}, y^{i-1}, z^{i-1}} \PP( v_i , x^{i-1}  , y^{i-1} , z^{i-1}   | y_i ,z_i, m,u^{i-1}, s^{i-1}  ,  y^n_{i+1} ,z^n_{i+1} )\\
=& \PP(v_i | y_i ,z_i ,m,u^{i-1}  , s^{i-1}  ,  y^n_{i+1} ,z^n_{i+1} )  \\
=& \PP(v_i | y_i ,z_i ,w_{1,i} , w_{2,i}  ).
\end{align}
The above equation corresponds to the Markov chain $V_i  -\!\!\!\!\minuso\!\!\!\!- ( Y_i ,Z_i, W_{1,i} , W_{2,i}   )  -\!\!\!\!\minuso\!\!\!\!- ( U_i , S_i, X_i)$ and it concludes the proof of Lemma \ref{lemma2:ConverseCE2SIMarkov1}.
\end{proof}


\begin{lemma}\label{lemma3:ConverseCE2SIMarkov1}
The causal encoding function and the memoryless property of the channel  induce the following Markov chain property
 \begin{align}
 Y_i  -\!\!\!\!\minuso\!\!\!\!-  (U_i , S_i ,M , U^{i-1} , S^{i-1}) -\!\!\!\!\minuso\!\!\!\!- (U_{i+1}^n , S_{i+1}^n , Y_{i+1}^n , Z_{i+1}^n ).
\end{align}
This  Markov chain  is satisfied for all $i \in \{1,\ldots,n\}$.
\end{lemma}

\begin{proof}[Lemma \ref{lemma3:ConverseCE2SIMarkov1}] 
\mrb{For all $(u^n,s^n,z^n,x^n,y^n,v^n,m)$,} we have
\begin{align}
&\PP( y_i  |  u_i , s_i ,m , u^{i-1} , s^{i-1} , u_{i+1}^n , s_{i+1}^n , y_{i+1}^n , z_{i+1}^n )  \nonumber \\
=&\sum_{x_i}  \PP(x_i ,  y_i  |  u_i , s_i ,m , u^{i-1} , s^{i-1} , u_{i+1}^n , s_{i+1}^n , y_{i+1}^n , z_{i+1}^n ) \\
=&\sum_{x_i}  \PP(x_i   |  u_i , s_i ,m , u^{i-1} , s^{i-1} , u_{i+1}^n , s_{i+1}^n , y_{i+1}^n , z_{i+1}^n )  \nonumber \\
& \PP(  y_i  | x_i , u_i , s_i ,m , u^{i-1} , s^{i-1} , u_{i+1}^n , s_{i+1}^n , y_{i+1}^n , z_{i+1}^n ) \label{eq:lemma3ConverseCE2sI1}  \\
=&\sum_{x_i}  \PP(x_i   |  u_i , s_i ,m , u^{i-1} , s^{i-1}  )   \PP(  y_i  | x_i , u_i , s_i ,m , u^{i-1} , s^{i-1} , u_{i+1}^n , s_{i+1}^n , y_{i+1}^n , z_{i+1}^n ) \label{eq:lemma3ConverseCE2sI2}  \\
=&\sum_{x_i}  \PP(x_i   |  u_i , s_i ,m , u^{i-1} , s^{i-1}  ) \cdot
\PP(  y_i  | x_i ,  s_i  ) \label{eq:lemma3ConverseCE2sI3}  \\
=&\sum_{x_i}  \PP(x_i ,  y_i  |  u_i , s_i ,m , u^{i-1} , s^{i-1}  ) =   \PP(  y_i  |  u_i , s_i ,m , u^{i-1} , s^{i-1}  ),
\end{align}
where \eqref{eq:lemma3ConverseCE2sI2} comes from the causal encoding function that induces the Markov chain $X_i -\!\!\!\!\minuso\!\!\!\!- (U_i , S_i ,M , U^{i-1} , S^{i-1}) -\!\!\!\!\minuso\!\!\!\!-  (U_{i+1}^n , S_{i+1}^n , Y_{i+1}^n , Z_{i+1}^n)$; \eqref{eq:lemma3ConverseCE2sI3} comes from the memoryless property of the channel that induces the Markov chain: $Y_i -\!\!\!\!\minuso\!\!\!\!-  (X_i ,  S_i) -\!\!\!\!\minuso\!\!\!\!- (U_i ,M , U^{i-1} , S^{i-1} , U_{i+1}^n , S_{i+1}^n , Y_{i+1}^n , Z_{i+1}^n )$. This concludes the proof of Lemma \ref{lemma3:ConverseCE2SIMarkov1}.
\end{proof}


\section{Converse proof of Theorem \ref{theo:LeakageCENoisyFeedback} }\label{sec:ConverseLeakageCENoisyFeedback}

The converse proof for information constraints \eqref{eq:theoremCENoisyFeedback2} and \eqref{eq:theoremCENoisyFeedback3} \mrb{follows from similar arguments as for the converse proof of} Theorem \ref{theo:LeakageCE}, in Appendix \ref{sec:ConverseLeakageCE}. We \mrb{now} prove the converse result for the information constraint \eqref{eq:theoremCENoisyFeedback1}. We consider the distribution $\QQ_{SXY_1Y_2V}$, \mrb{also denoted by $\QQ$,} and we introduce the random  event of error $E \in \{0,1\}$ defined \mrb{by}
\begin{align}
E = \Bigg\{
\begin{array}{lll}
0 \text{ if }& \big|\big|{Q}^n - \QQ  \big|\big|_{1} \leq \varepsilon\quad \Longleftrightarrow \quad (S^n,X^n,Y_1^n,Y_2^n,V^n) \in T_{\delta}(\QQ) ,\\
1 \text{ if }& \big|\big|{Q}^n - \QQ  \big|\big|_{1} > \varepsilon \quad \Longleftrightarrow \quad (S^n,X^n,Y_1^n,Y_2^n,V^n) \notin T_{\delta}(\QQ).
\end{array}
\Bigg.
\end{align}
Consider a sequence of code $c(n) \in \mc{C}$ that achieves the distribution $\QQ_{SXY_1Y_2V}$, \textit{i.e.} for which the probability of error  $\PP_{\sf{e}}(c) = \prob(E=1)$ goes to zero. We have
\begin{align}
n \textsf{R} \leq& \log_2 |\mc{M}| + n\varepsilon  \label{eq:ConvCausalEncodingNoisyF01}\\
=& H(M) + n\varepsilon  \label{eq:ConvCausalEncodingNoisyF02}\\
=&I(M;Y_1^n ) +  H(M|Y_1^n )   + n\varepsilon  \label{eq:ConvCausalEncodingNoisyF03}\\
\leq&I(M;Y_1^n  ) +  n2\varepsilon   \label{eq:ConvCausalEncodingNoisyF04}\\
=& \sum_{i=1}^n I(M;Y_{1,i}| Y_{1,i+1}^n  ) +  n\varepsilon   \label{eq:ConvCausalEncodingNoisyF05}\\
=& \sum_{i=1}^n I(M,S^{i-1} , Y_2^{i-1}  ;Y_{1,i} | Y_{1, i+1}^n  ) - \sum_{i=1}^n I(S^{i-1} , Y_2^{i-1}  ;Y_{1,i}|Y_{1,i+1}^n,M ) +  n2\varepsilon   \label{eq:ConvCausalEncodingNoisyF06}\\
=& \sum_{i=1}^n I(M,S^{i-1} , Y_2^{i-1}  ;Y_{1,i} | Y_{1, i+1}^n  ) -   \sum_{i=1}^n I(Y^n_{1,i+1}   ;   S_i,Y_{2,i} |   S^{i-1} , Y_2^{i-1}, M  ) +  n2\varepsilon   \label{eq:ConvCausalEncodingNoisyF07}\\
\leq& \sum_{i=1}^n I(M,S^{i-1} , Y_2^{i-1}, Y_{1, i+1}^n   ;Y_{1,i}   ) -   \sum_{i=1}^n I(Y^n_{1,i+1}   ;   S_i,Y_{2,i} |   S^{i-1} , Y_2^{i-1}, M  ) +  n2\varepsilon   \label{eq:ConvCausalEncodingNoisyF08}\\
=& \sum_{i=1}^n I(W_{1,i} , W_{2,i}    ;Y_{1,i}   ) -   \sum_{i=1}^n I(W_{2,i}  ;   S_i,Y_{2,i} |  W_{1,i}   ) +  n2\varepsilon   \label{eq:ConvCausalEncodingNoisyF09},
\end{align}
where \eqref{eq:ConvCausalEncodingNoisyF01}  comes from the definition of achievable rate $\textsf{R}$, stated in equation \eqref{eq:converse1ce};  \eqref{eq:ConvCausalEncodingNoisyF02}  comes from the uniform distribution of the random message $M$; \eqref{eq:ConvCausalEncodingNoisyF03}  comes from the definition of the mutual information; \eqref{eq:ConvCausalEncodingNoisyF04} comes from equation \eqref{eq:converse2ce} and Fano's inequality, see  \cite[pp. 19]{ElGamalKim(book)11};  \eqref{eq:ConvCausalEncodingNoisyF07} comes from Csisz\'{a}r Sum Identity, see \cite[pp. 25]{ElGamalKim(book)11}; \eqref{eq:ConvCausalEncodingNoisyF09} comes from the introduction of the auxiliary random variables $ W_{1,i} =  ( M,   S^{i-1} , Y_2^{i-1}  )$ and $ W_{2,i} =  Y^n_{1,i+1}$, that satisfy the properties corresponding to the  set of distributions $\Q_{\sf{f}}$, as proved in Lemma \ref{lemma:MarkovNF}.


\begin{lemma}\label{lemma:MarkovNF}
For all $i \in \{1,\ldots,n\}$, the auxiliary random variables $ W_{1,i} =  ( M,   S^{i-1} , Y_2^{i-1}  )$ and $ W_{2,i} =  Y^n_{1,i+1}$ satisfy following the properties corresponding to the  set of distributions $\Q_{\sf{f}}$
\begin{align}
& ( S_i) \text{ are independent of } W_{1,i} , \label{eq:ConvCENoisyFMarkov1} \\
& (Y_{1,i} , Y_{2,i})  -\!\!\!\!\minuso\!\!\!\!-   (X_i,S_i) -\!\!\!\!\minuso\!\!\!\!-  W_{1,i} , \label{eq:ConvCENoisyFMarkov2}  \\
& W_{2,i} -\!\!\!\!\minuso\!\!\!\!-   (S_i ,Y_{2,i},W_{1,i} ) -\!\!\!\!\minuso\!\!\!\!-   (X_i, Y_{1,i}) , \label{eq:ConvCENoisyFMarkov3}  \\
& V_i -\!\!\!\!\minuso\!\!\!\!-   (Y_{1,i},  W_{1,i},  W_{2,i}) -\!\!\!\!\minuso\!\!\!\!-  (X_i,S_i,Y_{2,i}). \label{eq:ConvCENoisyFMarkov4} 
\end{align}
\end{lemma}
\mrb{Then, \eqref{eq:ConvCausalEncodingNoisyF09} shows}
\begin{align}
n \textsf{R} \leq&   \sum_{i=1}^n I(W_{1,i} , W_{2,i}    ;Y_{1,i}   ) -   \sum_{i=1}^n I(W_{2,i}  ;   S_i,Y_{2,i} |  W_{1,i}   ) +  n2\varepsilon  \nonumber\displaybreak[0]\\
=&  n \bigg( I(  W_{1,T}  , W_{2,T}  ; Y_{1,T}  |T)   -    I( W_{2,T}   ;  S_T ,Y_{2,T}  |   W_{1,T} ,T)   + 2\varepsilon \bigg) \label{eq:ConvCausalEncodingNoisyF11} \displaybreak[0]\\
\leq&  n \bigg( I( T,  W_{1,T}  , W_{2,T}  ; Y_{1,T}  )   -    I( W_{2,T}   ;  S_T ,Y_{2,T}  |   W_{1,T} ,T)   + 2\varepsilon \bigg) \label{eq:ConvCausalEncodingNoisyF12} \displaybreak[0]\\
\leq&  n \max_{{\QQ}\in \Q_{\sf{f}} } \bigg( I(   W_{1}  , W_{2}  ; Y_{1,T}  )   -    I( W_{2}   ;  S_T ,Y_{2,T}  |   W_{1} )   + 2\varepsilon \bigg) \label{eq:ConvCausalEncodingNoisyF13} \displaybreak[0]\\
\leq&  n \max_{{\QQ}\in \Q_{\sf{f}} } \bigg( I(   W_{1}  , W_{2}  ; Y_{1,T}  | E=0)   -    I( W_{2}   ;  S_T ,Y_{2,T}  |   W_{1} ,E=0 )   + 3\varepsilon \bigg) \label{eq:ConvCausalEncodingNoisyF14} \displaybreak[0]\\
\leq&  n \max_{{\QQ}\in \Q_{\sf{f}} } \bigg( I(   W_{1}  , W_{2}  ; Y_{1} )   -    I( W_{2}   ;  S ,Y_{2}  |   W_{1}  )   + 4\varepsilon \bigg) \label{eq:ConvCausalEncodingNoisyF15} ,
\end{align}
where \eqref{eq:ConvCausalEncodingNoisyF11} comes from the introduction of the uniform random variable $T$ over $\{1 , \ldots, n\}$ and the introduction of the corresponding mean random variables $S_T$,  $W_{1,T}$,  $W_{2,T}$, $X_T$, $Y_{1,T}$,  $Y_{2,T}$, $V_T$; \eqref{eq:ConvCausalEncodingNoisyF12} and comes from the properties of the mutual information; \eqref{eq:ConvCausalEncodingNoisyF13} comes from the \mrb{identification of} $W_1$ with $(W_{1,T} , T)$,  $W_2$ with $W_{2,T}$ and taking the maximum over the distributions that belong to the set $\Q_{\sf{f}}$; \eqref{eq:ConvCausalEncodingNoisyF14} comes from the empirical coordination requirement, as stated in Lemma \ref{lemma:ErrorEventCoordination0} in Section \ref{lemma:CE}, since the sequences are not jointly typical with small error probability $\prob(E=1)$; \eqref{eq:ConvCausalEncodingNoisyF15} comes from Lemma \ref{lemma:3} in Appendix  \ref{lemma:CE}, that states that the distribution induced by the coding scheme $ \PP_{S_TX_TY_{1,T}Y_{2,T}V_T|E=0}$ is close to the target distribution $\QQ(s,x,y_1,y_2,v)$.  The \mrb{result of} \cite[Lemma 2.7, pp. 19]{CsiszarKorner(Book)11} concludes the converse proof of Theorem \ref{theo:LeakageCENoisyFeedback}.

\begin{proof}[Lemma \ref{lemma:MarkovNF}] 
Equation \eqref{eq:ConvCENoisyFMarkov1} comes from the i.i.d. property of the source $S$, the independence of $S$  with respect to the message $M$ and the causal encoding function, hence $S_i$ is independent of the past channel inputs $X^{i-1}$, hence $S_i$ is independent of   $Y_2^{i-1}$; \eqref{eq:ConvCENoisyFMarkov2} comes from the memoryless property of the channel, hence $(Y_{1,i} , Y_{2,i})$ is drawn with $(X_i,S_i)$; \eqref{eq:ConvCENoisyFMarkov3} comes from 
\begin{align}
&   \PP(x_i, y_{1,i} | s_i ,y_{2,i},w_{1,i}, w_{2,i}   )\\
=& \PP(x_i | s_i ,y_{2,i},m,   s^{i-1} , y_2^{i-1},y^n_{1,i+1}   ) \PP(y_{1,i} | x_i, s_i ,y_{2,i},m,   s^{i-1} , y_2^{i-1},y^n_{1,i+1}   ) \label{eq:markovNF2}\\
=& \PP(x_i | s_i ,y_{2,i},m,   s^{i-1} , y_2^{i-1}   ) \PP(y_{1,i} | x_i, s_i ,y_{2,i},m,   s^{i-1} , y_2^{i-1},y^n_{1,i+1}   ) \label{eq:markovNF3}\\
=& \PP(x_i | s_i ,y_{2,i},m,   s^{i-1} , y_2^{i-1}   ) \PP(y_{1,i} | x_i, s_i ,y_{2,i},m,   s^{i-1} , y_2^{i-1}
   ) \label{eq:markovNF4}\\
   =& \PP(x_i , y_{1,i} | s_i ,y_{2,i},w_{1,i}  )\label{eq:markovNF5} \qquad \forall (s^n,x^n,y_1^n,y_2^n,w_1^n,w_2^n).
\end{align}
Equation \eqref{eq:markovNF2} comes from the choice of the auxiliary random variables $ W_{1,i} =  ( M,   S^{i-1} , Y_2^{i-1}  )$ and $ W_{2,i} =  Y^n_{1,i+1}$; \eqref{eq:markovNF3} comes from the causal encoding that implies $X_i$ is a function of $(S_i ,M,   S^{i-1} , Y_2^{i-1})$ but not of $Y^n_{1,i+1}$; \eqref{eq:markovNF4} comes from the memoryless property of the channel that implies $Y_{1,i} $  is drawn depending on $( X_i, S_i ,Y_{2,i})$ and not on $Y^n_{1,i+1}$; \eqref{eq:markovNF5} concludes that the Markov chain  \eqref{eq:ConvCENoisyFMarkov3} holds.

Equation \eqref{eq:ConvCENoisyFMarkov4} comes from the following equations, for all $(u^n,s^n,z^n,x^n,y^n,v^n,m)$
\begin{align}
&   \PP(v_i |  y_{1,i},  w_{1,i},  w_{2,i} , x_i,s_i,y_{2,i})\\
=& \sum_{x^{i-1}, y_1^{i-1}}  \PP(v_i ,x^{i-1}, y_1^{i-1} |  y_{1,i},  m,   s^{i-1} , y_2^{i-1},  y^n_{1,i+1} , x_i,s_i,y_{2,i})\label{eq:markovLemmaNF2} \displaybreak[0]\\
=& \sum_{x^{i-1}, y_1^{i-1}}  \PP(x^{i-1} |  y_{1,i},  m,   s^{i-1} , y_2^{i-1},  y^n_{1,i+1} , x_i,s_i,y_{2,i}) \nonumber \displaybreak[0]\\
& \PP(y_1^{i-1} |  x^{i-1},y_{1,i},  m,   s^{i-1} , y_2^{i-1},  y^n_{1,i+1} , x_i,s_i,y_{2,i}) \nonumber \displaybreak[0]\\
 &\PP(v_i |  x^{i-1}, y_1^{i-1},y_{1,i},  m,   s^{i-1} , y_2^{i-1},  y^n_{1,i+1} , x_i,s_i,y_{2,i})\label{eq:markovLemmaNF3} \displaybreak[0]\\
=& \sum_{x^{i-1}, y_1^{i-1}}  \PP(x^{i-1} |  y_{1,i},  m,   s^{i-1} , y_2^{i-1},  y^n_{1,i+1} ) \label{eq:markovLemmaNF4}  \displaybreak[0]\\
&\PP(y_1^{i-1} |  x^{i-1},  y_{1,i},  m,   s^{i-1} , y_2^{i-1},  y^n_{1,i+1}) \label{eq:markovLemmaNF5} \displaybreak[0] \\
 &\PP(v_i  |  x^{i-1}, y_1^{i-1}, y_{1,i},  m,   s^{i-1} , y_2^{i-1},  y^n_{1,i+1} )\label{eq:markovLemmaNF6}\displaybreak[0] \\
 =& \PP(v_i |  y_{1,i},  w_{1,i},  w_{2,i}), \label{eq:markovLemmaNF7} 
\end{align}
where \eqref{eq:markovLemmaNF2} comes from the choice of the auxiliary random variables $ W_{1,i} =  ( M,   S^{i-1} , Y_2^{i-1}  )$ and $ W_{2,i} =  Y^n_{1,i+1}$; \eqref{eq:markovLemmaNF3} comes from the decomposition of the probability; \eqref{eq:markovLemmaNF4} comes from the causal encoding that implies $X^{i-1}$ is a function of $(M,   S^{i-1} , Y_2^{i-2})$ but not of $(X_i,S_i,Y_{2,i})$; \eqref{eq:markovLemmaNF5} comes from the memoryless property of the channel that implies $Y_1^{i-1} $  depends only on $( X^{i-1}, S^{i-1} , Y_2^{i-1} )$ and not on $(X_i,S_i,Y_{2,i})$; \eqref{eq:markovLemmaNF6} comes from the non-causal decoding that implies $V_i$  is a function of $(Y_1^{i-1}, Y_{1,i} , Y^n_{1,i+1} )$ but not of $(X_i,S_i,Y_{2,i})$; \eqref{eq:markovLemmaNF7} concludes that the Markov chain  \eqref{eq:ConvCENoisyFMarkov4} holds. This concludes the proof of Lemma \ref{lemma:MarkovNF}.
\end{proof}


\section{Converse proof of Theorem \ref{theo:LeakageSCE} }\label{sec:ConverseTheoSCE}

We consider that the triple of rate, information leakage and distribution $(\textsf{R},\textsf{E},\QQ_{SXYV})$ is achievable with a \mrb{code with strictly causal encoding}. We introduce the random  event of error $E \in \{0,1\}$ defined with respect to the achievable distribution $\QQ_{SXYV}$, \mrb{also denoted by $\QQ$, by}
\begin{eqnarray}
E = \Bigg\{
\begin{array}{lll}
0 &\text{ if }& \big|\big|{Q}^n - \QQ  \big|\big|_{\sf{tv}} \leq \varepsilon\quad \Longleftrightarrow \quad (S^n,X^n,Y^n,V^n) \in T_{\delta}(\QQ) ,\\
1 &\text{ if }& \big|\big|{Q}^n - \QQ  \big|\big|_{\sf{tv}} > \varepsilon \quad \Longleftrightarrow \quad (S^n,X^n,Y^n,V^n) \notin T_{\delta}(\QQ).
\end{array}
\Bigg.
\end{eqnarray}
The event $E=1$ occurs if the sequences $(S^n,X^n,Y^n,V^n)\notin T_{\delta}(\QQ)$ for the target distribution $\QQ$. By definition \ref{def:CodeLeakageSCE}, for all $\varepsilon>0$, there exists $\bar{n}\in \N^{\star}$ such that for all $n\geq \bar{n}$, there exists a code $c^{\star} \in \C(n,\mc{M})$ that satisfies
\begin{eqnarray}
\frac{\log_2 |\mc{M}|}{n}  &\geq& \textsf{R} - \varepsilon, \label{eq:converse1sce}  \\
 \bigg| \mc{L}_{\textsf{e}}(c^{\star})  - \textsf{E} \bigg|  &=& \bigg|  \frac{1}{n} I(S^n;Y^n)  - \textsf{E} \bigg| \leq  \varepsilon,   \label{eq:converse3sce} \\
\PP_{\textsf{e}}(c^{\star})  &=& \prob\bigg( M \neq \hat{M} \bigg) + \prob\bigg(\Big|\Big|Q^n - \QQ \Big|\Big|_{\sf{tv}}> \varepsilon\bigg) \leq \varepsilon. \label{eq:converse2sce}  
\end{eqnarray}
We introduce the auxiliary random variables $ W_{2,i} =  ( M , S^{i-1} ,Y_{i+1}^n)$ that satisfy the Markov chains of the set of distribution $\Q_{\sf{sc}}$, for all $i \in \{1,\ldots,n\}$:
\begin{eqnarray}
&&S_i \text{ independent of } X_i, \label{eq:ConverseSCEMarkov1} \\
&&Y_i -\!\!\!\!\minuso\!\!\!\!- (X_i , S_i ) -\!\!\!\!\minuso\!\!\!\!-  W_{2,i}  ,  \label{eq:ConverseSCEMarkov3}  \\
&&V_i -\!\!\!\!\minuso\!\!\!\!- (Y_i,  X_i,W_{2,i} ) -\!\!\!\!\minuso\!\!\!\!-  S_i  , \label{eq:ConverseSCEMarkov4} 
\end{eqnarray}
where \eqref{eq:ConverseSCEMarkov1} \mrb{follows} from the strictly causal encoding function; \eqref{eq:ConverseSCEMarkov3} \mrb{follows} from the memoryless property of the channel $ \mc{T}_{Y|XS}$; \eqref{eq:ConverseSCEMarkov4} \mrb{is obtained with a} slight modification of Lemma \ref{lemma2:ConverseCEMarkov1} in which only the random variable $S_i$ is removed from~\eqref{eq:lemma:CausalEnc2}, \eqref{eq:lemma:CausalEnc3}, and \eqref{eq:lemma:CausalEnc4}. It is a direct consequence of the strictly causal encoding function, the non-causal decoding function and the memoryless property of the channel $ \mc{T}_{Y|XS}$.

We introduce the random variable $T$ that is uniformly distributed over the indices $\{1 , \ldots, n\}$ and the corresponding mean random variables  $W_{2,T}$, $S_T$, $X_T$, $Y_T$, $V_T$. The  auxiliary random variable $W_2=(W_{2,T} , T)$   belongs to the set of distributions $\Q_{\sf{sc}}$  and satisfies the information constraints of Theorem \ref{theo:LeakageSCE}
\begin{align}
 I(S;X,W_2,Y) \leq \textsf{E}&\leq H(S),\label{eq:theoremSCEconverse2bis}\\
 \textsf{R} + \textsf{E} &\leq I(X,S; Y). \label{eq:theoremSCEconverse3bis}
\end{align}

 \textbf{First Constraint:}
\begin{align}
n \textsf{E} \geq & I(S^n ; Y^n) -  n\varepsilon\label{eq:SecondConverse1}\displaybreak[0]\\
= & I(S^n;Y^n , M)  -  I(S^n ; M | Y^n) -  n\varepsilon\label{eq:SecondConverse2}\displaybreak[0]\\
\geq & n H(S)  - H(S^n | Y^n , M)  -  H(M | Y^n)-  n\varepsilon \label{eq:SecondConverse3}\displaybreak[0]\\
\geq & n H(S)  - H(S^n | Y^n , M)  -  n2\varepsilon \label{eq:SecondConverse4}\\
= & n H(S)  -  \sum_{i = 1}^n H(S_i | Y^n , M, S^{i-1} )  -  n2\varepsilon \label{eq:SecondConverse5}\displaybreak[0]\\
\geq & n H(S)  -  \sum_{i = 1}^n H(S_i | Y_{i+1}^n , Y_i , M, S^{i-1} )  -  n2\varepsilon \label{eq:SecondConverse6}\displaybreak[0]\\
= & n H(S)  -  \sum_{i = 1}^n H(S_i | Y_{i+1}^n ,Y_i, M, S^{i-1},X_i )  -  n2\varepsilon \label{eq:SecondConverse7}\displaybreak[0]\\
= & n H(S)  -  \sum_{i = 1}^n H(S_i | W_{2,i} ,X_i,Y_i )  -  n2\varepsilon \label{eq:SecondConverse8} \displaybreak[0]  \\
= & n H(S)  -  n  H(S_T |W_{2,T} , X_T , Y_T , T )  -  n2\varepsilon \label{eq:SecondConverse9}\displaybreak[0] \\
= & n H(S)  -  n  H(S_T |W_{2} , X_T , Y_T  )  -  n2\varepsilon \label{eq:SecondConverse10} \displaybreak[0]\\
\geq & n H(S)  -  n  H(S_T |W_{2} , X_T , Y_T ,E=0 )  -  n3\varepsilon \label{eq:SecondConverse11}\displaybreak[0] \\
\geq & n H(S)  -  n  H(S |W_{2} , X , Y )  -  n4\varepsilon \label{eq:SecondConverse12} \displaybreak[0]\\
= & n  \bigg(  I(S ; W,X,Y)  -  4\varepsilon \bigg),  \label{eq:SecondConverse13}
\end{align}
where  \eqref{eq:SecondConverse1} comes from the definition of achievable information leakage rate $\textsf{E}$, stated in equation \eqref{eq:converse3sce}; \eqref{eq:SecondConverse4} comes from equation \eqref{eq:converse2sce} and  Fano's inequality, see \cite[pp. 19]{ElGamalKim(book)11}; \eqref{eq:SecondConverse7} comes from the strictly causal encoding $X_i = f_i(M , S^{i-1})$ that  implies $I(S_i  ; X_i | Y_{i+1}^n ,Y_i, M, S^{i-1} ) =0$,  for all $i\in \{1,\ldots,n\}$; \eqref{eq:SecondConverse8}  comes from the introduction of the auxiliary random variable $ W_{2,i} =  ( M,S^{i-1} , Y_{i+1}^n )$, for all $i \in \{1,\ldots,n\}$; \eqref{eq:SecondConverse9}  comes from the introduction of the uniform random variable $T$ over $\{1 , \ldots, n\}$ and the introduction of the corresponding mean random variables  $S_T$,  $W_{2,T}$, $X_T$, $Y_T$; \eqref{eq:SecondConverse10} comes from identifying $W_2=(W_{2,T} , T)$; \eqref{eq:SecondConverse11}  comes from the empirical coordination requirement as stated in Lemma \ref{lemma:ErrorEventCoordination01}, \mrb{since} the sequences of symbols $(S^n,X^n,Y^n,V^n)$ are not jointly typical with small error probability $\prob(E=1)$; \eqref{eq:SecondConverse12} comes from Lemma \ref{lemma:3}. The sequences of symbols $(S^n,X^n,Y^n,V^n)$ are jointly typical, hence the distribution of the mean random variables $ \PP_{S_TX_TY_TV_T|E=0}$ is close to the target distribution $\QQ_{SXYV}$.  The \mrb{result of} \cite[Lemma 2.7, pp. 19]{CsiszarKorner(Book)11} concludes.

\textbf{Second Constraint:}
 From \eqref{eq:ThirdConverse3ce}, we have
\begin{align}
n \textsf{E} \leq &  n  \bigg( H(S) + \varepsilon \bigg). 
\end{align}

 \textbf{Third Constraint:}
\begin{eqnarray}
n \bigg( \textsf{E} +  \textsf{R} \bigg) 
&\leq & I(S^n ; Y^n)  + H(M) +  n2\varepsilon \label{eq:FourthConverse1}\\
&= &  I(S^n ; Y^n)  + I(M ; Y^n ) +  H(M |Y^n) +  n2\varepsilon \label{eq:FourthConverse2}\\
&\leq &  I(S^n ; Y^n)  + I(M ; Y^n ) +  n3\varepsilon \label{eq:FourthConverse3}\\
&\leq &  I(S^n ; Y^n)  + I(M ; Y^n | S^n) +  n3\varepsilon \label{eq:FourthConverse4}\\
&= &  I(S^n ,M ; Y^n)  +  n3\varepsilon \label{eq:FourthConverse5}\\
&\leq &  I(S^n ,X^n ; Y^n)  +  n3\varepsilon \label{eq:FourthConverse6}\\
&\leq &  n \bigg( I(S ,X ; Y)  +  3\varepsilon \bigg), \label{eq:FourthConverse7}
\end{eqnarray}
where  \eqref{eq:FourthConverse1} comes from the definition of achievable rate and information  leakage $( \textsf{R} ,  \textsf{E})$, stated in equations \eqref{eq:converse1sce} and  \eqref{eq:converse3sce}; \eqref{eq:FourthConverse3} comes from equation \eqref{eq:converse2sce} and  Fano's inequality, see \cite[pp. 19]{ElGamalKim(book)11}; \eqref{eq:FourthConverse4}  comes from the independence between the message $M$ and the channel states $S^n$, hence $I(M ; Y^n ) \leq I(M ; Y^n ,S^n ) = I(M ; Y^n |S^n ) $; \eqref{eq:FourthConverse6}  comes from the Markov chain $Y^n  -\!\!\!\!\minuso\!\!\!\!- ( X^n,S^n  )   -\!\!\!\!\minuso\!\!\!\!- M$, induced by the channel; \eqref{eq:FourthConverse7}  comes from the memoryless property of the channel that implies: $ H( Y^n | S^n ,X^n) = n H( Y |X,S)$.

 \textbf{Conclusion:} 
If the triple of rate, information leakage and distribution $(\textsf{R},\textsf{E},\mrb{\QQ_{SXYV}})$ is achievable with \mrb{a code with} strictly causal encoding, then the following equations are satisfied  for all $\varepsilon>0$:
\begin{eqnarray}
 I(S;X,Y,W_2) - 4 \varepsilon \leq &\textsf{E}&\leq H(S) + \varepsilon,\\
&\textsf{R} + \textsf{E} &\leq I(X,S; Y)+ 3 \varepsilon.
\end{eqnarray}
This corresponds to  \eqref{eq:theoremSCEachie2} and  \eqref{eq:theoremSCEachie3} and \mrb{this} concludes the converse proof of Theorem \ref{theo:LeakageSCE}.

\end{document}